\documentclass[12pt, reqno]{amsart}
\usepackage{mathtools}
\usepackage{amssymb,amsmath,dsfont}
\usepackage[table,xcdraw]{xcolor}
\usepackage{pgfplots}
\pgfplotsset{compat=1.14}
\usepackage{enumitem}
\usepackage{calc}
\usepackage{hyperref}
\usepackage{comment}

\usepackage[margin=1in]{geometry}

\usepackage{booktabs}
\usepackage{caption} 
\newtheorem{thm}{Theorem}[section]
\newtheorem{theorem}[thm]{Theorem}
\newtheorem{lemma}[thm]{Lemma}
\newtheorem{lem}[thm]{Lemma}
\newtheorem{prop}[thm]{Proposition}
\newtheorem{fact}[thm]{Fact}
\newtheorem{cor}[thm]{Corollary}
\newtheorem{defn}[thm]{Definition}
\newtheorem{claim}[thm]{Claim}
\newtheorem{remark}[thm]{Remark}
\theoremstyle{plain}

\newcommand{\Z}{\mathbb{Z}}
\newcommand{\ep}{\varepsilon}
\newcommand{\R}{\mathbb{R}}
\newcommand{\N}{\mathbb N}
\newcommand{\white}{\text{white}}
\DeclareMathOperator{\cov}{Cov}
\DeclareMathOperator{\var}{Var}
\renewcommand{\P}{\mathbb{P}}
\newcommand{\E}{\mathbb{E}}

\newcommand{\eqd}{\,{\buildrel d \over =}\,}
\DeclareMathOperator{\GE}{GE}

\DeclareMathOperator{\std}{Std}
\DeclareMathOperator{\med}{Med}
\DeclareMathOperator{\supp}{Supp}
\DeclareMathOperator{\DE}{DE}

\DeclareMathOperator{\law}{law}

\DeclareMathOperator{\Dob}{Dob}
\DeclareMathOperator{\sign}{sign}

\newtheorem*{assumption*}{\assumptionnumber}
\providecommand{\assumptionnumber}{}
\makeatletter
\newenvironment{assumption}[1]
 {%
  \renewcommand{\assumptionnumber}{Assumption #1}%
  \begin{assumption*}%
  \protected@edef\@currentlabel{#1}%
 }
 {%
  \end{assumption*}
 }
\makeatother

\newtheorem*{disorder*}{\assumptionnumber}
\providecommand{\assumptionnumber}{}


\title{Minimal surfaces in random environment}

\author{Barbara Dembin}
\address{Barbara Dembin\hfill\break
    IRMA, CNRS et Université de Strasbourg, Strasbourg,France}
\email{barbara.dembin@math.unistra.fr}
\author{Dor Elboim}
\address{Dor Elboim\hfill\break
    Department of Mathematics,
    Stanford University,
    California, United States.}
\email{dorelboim@gmail.com}
\author{Daniel Hadas}
\address{Daniel Hadas\hfill\break School of Mathematical Sciences, Tel Aviv University, Tel Aviv, Israel.}
\email{danielhadas1@mail.tau.ac.il}
\author{Ron Peled}
\address{Ron Peled\hfill\break Department of Mathematics, University of Maryland, College Park, United States.\hfill\break
School of Mathematical Sciences, Tel Aviv University, Tel Aviv, Israel.}
\email{peledron@tauex.tau.ac.il}
\date{\today}

\begin{document}

\begin{abstract}
A minimal surface in a random environment (MSRE) is a surface which minimizes the sum of its elastic energy and its environment potential energy, subject to prescribed boundary conditions. Apart from their intrinsic interest, such surfaces are further motivated by connections with disordered spin systems and first-passage percolation models. We wish to study the geometry of $d$-dimensional minimal surfaces in a $(d+n)$-dimensional random environment. Specializing to a model that we term harmonic MSRE, in an ``independent'' random environment, we rigorously establish bounds on the geometric and energetic fluctuations of the minimal surface, as well as versions of the scaling relation $\chi=2\xi+d-2$ that ties together these two types of fluctuations. In particular, we prove, for all values of $n$, that the surfaces are delocalized in dimensions $d\le 4$ and localized in dimensions $d\ge 5$. Moreover, the surface delocalizes with power-law fluctuations when $d\le 3$ and sub-power-law fluctuations when $d=4$. Our localization results apply also to harmonic minimal surfaces in a periodic random environment.
\end{abstract}

\maketitle

\begin{figure}[htp]
    \centering
    \includegraphics[width=17cm]{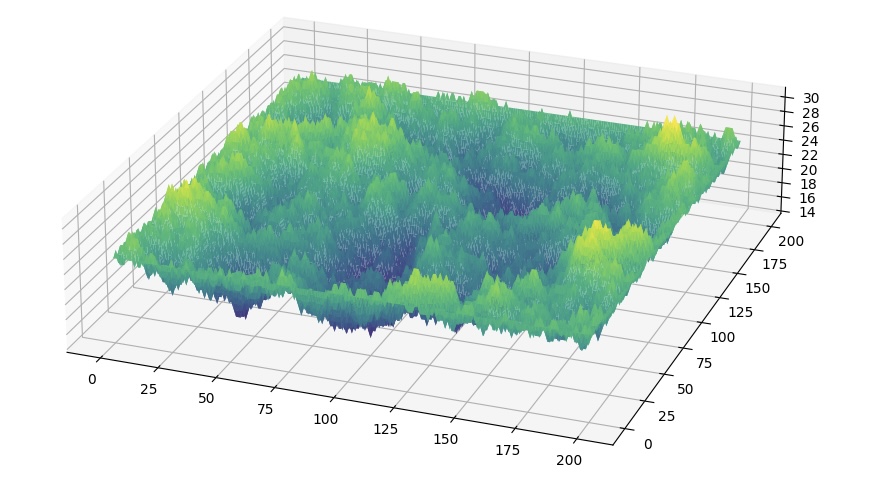}
    \caption{A minimal surface in independent disorder with $d=2$, $n=1$.}
    \label{fig:d=2}
\end{figure}

\setcounter{tocdepth}{2}
\newpage
\tableofcontents

\section{Introduction}

The mathematical theory of minimal surfaces
has been a topic of enduring interest since the 18th century works of Euler and Lagrange (see~\cite{colding2011course, meeks2012survey, de2022regularity}). In this work we embark on a mathematical exploration of the properties of minimal surfaces \emph{in a random environment}, i.e., surfaces placed in an inhomogeneous, random media, which minimize the sum of an internal (elastic) energy and an environment potential energy subject to prescribed boundary conditions. Apart from its intrinsic interest, the study of such surfaces is further motivated by connections with disordered spin systems and first-passage percolation models.

We focus on a specific model, that we term \emph{harmonic} minimal surfaces in random environment, having special properties which facilitate its analysis. 
The next sections describe the model and our results. Further discussion of the mathematics and physics background is presented in Section~\ref{sec:background} and Section~\ref{sec:disorder types}.

\subsection{The model}\label{subsection : the model}

Let $d,n\ge 1$ be integers. We model a $d$-dimensional minimal surface in a $(d+n)$-dimensional random environment by the ground configurations (functions of minimal energy) of the following Hamiltonian. The surface is modeled by a function $\varphi:\Z^d\to\R^n$, defined on the cubic lattice $\Z^d$ and having $n$ components. Given an \emph{environment} $\eta:\Z^d\times\R^n\to(-\infty,\infty]$, later taken to be random and termed the \emph{disorder}, and an \emph{environment strength} $\lambda>0$, the \emph{formal Hamiltonian} for $\varphi$ is
\begin{equation}\label{eq:formal Hamiltonian}
    H^{\eta,\lambda}(\varphi):=\frac{1}{2}\sum_{u\sim v}\|\varphi_u - \varphi_v\|^2 + \lambda \sum_{v} \eta_{v,\varphi_v},
\end{equation}
where $\|\cdot\|$ is the Euclidean norm in $\R^n$ and $u\sim v$ indicates that $u,v\in\Z^d$ are adjacent.

Our goal is to study the minimizers of $H^{\eta,\lambda}$ in finite domains with prescribed boundary values.
Given a finite $\Lambda\subset\Z^d$ and a function $\tau:\Z^d\to\R^n$, the \emph{finite-volume Hamiltonian} in~$\Lambda$ is given by
\begin{equation}\label{eq:finite volume Hamiltonian}
H^{\eta,\lambda,\Lambda}(\varphi):=\frac{1}{2}\sum_{\substack{u\sim v\\\{u,v\}\cap\Lambda\neq\emptyset}}\|\varphi_u - \varphi_v\|^2 + \lambda \sum_{v\in\Lambda} \eta_{v,\varphi_v},
\end{equation}
and the \emph{configuration space with boundary value $\tau$ outside $\Lambda$} is given by
\begin{equation}
   \Omega^{\Lambda,\tau} := \{\varphi:\Z^d\to\R^n\colon \varphi_v=\tau_v\text{ for $v\in\Z^d\setminus\Lambda$}\}.
\end{equation}
We write $\varphi^{\eta,\lambda,\Lambda,\tau}$ for the \emph{ground configuration} of the finite-volume model, i.e., for the $\varphi\in \Omega^{\Lambda,\tau}$ which minimizes $H^{\eta,\lambda,\Lambda}$ (we assume that minimizers exist in~\ref{as:exiuni} below). In the case of multiple minimizers, we let $\varphi^{\eta,\lambda,\Lambda,\tau}$ denote the first in some lexicographic order (defined from a fixed total order on $\Z^d$ and a lexicographic order on $\R^n$). We let
\begin{equation}
   \GE^{\eta,\lambda,\Lambda,\tau}:=H^{\eta,\lambda,\Lambda}(\varphi^{\eta,\lambda,\Lambda,\tau})
\end{equation}
be the \emph{ground energy}. To lighten notation, we shall often omit the superscript $\tau$ when $\tau\equiv 0$.

We call this model \emph{harmonic minimal surfaces in random environment} to highlight the important role that harmonic functions on the lattice play in its analysis (see Section~\ref{sec:main identity and consequences})\footnote{Similarly, the lattice Gaussian free field, for which~\eqref{eq:formal Hamiltonian} is a disordered version, is sometimes termed the harmonic crystal (e.g., in~\cite{brascamp1975statistical}).}.

\subsection{The disorder}
The properties of the minimal surfaces of our model depend crucially on the choice of  disorder $\eta$. Our main, but not exclusive, focus is on ``independent disorders'' which informally means that the $(\eta_{v,\cdot})_{v\in\Z^d}$ are independent processes and that, for each $v$, the process $t\mapsto\eta_{v,t}$ is \emph{stationary with finite-range correlations}. This choice is related to first-passage percolation and disordered (random-bond) Ising ferromagnet models, as discussed in Section~\ref{sec:disorder types} where we also discuss other ``disorder universality classes''. We note that our localization results apply also to the class of ``periodic disorders'' (Section~\ref{sec:periodic disorder}). 

Figure~\ref{fig:d=2}, Figure~\ref{fig:n=1} and Figure~\ref{fig:n=2} present simulations of discrete minimal surfaces in an independent disorder having $d+n\le 3$ (directed first-passage percolation and an analogous surface model). These are believed to behave similarly to the model~\ref{eq:formal Hamiltonian} (see Section~\ref{sec:independent disorder}).

\begin{figure}[tp]
    \centering
    \includegraphics[width=17cm]{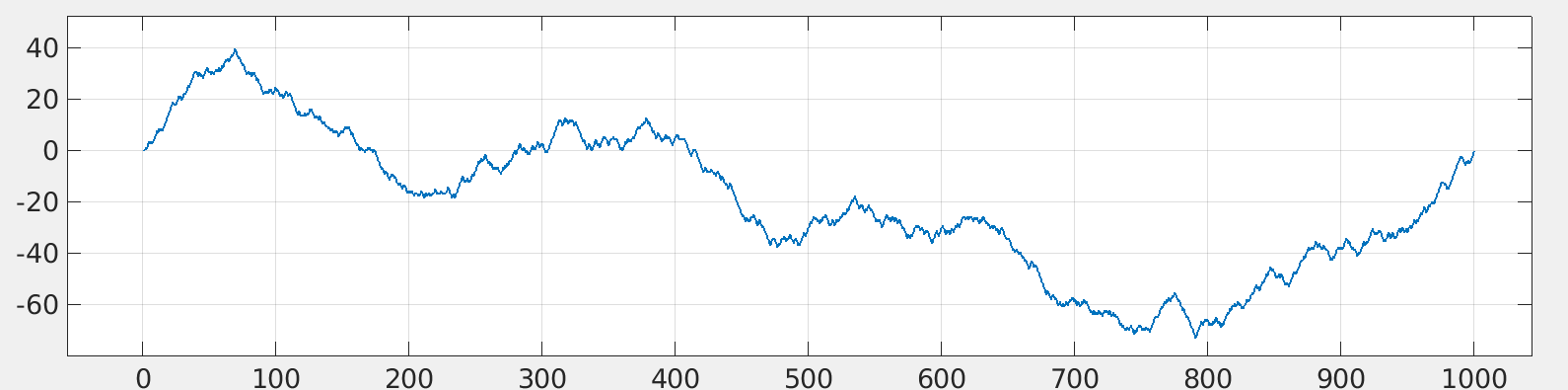}
    \caption{Simulation of directed first-passage percolation with $d=n=1$.
    }
    \label{fig:n=1}
\end{figure}

\begin{figure}[tp]
    \centering
    \includegraphics[width=17cm]{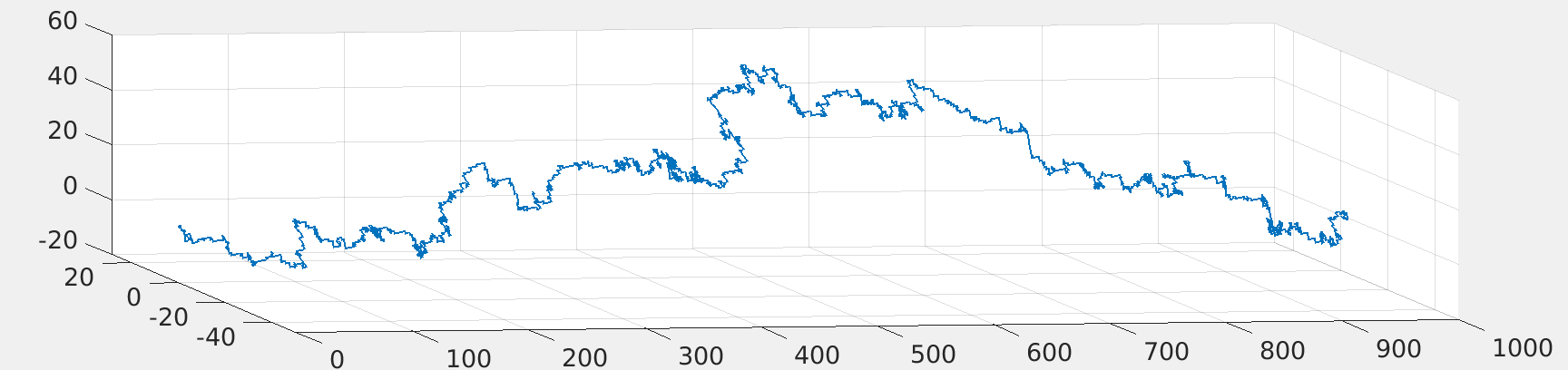}
    \caption{Simulation of directed first-passage percolation with $d=1$, $n=2$.
    }
    \label{fig:n=2}
\end{figure}

We proceed to describe a specific class of disorders, termed $\eta^\white$, to which all of our results apply. Following that, as many of our results (especially on localization and the scaling relation) apply significantly more generally, we continue with a list of general assumptions on the disorder which we will use in our proofs (all these assumptions are satisfied for $\eta^\white$).

\subsubsection{Smoothed white noise disorder $\eta^\white$.}\label{sec:whitenoise}
Our main disorder example consists of a smoothed version of white noise, as we now define. Recall that a white noise $W$ on $\R^n$ is a centered Gaussian process, indexed by (real-valued) functions $f\in L^2(\R^n)$, with covariance given by $\cov(W(f),W(g))=\int f(t)\cdot g(t)dt$ (see, e.g.,~\cite[Section 1.4]{le2016brownian}).

Let $b\in L^2(\R^n)$ be a function (termed the ``bump function'') satisfying:
\begin{enumerate}
\item \label{as:lipschitz} $b$ is a Lipschitz function: $\exists\mathcal{L}>0$ such that $|b(t)-b(s)|\le \mathcal{L}\|t-s\|$ for $t,s\in\R^n$,
\item $b\ge 0$ everywhere and $b(t)=0$ when $\|t\|\ge 1$,
\item \label{as:standard_1} $\int b(t)^2\,dt=1$.
\end{enumerate}
Let $(W_v)_{v\in\Z^d}$ be independent white noises on $\R^n$. For each $v\in\Z^d$, define $t\mapsto \eta_{v,t}$ as the convolution of $W_v$ with $b$, i.e.,
\begin{equation}\label{eq:eta white definition}
   \eta_{v,t} := W_v(b(\cdot-t)).
\end{equation}
In the sequel we denote by $\eta^\white$ the (sample path) continuous  version of the convolution~\eqref{eq:eta white definition} with any choice of $b$ satisfying the above properties (the precise choice of $b$ is immaterial to our analysis); the continuous version exists by (a generalization of) Kolmogorov's theorem (see, e.g., \cite[Theorem 4.23]{Kallenberg}), whose hypothesis follows from the fact that for each $v$ and every $t,s\in \mathbb R^n$, $\eta _{v,t}-\eta _{v,s}$ has the centered Gaussian distribution with standard deviation is at most a constant times $\mathcal L \|t-s\|^2$. Our results often have constants depending on the choice of $\eta$; when applying such results to $\eta^\white$, the constants only depend on $n$ and the (minimal) value of the Lipschitz constant $\mathcal{L}$ of the bump function~$b$.

\subsubsection{Assumptions on the disorder}\label{sec:disorder assumptions}
We now describe general assumptions on $\eta$ that will be used in our proofs. These are all verified for $\eta^\white$ in Section~\ref{sec:verifying assumptions}. The general setup highlights the important features of $\eta$ and allows the results to apply to a wide variety of disorders.

\smallskip
Our first assumption, used throughout, ensures the existence of configurations with minimal energy. It implicitly requires the disorder to be defined with a sufficiently rich sigma algebra to include the described event (our main disorder examples satisfy this as they either have continuous sample paths or are simple transformations of point processes; see Section~\ref{sec:disorder types}).
\begin{assumption}{(ExiMin)}\label{as:exiuni}For every $\lambda>0$, finite $\Lambda\subset\Z^d$ and every $\tau:\Z^d\to\R^n$, almost surely, $H^{\eta,\lambda,\Lambda,\tau}$ attains a minimum on every closed subset of $\Omega^{\Lambda,\tau}$ (the minimum may be $+\infty$ but cannot be $-\infty$). Moreover, the global minimum satisfies $\GE^{\eta,\lambda,\Lambda,\tau}<\infty$ almost surely.
\end{assumption}
A major role in our analysis is played by the following stationarity assumption, stating that the distribution of the disorder $\eta$ is invariant under shifts.
\begin{assumption}{(Stat)}\label{as:stat} For every $s:\Z^d\to\R^n$, the \emph{shifted disorder} $\eta^s:\Z^d\times\R^n\to(-\infty,\infty]$ defined by
\begin{equation}\label{eq:eta s def}
    \eta^s_{v,t}:=\eta_{v,t-s_v}
\end{equation}
has the same distribution as $\eta$.
\end{assumption}

It is also natural to assume that the disorder processes are independent between $v\in\Z^d$ and have finite-range correlations in $\R^n$. This is made explicit in the following assumption.

\begin{assumption}{(Indep)}\label{as:indep} The processes $(\eta _{v,\cdot})_{v\in\Z^d}$ are independent. Additionally, the restrictions $\eta |_{\{v\}\times A}$ and $\eta |_{\{v\}\times B}$ are independent for each $v\in\Z^d$ and each $A,B\subset\R^n$ satisfying $\inf_{t\in A, s\in B}\|t-s\|\ge 2$.
\end{assumption}
Note that if~\ref{as:indep} is assumed then~\ref{as:stat} follows simply from requiring that for each $v\in\Z^d$, the process $t\mapsto\eta_{v,t}$ is stationary (under translations in $\R^n$).

\begin{remark}\label{rem:wide variety of disorders}
    The above assumptions turn out to already suffice for our scaling relation results (Section~\ref{sec:scaling relation}), and yet allow a very wide variety of disorders: 
For instance, they allow the distributions of the processes $(\eta_{v,\cdot})$ to vary arbitrarily between vertices $v$, possibly making some parts of $\Lambda$ have ``stronger'' disorder than other parts. Additionally, as no restrictions are made on the tail decay of the disorder values (as long as~\ref{as:exiuni} is satisfied), the assumptions allow disorders with heavy tails.

Note also that as these assumptions do not reference the environment strength $\lambda$, it may be absorbed into $\eta$ when no other assumptions are imposed.
\end{remark}

\smallskip
Our final assumption, used in our localization results, concerns the concentration of the ground energy. We assume that the variance of the energy grows at most like the domain size, with sub-Gaussian tails on this scale. Moreover, we require a conditional version of this statement, stating that when conditioning on some of the disorder processes, concentration holds with the domain size replaced by the number of remaining disorder processes.

For a subset $A\subset\Z^d$ we write $\eta_A$ for the collection of processes $(\eta_{v,\cdot})_{v\in A}$. We also use the notation $A^c:=\Z^d\setminus A$.
\begin{assumption}{(Conc)}\label{as:conc} 
There exist $K,\kappa>0$ such that for every $\lambda>0$, every finite $\Delta\subset\Lambda\subset\Z^d$, every $\tau:\Z^d\to\R^n$ and every $\rho>0$, we have $\E|\GE^{\eta,\lambda,\Lambda,\tau}|<\infty$ and, almost surely,
\begin{equation}\label{eq:tal}
    \P\left(\Big|\GE^{\eta,\lambda,\Lambda,\tau} - \E\big(\GE^{\eta,\lambda,\Lambda,\tau}\mid \eta_{\Delta^c}\big)\Big|\ge \rho\lambda\sqrt{|\Delta|}\Biggm|\eta_{\Delta^c}\right)\le K e^{-\kappa\rho^2}.
\end{equation}
\end{assumption}
This type of sub-Gaussian concentration is often derived from Talagrand's or Hoeffding's concentration inequalities (this route is taken for periodic disorders in Section~\ref{sec:periodic disorder}). For $\eta^\white$, we verify~\ref{as:conc} using the Gaussian concentration inequality (Lemma~\ref{lem:ass_conc}).

\smallskip
Our delocalization results rely on perturbations of the disorder for which the Gaussian structure of $\eta^\white$ is convenient. We refrain from phrasing general assumptions allowing such perturbations and limit our proofs of the delocalization results to $\eta^\white$.

\subsubsection{Effect of $\lambda$}\label{sec:effect of lambda} Changing the environment strength $\lambda$ is equivalent to rescaling the minimal surface, ground energy and environment. Indeed, the Hamiltonian~\eqref{eq:finite volume Hamiltonian} satisfies
\begin{equation}
    H^{\eta,\lambda,\Lambda}(\varphi) = \lambda\Bigg(\frac{1}{2}\sum_{\substack{u\sim v\\\{u,v\}\cap\Lambda\neq\emptyset}}\Big\|\frac{\varphi_u}{\sqrt{\lambda}} - \frac{\varphi_v}{\sqrt{\lambda}}\Big\|^2 + \sum_{v\in\Lambda} \eta_{v,\varphi_v}\Bigg)=\lambda H^{\eta^\lambda,1,\Lambda}\bigg(\frac{\varphi}{\sqrt{\lambda}}\bigg),
\end{equation}
where the environment $\eta^\lambda:\Z^d\times\R^n\to(-\infty,\infty]$ is given by
\begin{equation}
    \eta^{\lambda}_{v,t}:=\eta_{v,t\sqrt{\lambda}}.
\end{equation}
We note that when $\eta=\eta^\white$ then the rescaled disorder $\eta^\lambda$ is given by the formula~\eqref{eq:eta white definition} with the bump function $b$ replaced by the rescaled function $b^\lambda$ given by $b^\lambda(t):=\lambda^{n/4}b(\sqrt{\lambda}t)$.

\subsection{Main results}
We present results on the following topics: (i) the height fluctuations of the minimal surface, (ii) the scaling relation between the ground energy and height fluctuations, and (iii) lower bounds on the ground energy fluctuations. In each statement, we indicate the assumptions on the disorder under which it is proved, or write $\eta=\eta^\white$ if it is established only for that specific example.

Throughout we let \begin{equation}\label{eq:Lambda L def}
    \Lambda_L:=\{-L,-L+1,\ldots, L\}^d
\end{equation}
for an integer $L\ge 0$. We denote by
\begin{equation}
    d_\infty(v,\Lambda_L^c):=\min_{w\in\Z^d\setminus\Lambda_L}\|v-w\|_\infty
\end{equation}
the $\ell_\infty$ distance of $v$ from the outer boundary of $\Lambda_L$.
Lastly, for a random variable $X$, we write $\med(X)$ for an arbitrary median of $X$ and $\std(X)$ for the standard deviation of $X$ (which is defined, possibly as $+\infty$, when $\E|X|<\infty$).

\begin{table}[]
\begin{tabular}{@{}|lccc|@{}}
\toprule
\rowcolor[HTML]{FFCCC9}   
\multicolumn{4}{|c|}{\cellcolor[HTML]{96C5FF}Height fluctuations for $n=1$}  
\\ \midrule
\multicolumn{1}{|l|}{}                                & \multicolumn{1}{c|}{Lower bound}                       & \multicolumn{1}{c|}{Predicted}                                & Upper bound \\ \midrule
\rowcolor[HTML]{EFEFEF} 
\multicolumn{1}{|l|}{\cellcolor[HTML]{EFEFEF}$d=1$}   & \multicolumn{1}{c|}{\cellcolor[HTML]{EFEFEF}$L^{0.6}$} & \multicolumn{1}{c|}{\cellcolor[HTML]{EFEFEF}$L^{2/3}$}        & $L^{0.75}$  \\ \midrule
\multicolumn{1}{|l|}{$d=2$}                           & \multicolumn{1}{c|}{$L^{0.4}$}                         & \multicolumn{1}{c|}{$L^{0.41\pm0.01}$}                         & $L^{0.5}$   \\ \midrule
\rowcolor[HTML]{EFEFEF} 
\multicolumn{1}{|l|}{\cellcolor[HTML]{EFEFEF}$d=3$}   & \multicolumn{1}{c|}{\cellcolor[HTML]{EFEFEF}$L^{0.2}$} & \multicolumn{1}{c|}{\cellcolor[HTML]{EFEFEF}$L^{0.22\pm0.01}$} & $L^{0.25}$  \\ \midrule
\multicolumn{1}{|l|}{$d=4$}                           & \multicolumn{1}{c|}{$(\log \log L)^{0.2}$}                                 & \multicolumn{1}{c|}{$(\log L)^{0.2083\ldots}$}                  & $\log L$    \\ \midrule
\rowcolor[HTML]{EFEFEF} 
\multicolumn{1}{|l|}{\cellcolor[HTML]{EFEFEF}$d\ge5$} & \multicolumn{1}{c|}{\cellcolor[HTML]{EFEFEF}1}          & \multicolumn{1}{c|}{\cellcolor[HTML]{EFEFEF}1}                & 1           \\ \bottomrule
\end{tabular}
\caption{The order of magnitudes of the lower and upper bounds of Theorem~\ref{thm:localization} and Theorem~\ref{thm:delocalization intro} on the transversal fluctuations of minimal surfaces on $\Lambda_L$ with $n=1$, $\eta^\white$ disorder and fixed $\lambda$. These are presented alongside the predictions from the physics literature (see Section~\ref{sec:physics results background}).
}\label{table:transversal fluctuations n=1} 
\end{table}

\subsubsection{Transversal fluctuations}
Our first pair of results concern the localization and delocalization properties of the minimal surfaces. More precisely, we consider the minimal surface on the domain $\Lambda_L$ with zero boundary values and estimate the norm of the surface at vertices of the domain. Confirming physics predictions (see Section~\ref{sec:physics results background}), we prove that the surface heights have, at least as two-sided bounds, \emph{power-law} fluctuations in dimensions $d=1,2,3$, while having unbounded \emph{sub-power-law} fluctuations in dimension $d=4$ and \emph{bounded} fluctuations in all dimensions $d\ge5$. These results are established for all values of~$n$.

\smallskip
The upper bounds in our localization result, stated next, are smaller for vertices close to the boundary of~$\Lambda_L$. They are accompanied by an upper bound on the \emph{maximal} fluctuation with an iterated logarithm behavior for $d=1$ and a fractional logarithm behavior for $d\in\{2,3\}$.
\begin{theorem}[Localization. \ref{as:exiuni}+\ref{as:stat}+\ref{as:conc}]\label{thm:localization}
There exist $C,c>0$, depending only on $d,n$ and the parameters $(K,\kappa)$ of~\eqref{eq:tal}, such that for each integer $L\ge 1$, each $\lambda>0$ and each $v\in\Lambda_L$,
\begin{equation}\label{eq:localization estimate}
   \E(\|\varphi^{\eta,\lambda,\Lambda_L}_v\|)\le C\sqrt\lambda\begin{cases}r_v^{\frac{4-d}{4}}&d=1,2,3\\ \log(r_v+1)&d=4\\1&d\ge 5\end{cases}
\end{equation}
where we write $r_v:=d_\infty(v,\Lambda_L^c)$.
Moreover, for $d=1$,
\begin{equation}\label{eq:loglog}
    \P\left(\forall v\in\Lambda_L, \  \|\varphi^{\eta,\lambda,\Lambda_L}_v\|\le t\sqrt\lambda \left(\log\log \frac{3L}{r_v}\right)^ {1/4} r_v^ {\frac{4-d}4} \right)\ge 1-Ce^ {-ct^4}
\end{equation}
and for $d\in\{2,3\}$,
\begin{equation}\label{eq:locintro}
    \P\left(\forall v\in\Lambda_L, \  \|\varphi^{\eta,\lambda,\Lambda_L}_v\|\le t\sqrt\lambda \left(\log \frac{2L}{r_v}\right)^ {1/4} r_v^ {\frac{4-d}4} \right)\ge 1-Ce^ {-ct^4}.
\end{equation}
\end{theorem}
An extension of the results in dimensions $d\le 3$ under a stronger concentration hypothesis is presented in Theorem~\ref{thm:scaling relation>} below (see also Remark~\ref{rem:chi conc}).

\smallskip
Our delocalization result, stated next, implies that a uniformly-positive fraction of the surface heights exceed a stated delocalization threshold $h$ with uniformly-positive probability.
\begin{theorem}[Delocalization. $\eta = \eta^\white$]\label{thm:delocalization intro} There exists $c>0$, depending only on $n$ and the Lipschitz constant $\mathcal{L}$ of the bump function in~\eqref{eq:eta white definition}, such that the following holds. Let $L\ge 1$ integer and $\lambda>0$. Then
\begin{equation}\label{eq:expected fraction of delocalization}
    \frac{\mathbb E \big[ | \{v\in \Lambda _L : \|\varphi^{\eta,\lambda,\Lambda_L}_v\| \ge h \} | \big]}{|\Lambda_L|} \ge c
\end{equation}
in each of the following cases:
\begin{equation}\label{eq:h cases in delocalization}
h = \begin{cases} \max\big\{c_\lambda^0 L^{\frac{3}{4+n}}, c_\lambda^1 L^{\frac{1}{2}}\big\} & d=1,\\
c_\lambda^0 L^{\frac{4-d}{4+n}} & d\in\{2,3\},\\
c_\lambda^0 (\log \log L)^{\frac{1}{4+n}} & d=4
\end{cases}
\end{equation}
as long as this value of $h$ is at least $1$, with $c_\lambda^0=c\lambda^{\frac{2}{4+n}}$ and $c_\lambda^1=c\min(\lambda,\lambda^ {-1})^{\frac 12}$.

Additionally, in dimensions $d\ge 5$ we have that~\eqref{eq:expected fraction of delocalization} holds with $h=1$ and $c$ replaced by a constant depending only on $d,n,\lambda$ and the Lipschitz constant $\mathcal{L}$ of the bump function in~\eqref{eq:eta white definition}.
\end{theorem}

We also remark, in relation to the theorem, that if the bump function $b$ defining $\eta^\white$ is radially symmetric then the heights $\varphi^{\eta,\lambda,\Lambda_L}_v$ have a rotationally-invariant distribution.

A mild extension of the theorem is provided by Theorem~\ref{thm: deloc}.

\smallskip
It is common in the physics literature to assume that the heights of the minimal surface in~$\Lambda_L$ (with zero boundary values) fluctuate according to the power law $L^\xi$ for some exponent $\xi$ depending on the dimension $d$ and the number of components $n$. With this notation, we may informally express the above results, for a fixed value of $\lambda$, as
\begin{alignat}{2}
    3/5\le\mbox{}&\xi\le 3/4&& \text{for }d=n=1,\\
    1/2\le\mbox{}& \xi\le 3/4&&\text{for }d=1\text{ and } n\ge 2,\\
    \frac{4-d}{4+n}\le\mbox{}&\xi\le \frac{4-d}{4}\qquad&&\text{for }d\in\{2,3\}\text{ and }n\ge1,\label{eq:xi first bounds}\\
    &\xi = 0&&\text{for }d\ge 5\text{ and }n\ge 1.
\end{alignat}

In the critical dimension $d=4$ we obtain sub-power-law upper and lower bounds which the notation $\xi$ is not refined enough to describe. A summary of our results for fixed $\lambda$ and for $n=1$ is presented in Table~\ref{table:transversal fluctuations n=1}, alongside predictions from the physics literature.

\smallskip
The exponent $\frac{4-d}{4+n}$, $1\le d\le 3$, is called a \emph{Flory-type estimate} in the physics literature and our result confirms predictions that it yields a rigorous lower bound on the height fluctuation exponent $\xi$; see discussion in Section~\ref{sec:transversal fluctuations}.

\subsubsection{Scaling relation}\label{sec:scaling relation}
It is also common in the physics literature to assume that the ground energy fluctuations follow the power law $L^\chi$ where the exponent $\chi$ again depends on the dimension $d$ and the number of components $n$. In this context, it has been put forth (see Section~\ref{sec:physics results background} for background) that the exponents $\chi,\xi$ satisfy the universal \emph{scaling relation}
\begin{equation}\label{eq:scaling relation}
\chi = 2\xi + d-2.
\end{equation}

Our next two statements, taken together, put the scaling relation on rigorous footing. In addition, a more precise version of the relation is established in dimension $d=1$, as presented in the following Section~\ref{sec:1d scaling relation}. Lastly, a version of the relation based on assuming existence of the exponents in suitable senses is presented in Section~\ref{sec:scaling relation assuming existence}. The scaling relation is further discussed in Section~\ref{sec:scaling relation interpretations}.

\smallskip
We start with a version of the $\ge$ direction of the scaling relation~\eqref{eq:scaling relation}, valid for all $d,n\ge 1$. In this version, energy fluctuations are measured by the deviation of the ground energy from its median while transversal fluctuations are measured using the \emph{average} of the surface. In dimensions $d=1,2$, the \emph{pointwise} transversal fluctuations may be used (up to a logarithmic loss for $d=2$). The $\ge$ direction of~\eqref{eq:scaling relation} is obtained from~\eqref{eq:ge direction of scaling relation} below by substituting $h = L^\xi$.

\begin{theorem}[Version of $\chi\ge2\xi+d-2$. \ref{as:exiuni}+\ref{as:stat}]\label{theorem:half scaling relation}
There exists $c>0$, depending only on the dimension $d$, such that for each integer $L\ge 1$, each $\lambda>0$, each unit vector $e\in\R^n$ and all $h>0$,
\begin{equation}\label{eq:ge direction of scaling relation}
    \P\left(\left|\GE^{\eta,\lambda,\Lambda_L}-\med(\GE^{\eta,\lambda,\Lambda_L})\right|\ge ch^2 L^{d-2}\right)\ge \frac{1}{3}\,\P\left(\bigg|\frac{1}{|\Lambda_L|}\sum_{v\in\Lambda_L} \varphi^{\eta,\lambda,\Lambda_{L}}_v\cdot e\bigg| \ge h\right).
\end{equation}
Moreover, in dimensions $d=1,2$, for each $v\in\Lambda_L$, with $r_v:=d_\infty(v,\Lambda_L^c)$, 
\begin{alignat}{3}\label{eq:one dimension pointwise lower bound}&d=1:\quad  &&\mathbb P \big( |\GE^{\eta,\lambda,\Lambda_L} -\med(\GE^{\eta,\lambda,\Lambda_L}) | \ge c h^2/r_v \big)\ge \frac{1}{3}\,\mathbb P \big( | \varphi^{\eta,\lambda,\Lambda_L}_v\cdot e  | \ge h \big),\\
\label{eq:d=2 pointwise scaling relation}&d=2:&&
         \P\left(\left|\GE^{\eta,\lambda,\Lambda_L}-\med(\GE^{\eta,\lambda,\Lambda_L})\right|\ge \frac{ch^2}{\log (r_v+1)}\right)\ge \frac{1}{3}\,\P\left(\left| \varphi^{\eta,\lambda,\Lambda_L}_v\cdot e\right| \ge h\right).
\end{alignat}
\end{theorem}

We emphasize that~\eqref{eq:ge direction of scaling relation} (or~\eqref{eq:one dimension pointwise lower bound}, or~\eqref{eq:d=2 pointwise scaling relation} with a logarithmic loss) is more detailed than the $\ge$ direction of the scaling relation~\eqref{eq:scaling relation} in several ways:
\begin{enumerate}
    \item It applies not only on the level of the exponents $\chi, \xi$ but directly on the fluctuations themselves (i.e., no sub-power corrections to the scaling relation are necessary).
    \item It holds \emph{uniformly} for all disorders satisfying the assumptions~\ref{as:exiuni} and~\ref{as:stat} (this allows a very wide variety of disorders; see remark~\ref{rem:wide variety of disorders}), all values of the disorder strength $\lambda$ and even all values of the number of components~$n$.
    \item It compares the distributions of the transversal and ground energy fluctuations not only in their bulks but also in their tails.
\end{enumerate}

We remark that, for fixed $\lambda$, combining~\eqref{eq:ge direction of scaling relation} with our concentration assumption~\ref{as:conc} (which implies that $\chi\le d/2$) shows that the average height of the surface is at most of order~$L^{\frac{4-d}{4}}$. In particular, the average height decays to zero with $L$ in dimensions $d\ge 5$. As the height at a point does not decay to zero (Theorem~\ref{thm:delocalization intro}), it follows that a pointwise relation such as~\eqref{eq:one dimension pointwise lower bound} or~\eqref{eq:d=2 pointwise scaling relation} cannot hold in dimensions $d\ge 5$ (under~\ref{as:conc}). Similarly, a pointwise relation cannot hold in dimension $d=4$ (under~\ref{as:conc}) as it would imply that the height at points in the bulk of the domain is tight as $L\to\infty$, contradicting our delocalization result (Theorem~\ref{thm:delocalization intro}). It is unclear whether a pointwise relation holds in dimension $d=3$ (see also Section~\ref{sec:improving fluctuation estimates} and Section~\ref{sec:scaling relation interpretations}).

\smallskip
We continue with a version of the $\le$ direction of the scaling relation~\eqref{eq:scaling relation}, also valid for all $d,n\ge 1$. This version bounds the energy fluctuations which arise from the disorder in the \emph{middle part} of the domain (the sub-domain $\Lambda_{\lfloor L/2\rfloor}$), by comparing the ground energy in the given disorder to the ground energy after the disorder is resampled in this middle part. Transversal fluctuations are measured using the \emph{maximum} of the surface. The $\le$ direction of~\eqref{eq:scaling relation} is obtained from~\eqref{eq:le direction of scaling relation} below by substituting $h = L^\xi$.

Given $\Delta\subset\Z^d$, define a disorder $\eta[\Delta]:\Z^d\times\R^n\to(-\infty,\infty]$, with the same distribution as $\eta$, by resampling $\eta$ on $\Delta\times\R^n$. Precisely, we set $\eta[\Delta] := \eta$ on $\Delta^c\times\R^n$ and, conditioned on this restriction of $\eta$ to $\Delta^c\times\R^n$, we let $\eta[\Delta]$ on $\Delta\times\R^n$ be an independent copy of $\eta$ on $\Delta\times\R^n$ (though the conditioning is unimportant here as we shall assume~\ref{as:indep}).

\begin{theorem}[Version of $\chi\le2\xi+d-2$. \ref{as:exiuni}+\ref{as:stat}+\ref{as:indep}]\label{theorem:other half of scaling relation}
There exists $C>0$, depending only on the dimension $d$, such that for each integer $L\ge 1$, each $\lambda>0$, each unit vector $e\in\R^n$ and all $h\ge 1$, 
\begin{equation}\label{eq:le direction of scaling relation}
    \P\left(\big|\GE^{\eta,\lambda,\Lambda_{L}}-\GE^{\eta[\Lambda_{\lfloor L/2\rfloor}],\lambda,\Lambda_{L}}\big|\ge Ch^2 L^{d-2}\right) \le C\P\left(\max_{v\in\Lambda_L} \left|\varphi^{\eta,\lambda,\Lambda_{L}}_v\cdot e\right| \ge h\right).
\end{equation}
\end{theorem}

Equation~\eqref{eq:le direction of scaling relation} is more detailed than the $\le$ direction of the scaling relation~\eqref{eq:scaling relation}, in the same ways described above that Theorem~\ref{theorem:half scaling relation} was more detailed than the $\ge$ direction (except that assumption~\ref{as:indep} should now be added to the list of assumptions in (2) above).

In Theorem~\ref{theorem:local other half of scaling relation} we present a \emph{local} extension of Theorem~\ref{theorem:other half of scaling relation}, which implies, roughly, that if the height fluctuations at vertices at distance $\ell$ from the boundary are at most of order $h$ then the ground energy fluctuations due to the disorder above these vertices are at most of order $h^2\ell^{d-2}$. This allows to bound the (full) ground energy fluctuations via upper bounds on the transversal fluctuations at all distances from the boundary (see also Theorem~\ref{thm:scalingrelation<}). For $d=1$, such an upper bound, together with a near matching lower bound, is explained in the next section. For $d=2$, an analogous lower bound is obtained from~\eqref{eq:d=2 pointwise scaling relation}.

\smallskip

\paragraph{\emph{The scaling relation in dimension $d=1$}}\label{sec:1d scaling relation}
Our techniques yield the following rather precise form of the scaling relation for $d=1$ (and arbitrary $n$).
\begin{theorem}[\ref{as:exiuni}+\ref{as:stat}]\label{thm: scaling relation d=1}
    Suppose $d=1$. There exist universal $C,c>0$ such that the following holds for each integer $L\ge 1$, each $\lambda>0$ and each unit vector $e\in\R^n$. Suppose that $\E(|\GE^{\eta,\lambda,\Lambda_L}|)<\infty$. Define, for integer $k\ge 0$,
    \begin{equation}
        M_k := \max_{L-k\le |v|\le L} |\varphi^{\eta,\lambda,\Lambda_{L}}_v\cdot e|.
    \end{equation}
    Then, on the one hand,
    \begin{equation}\label{eq:std lower bound}
        \std(\GE^{\eta,\lambda,\Lambda_L})\ge c\max_{0\le j\le \lceil\log_2 L\rceil} 2^{-j}(\E M_{2^j})^2
    \end{equation}
    and, on the other hand, if the disorder $\eta$ also satisfies~\ref{as:indep} then
    \begin{equation}\label{eq:std upper bound}
        \std(\GE^{\eta,\lambda,\Lambda_L})\le C\sum_{0\le j\le \lceil\log_2 L\rceil} 2^{-j}\left(1+\sqrt{\E M_{2^j}^4}\right).
    \end{equation}
\end{theorem}
We make two remarks regarding the theorem.
\begin{enumerate}
    \item The bounds are indeed rather precise: it is natural to expect the quantities $(\E M_k)^2$ and $\sqrt{\E M_k^4}$ to be of the same order (i.e., the same up to universal multiplicative constants). In this case, when also $\E M_k$ is at least of order $1$, the lower and upper bounds,~\eqref{eq:std lower bound} and~\eqref{eq:std upper bound}, differ by at most a multiplicative factor of order $\log L$.
    
    We also emphasize that the theorem holds \emph{uniformly} for all disorders satisfying the assumptions, all values of the disorder strength $\lambda$ and even all values of the number of components~$n$.
    
    \item In the most standard case, when all the $(\eta_{v,\cdot})_v$ have the same distribution, we expect that the main contribution to both the lower and upper bounds comes from the ``bulk term'' $j = \lceil \log_2 L\rceil$, leading to the familiar form of the scaling relation $\chi = 2\xi-1$.
    
    However, the theorem captures a richer structure, revealing that the energy fluctuations are determined from the height fluctuations via \emph{contributions from all boundary scales}. The assumptions~\ref{as:exiuni},\ref{as:stat} and~\ref{as:indep} do not force the $(\eta_{v,\cdot})_v$ to have identical distributions (see also Remark~\ref{rem:wide variety of disorders}) and the theorem applies equally well to cases in which the distributions differ significantly; for instance, when each $\eta_{v,\cdot}$ has distribution $\alpha_v \eta^\white$ for arbitrary positive $(\alpha_v)$ (fixing $\lambda=1$). For some choices of $(\alpha_v)$ we expect the largest term of the form $2^{-j}(\E M_{2^j})^2$ to be a ``boundary term'', i.e., to have $j\ll \log_2 L$, in which case the theorem reveals that the scaling relation holds in a different form, stating that $\std(\GE)\approx 2^{-j}(\E M_{2^j})^2$ for that $j$ (up to a possible logarithmic factor, and assuming that $(\E M_k)^2$ is comparable to $1+\sqrt{\E M_k^4}$).

    Another case of interest is when all terms $2^{-j}(\E M_{2^j})^2$ have the same order, in which case it may indeed be that $\std(GE)$ is larger than their magnitude by a poly-logarithmic factor in $L$. This is reminiscent of the behavior in dimension $d=4$, where the predicted fluctuations involve a poly-logarithmic factor (Table~\ref{table:transversal fluctuations n=1}).
\end{enumerate}

\smallskip
\paragraph{\emph{The scaling relation assuming existence of the exponents}}\label{sec:scaling relation assuming existence} In this section we present versions of the scaling relation which are valid under the assumptions that the exponents $\chi$ and $\xi$ are well defined in suitable senses.

\smallskip

{\textbf{The direction $\chi\ge2\xi+d-2$}:} It is natural to assume that the ground energy has stretched-exponential tails on the scale $\ell^\chi$ on box domains of side length $\ell$ (see also Section~\ref{sec:regularity of the distributions}), as phrased in~\eqref{eq:chi conc} below. This implies localization with the predicted exponent $\xi = \frac{1}{2}(\chi - (d-2))$ in the following strong sense (extending the $d\le 3$ cases of Theorem~\ref{thm:localization}).

Recall that the notation $\eta_A$ refers to $(\eta_{v,\cdot})_{v\in A}$ (as defined before assumption~\ref{as:conc}).

\begin{theorem}[\ref{as:exiuni}+\ref{as:stat}] \label{thm:scaling relation>}
Let $\lambda>0$, integer $L\ge 1$ and $d-2<\chi<d$. Suppose that the following concentration estimate holds for some $\alpha, f(\lambda),K,\kappa>0$:
\begin{equation}\tag{$\chi$-conc}\label{eq:chi conc}
\begin{split}
    &\text{for all $\rho>0$, integer $0\le\ell\le L$ and $\Delta=(v+\Lambda_\ell)\cap \Lambda_L$ with $v\in\Lambda_L$,}\\
    &\P\left(\Big|\GE^{\eta,\lambda,\Lambda_L} - \E\big(\GE^{\eta,\lambda,\Lambda_L}\mid \eta_{\Delta^c}\big)\Big|\ge \rho f(\lambda)\ell^{\chi}\Biggm|\eta_{\Delta^c}\right)\le K e^{-\kappa\rho^\alpha}\quad\text{almost surely}.
\end{split}
\end{equation}
Then there exist $C,c>0$, depending only on $d,n,\chi,\alpha,K,\kappa$, such that for each $v\in\Lambda_L$,
\begin{equation}\label{eq:pointwise localization under exponent assumption}
   \E(\|\varphi^{\eta,\lambda,\Lambda_L}_v\|)\le C \sqrt{f(\lambda)}\,r_v^{\frac{\chi - (d-2)}{2}}
\end{equation}
where $r_v:=d_\infty(v,\Lambda_L^c)$. Moreover, for $d=1$,
\begin{equation}\label{eq:loglog chi concentration}
    \P\left(\forall v\in\Lambda_L, \  \|\varphi^{\eta,\lambda,\Lambda_L}_v\|\le t \sqrt{f(\lambda)}\left(\log\log \frac{3L}{r_v}\right)^ {1/(2\alpha)} r_v^{\frac{\chi - (d-2)}{2}} \right)\ge 1-Ce^ {-ct^{2\alpha}}
\end{equation}
and, for $d\ge2$,
\begin{equation}\label{eq:log chi concentration}
    \P\left(\forall v\in\Lambda_L, \  \|\varphi^{\eta,\lambda,\Lambda_{L}}_v\|\le t\sqrt{f(\lambda)}\left(\log \frac{2L}{r_v}\right)^ {1/(2\alpha)} r_v^{\frac{\chi - (d-2)}{2}} \right)\ge 1-Ce^ {-ct^{2\alpha}}.
\end{equation}
\end{theorem}
\begin{remark}\label{rem:chi conc}
Assumption~\eqref{eq:chi conc} is only stated for the given $\lambda,L$, so that, in principle, $f(\lambda)$ is a constant which may be absorbed into $\kappa$. We chose to phrase~\eqref{eq:chi conc} this way since it seems natural that~\eqref{eq:chi conc} holds uniformly in $\lambda$ when $f(\lambda)$ is a properly chosen function of~$\lambda$ (see Section~\ref{sec:dependence on the disorder strength}) and also as we would like to unify the notation with assumption~\ref{as:conc} in order to obtain unified proofs of Theorem~\ref{thm:localization} and Theorem~\ref{thm:scaling relation>}. Indeed, assumption~\ref{as:conc} implies~\eqref{eq:chi conc} uniformly in $\lambda$ and $L$ with $\chi = d/2, \alpha=2$ and $f(\lambda)=\lambda$. Related to this, we mention that the proof of Theorem~\ref{thm:localization} for specific $\lambda,L$ only uses assumption~\ref{as:conc} for that $\lambda$ and for $\Lambda=\Lambda_L$.

The requirement $d-2<\chi<d$ is natural as it corresponds to $0<\xi<1$ by the scaling relation~\eqref{eq:scaling relation}. Note, however, that the following further restrictions on $\chi$ apply if~\eqref{eq:chi conc} is to be satisfied uniformly in $L$, for fixed $\lambda$: if $\eta$ satisfies~\ref{as:conc} (i.e., $\eta=\eta^\white$) then $\chi\le d/2$, whence $\chi>d-2$ is only possible for $d\le 3$. Moreover, for $\eta=\eta^\white$ we have $\chi\ge \frac{d-1}{2}$ due to the boundary contribution to the energy fluctuations~\eqref{eq:boundary contribution to energy fluctuations}. Additional lower bounds on $\chi$ are summarized in~\eqref{eq:best lower bounds on energy fluctuations}.

    While it is convenient to phrase~\ref{eq:chi conc} for all $\rho>0$, the proofs of the statements in Theorem~\ref{thm:scaling relation>} use only a limited range of values of $\rho$. Moreover, in obtaining the pointwise localization statement~\eqref{eq:pointwise localization under exponent assumption}, the stretched-exponential concentration assumption~\ref{eq:chi conc} may be replaced by a suitable (weaker) power-law concentration assumption.

    For $\eta=\eta^\white$ with fixed $\lambda$, say, the stretched-exponential exponent in~\ref{eq:chi conc} may be expected to satisfy $\alpha =  d / (d-\chi)$. This is because the probability that $\GE^{\eta,\lambda,\Lambda_L}$ lies below its expectation by order~$L^d$ is naturally expected to be $\exp(-\Theta(L^d))$ as $L\to\infty$. Similarly, an upper-tail stretched-exponential exponent may be predicted by noting that the probability that $\GE^{\eta,\lambda,\Lambda_L}$ exceeds its expectation by order~$L^d$ is expected to be of order $\exp(-\Theta(L^{d+n}))$ (see Ganguly--Hegde~\cite{ganguly2023optimal} for related statements for $d=n=1$ in directed last-passage percolation).

    Several extensions of the theorem are provided in Section~\ref{sec:loc}:
A tail estimate extending the pointwise expectation bounds~\eqref{eq:localization estimate} and~\eqref{eq:pointwise localization under exponent assumption} is in Lemma~\ref{lem:loc}. Additionally, Lemma~\ref{lem:holder} controls the regularity of the minimal surface in the sense of bounding the height differences $\|\varphi^{\eta,\lambda,\Lambda_L}_u - \varphi^{\eta,\lambda,\Lambda_L}_v\|$ for arbitrary $u,v$. Additionally, a version of the results with the constants $C,c$ uniform in $n$ holds if one replaces $\|\varphi^{\eta,\lambda,\Lambda_L}_v\|$ by $|\varphi^{\eta,\lambda,\Lambda_L}_v\cdot e|$ for a unit vector $e\in\R^n$.
\end{remark}

\smallskip

{\textbf{The direction $\chi\le2\xi+d-2$}:} The next theorem derives a bound on the energy fluctuations in the case that, for some exponent $\xi\ge 0$, the minimal surface is contained up to a lower-order correction within an envelope of the form $d_\infty(v,\Lambda_L^c)^\xi$.
\begin{theorem}[ \ref{as:exiuni}+\ref{as:stat}+\ref{as:indep}]\label{thm:scalingrelation<} Let $\xi\ge0$ and $\beta>0$. There exists $C>0$, depending only on $d,\xi$ and $\beta$, such that the following holds. Let $\lambda>0$ and integer $L\ge 1$. Denote
\begin{equation}\label{eq:kappa def}
\kappa:=\E\left [\max_{v\in\Lambda_L}\left(\frac{\|\varphi^{\eta,\lambda,\Lambda_{L}}_v\|}{r_v^{\xi}(\log \frac L {r_v})^{\beta }}\right)^4\right]
\end{equation}
where we write $r_v:=d_\infty(v,\Lambda_L^c)$. Then
\begin{equation}\label{eq:std upper bound with exponent}
    \std(\GE^ {\eta,\lambda,\Lambda_L})\le C\sqrt \kappa\begin{cases}L ^ {2\xi +d-2} & \xi>\frac{3-d}{4}\\
    L ^ {2\xi +d-2}(\log L)^{2\beta+\frac{1}{2}}& \xi = \frac{3-d}{4}\\
   (\log L)^{2\beta } L^{(d-1)/2}&\xi<\frac{3-d}{4}
    \end{cases}.
\end{equation}
\end{theorem}

\begin{remark}
    The conclusion~\eqref{eq:std upper bound with exponent} with $\beta=0$ also follows (with the same proof) using the following alternative definition of $\kappa$, in which the maximum is taken over cubes whose distance to the boundary of $\Lambda_L$ is commensurate with their side length:
    \begin{equation}\label{eq:kappa alternative def}
        \kappa:=\max_{0\le \ell\le L}\max_{\substack{\Delta = w+\Lambda_\ell\\w\in\Lambda_L,\, r_w\ge2\ell}}\E\left [\max_{v\in\Delta}\left(\frac{\|\varphi^{\eta,\lambda,\Lambda_{L}}_v\|}{r_v^{\xi}}\right)^4\right].
    \end{equation}
    We note that the arguments leading to~\eqref{eq:loglog chi concentration} and~\eqref{eq:log chi concentration} in the proof of Theorem~\ref{thm:scaling relation>} may also be used to control this variant of $\kappa$.

    The regime $\xi\ge \frac{3-d}{4}$ is natural as it translates to $\chi\ge\frac{d-1}{2}$ by the scaling relation, matching the lower bound obtained from the boundary contribution for $\eta=\eta^\white$ in~\eqref{eq:boundary contribution to energy fluctuations}. 
    \end{remark}

{\textbf{The equality $\chi=2\xi+d-2$}:} Let us detail how Theorem~\ref{thm:scaling relation>} and Theorem~\ref{thm:scalingrelation<} combine to yield the scaling relation. Fix $\lambda>0$ and integer $d,n\ge 1$. Let $\eta$ be a disorder, or sequence of disorders depending on~$L$, satisfying assumptions \ref{as:exiuni},\ref{as:stat},\ref{as:indep} and the following:
\begin{enumerate}
    \item There exists $d-2<\chi<d$ such that $\std(\GE^{\eta,\lambda,\Lambda_L}) = \Theta(L^\chi)$ as $L\to\infty$. Moreover, the concentration estimate~\eqref{eq:chi conc} holds for all $L$ with this $\chi$ and the same $\alpha, f(\lambda),K,\kappa$.
    \item There exists $\xi\ge \max\{0,\frac{3-d}{4}\}$ such that $\E(\|\varphi^{\eta,\lambda,\Lambda_L}_0\|)=\Theta(L^\xi)$ as $L\to\infty$. Moreover, with this $\xi$ and some $\beta>0$, the quantity $\kappa$ of~\eqref{eq:kappa def} satisfies $\limsup_{L\to\infty}\kappa<\infty$.
\end{enumerate}
Then the exponents $\chi,\xi$ satisfy the scaling relation~\eqref{eq:scaling relation}.

\subsubsection{Lower bounds on the ground energy fluctuations} 

Our final set of main results concerns the fluctuations of the ground energy $\GE^{\eta,\lambda,\Lambda_L}$ on the cube domain $\Lambda_L$. Our concentration assumption~\ref{as:conc} (which is verified, e.g., for $\eta^\white$) implies that these fluctuations are at most of order $\lambda L^{d/2}$. We present here several bounds in the opposite direction.

First, observe that in dimensions $d=1,2$, lower bounds on the ground energy fluctuations are immediately obtained from the version of the relation $\chi\ge 2\xi+d-2$ given in~\eqref{eq:one dimension pointwise lower bound} and~\eqref{eq:d=2 pointwise scaling relation} together with Theorem~\ref{thm:delocalization intro}. This is made explicit in the following corollary.
\begin{cor}[$\eta = \eta^\white$]\label{cor:lower bounds on energy fluctuations in d=1,2}
There exists $c>0$, depending only on $n$ and the Lipschitz constant $\mathcal{L}$ of the bump function in~\eqref{eq:eta white definition}, such that for each integer $L\ge 2$ and each $\lambda>0$,
\begin{alignat}{3}\label{eq:one two dimensions energy fluctuation lower bound}  &\text{d=1:}&&\quad\mathbb P \big( |\GE^{\eta,\lambda,\Lambda_L} -\med(\GE^{\eta,\lambda,\Lambda_L}) | \ge c h^2/L \big)\ge c,\\
&\text{d=2:}&&\quad\mathbb P \big( |\GE^{\eta,\lambda,\Lambda_L} -\med(\GE^{\eta,\lambda,\Lambda_L}) | \ge c h^2/\log{L} \big)\ge c\end{alignat}
where $h$ is given by~\eqref{eq:h cases in delocalization} (in either of the two cases $d\in\{1,2\}$) as long as that value of $h$ is at least $1$.
\end{cor}

Second, if the minimal surface is localized to height $h$ above a domain $\Lambda^0$ then its energy fluctuations are at least of order $\sqrt{|\Lambda^0|}/h^{n/2}$. This results from the fluctuations of the average disorder in $\Lambda^0\times\{t\in\R^n\colon\|t\|\le h\}$. In particular, it yields a version of the exponent relation
\begin{equation}\label{eq:lower bound on ground energy fluctuations via average disorder}
    \chi\ge \frac{d - n\xi}{2},
\end{equation}
whose special case when $d=1$ is the Wehr--Aizenman~\cite{wehr1990fluctuations} bound of first-passage percolation (proved there via martingale methods).
\begin{theorem}[$\eta = \eta^\white$]\label{thm:lower bounds on energy fluctuations by localization} For each $\delta>0$ there exists $c_\delta>0$ depending only on $\delta,n$ and the Lipschitz constant $\mathcal{L}$ of the bump function in~\eqref{eq:eta white definition} such that the following holds. For each finite $\Lambda^0\subset\Lambda^1\subset\Z^d$ and each $\lambda>0$, if $h\ge 1$ satisfies 
\begin{equation}\label{eq:assumption expected fraction of delocalizatio}
    \frac{\mathbb E \big[ | \{v\in \Lambda^0 : \|\varphi^{\eta,\lambda,\Lambda^1}_v\| \le h \} | \big]}{|\Lambda^0|} \ge \delta
\end{equation}
then
\begin{equation}
    \P\left(\left|\GE^{\eta,\lambda,\Lambda^1}-\med(\GE^{\eta,\lambda,\Lambda^1})\right|\ge c_\delta \lambda\sqrt{\frac{|\Lambda^0|}{h^n}}\right)\ge c_\delta.
\end{equation}
\end{theorem}
Specific lower bounds on the ground energy fluctuations are obtained by combining Theorem~\ref{thm:lower bounds on energy fluctuations by localization} with the localization results of Theorem~\ref{thm:localization}. For instance, taking $\Lambda^1=\Lambda_L$ we obtain, for $\eta = \eta^\white$, each integer $L\ge 2$ and each $\lambda>0$, the two lower bounds:
\begin{enumerate}
    \item $\Lambda^0=\Lambda_L$:
    \begin{equation}\label{eq:bulk contribution to energy fluctuations}
       \P\left(\left|\GE^{\eta,\lambda,\Lambda_L}-\med(\GE^{\eta,\lambda,\Lambda_L})\right|\ge c_\lambda \frac{L^{d/2}}{h^{n/2}}\right)\ge c,
    \end{equation}
    with
    \begin{equation}
        h = \begin{cases}L^{\frac{4-d}{4}}&d=1,2,3\\
        \log(L)&d=4\\
        1&d\ge 5
        \end{cases},
    \end{equation}
    \item $\Lambda^0$ the interior vertex boundary of $\Lambda_L$:
    \begin{equation}\label{eq:boundary contribution to energy fluctuations}
    \P\left(\left|\GE^{\eta,\lambda,\Lambda_L}-\med(\GE^{\eta,\lambda,\Lambda_L})\right|\ge c_\lambda L^{\frac{d-1}{2}}\right)\ge c,
\end{equation}
\end{enumerate}
where $c>0$ depends only on $n$ and the Lipschitz constant $\mathcal{L}$ of the bump function in~\eqref{eq:eta white definition} and $c_\lambda>0$ depends additionally on $d$ and $\lambda$.

\begin{table}[]
\begin{tabular}{@{}|lccc|@{}}
\toprule
\rowcolor[HTML]{FFCCC9} 
\multicolumn{4}{|c|}{\cellcolor[HTML]{96C5FF}Ground energy fluctuations for $n=1$}  
\\ \midrule
\multicolumn{1}{|l|}{}                                & \multicolumn{1}{c|}{Lower bound}                       & \multicolumn{1}{c|}{Predicted}                                & Upper bound \\ \midrule
\rowcolor[HTML]{EFEFEF} 
\multicolumn{1}{|l|}{\cellcolor[HTML]{EFEFEF}$d=1$}   & \multicolumn{1}{c|}{\cellcolor[HTML]{EFEFEF}$L^{0.2}$} & \multicolumn{1}{c|}{\cellcolor[HTML]{EFEFEF}$L^{1/3}$}        & $L^{0.5}$  \\ \midrule
\multicolumn{1}{|l|}{$d=2$}                           & \multicolumn{1}{c|}{$L^{0.8}/\log(L)$}                         & \multicolumn{1}{c|}{$L^{0.84\pm0.03\ldots}$}                         & $L$   \\ \midrule
\rowcolor[HTML]{EFEFEF} 
\multicolumn{1}{|l|}{\cellcolor[HTML]{EFEFEF}$d=3$}   & \multicolumn{1}{c|}{\cellcolor[HTML]{EFEFEF}$L^{1.375}$} & \multicolumn{1}{c|}{\cellcolor[HTML]{EFEFEF}$L^{1.45\pm0.04\ldots}$} & $L^{1.5}$  \\ \midrule
\multicolumn{1}{|l|}{$d=4$}                           & \multicolumn{1}{c|}{$L^2 / \sqrt{\log L}$}                                 & \multicolumn{1}{c|}{$?$}                  & $L^2$    \\ \midrule
\rowcolor[HTML]{EFEFEF} 
\multicolumn{1}{|l|}{\cellcolor[HTML]{EFEFEF}$d\ge5$} & \multicolumn{1}{c|}{\cellcolor[HTML]{EFEFEF}$L^{d/2}$}          & \multicolumn{1}{c|}{\cellcolor[HTML]{EFEFEF}$L^{d/2}$}                & $L^{d/2}$           \\ \bottomrule
\end{tabular}
\caption{The order of magnitudes of the lower and upper bounds on the ground energy fluctuations in the cube domain $\Lambda_L$ with $n=1$, $\eta^\white$ disorder and fixed $\lambda$. These are presented alongside the predictions from the physics literature (see Section~\ref{sec:physics results background}).
}\label{table:energy fluctuations n=1} 
\end{table}

The following list summarizes the best lower bounds on the ground energy fluctuations for each dimension $d$ and number of components $n$, which are obtained for $\eta=\eta^\white$ and fixed $\lambda$ from Corollary~\ref{cor:lower bounds on energy fluctuations in d=1,2} and the combination of Theorem~\ref{thm:lower bounds on energy fluctuations by localization} with Theorem~\ref{thm:localization}:
\begin{equation}\label{eq:best lower bounds on energy fluctuations}
    \E\left(\left|\GE^{\eta,\lambda,\Lambda_L}-\med(\GE^{\eta,\lambda,\Lambda_L})\right|\right)\ge c_\lambda\begin{cases}L^{\frac{1}{5}}&d=n=1,\\
    1&d=1, n\ge 2,\\
    L^{\frac{4}{4+n}}(\log L)^{-1}&d=2, 1\le n\le 3,\\
    \sqrt{L}&d=2, n\ge 4,\\
    L^{\frac{3}{2} - \frac{n}{8}}&d=3, 1\le n\le 3,\\
    L&d=3, n\ge 4,\\
    L^2(\log L)^{-\frac{n}{2}}&d=4,\\
    L^{\frac{d}{2}}&d\ge 5.
    \end{cases}
\end{equation}
We also remind that assumption~\ref{as:conc} provides the upper bound $C\lambda L^{d/2}$, which in dimensions $d\ge 5$ is sharp up to the dependence on $\lambda$. Here $C,c_\lambda>0$ depend only on $d,n$ and the Lipschitz constant $\mathcal{L}$ of the bump function in~\eqref{eq:eta white definition}, with $c_\lambda$ depending additionally on~$\lambda$. Table~\ref{table:energy fluctuations n=1} summarizes the lower and upper bounds for the case $n=1$ alongside the physics predictions.

\subsection{Background and comparison with previous work}\label{sec:background}

In this section we discuss some of the physics and mathematics literature on related models and point out how our results fit within the existing picture.

\subsubsection{Harmonic minimal surfaces in random environment} The model~\eqref{eq:formal Hamiltonian} and closely-related variants have received significant attention in the physics literature, e.g., in the papers~\cite{grinstein1982roughening, villain1982commensurate, grinstein1983surface, huse1985pinning, fisher1986interface, kardar1987domain, nattermann1987interface, nattermann1988random, balents1993large,  middleton1995numerical, emig1998roughening, scheidl2000interface, le2004functional, husemann2018field} and the reviews~\cite{forgacs1991behavior,giamarchi1998statics, wiese2003functional, giamarchi2009disordered,ferrero2021creep,wiese2022theory}. As we discuss in Section~\ref{sec:disorder types}, the model has been proposed as an approximation to the domain walls of various disordered spin systems and also as a height function approximation to the angle field of the random-field XY model. In the mathematically-rigorous literature, however, the model has received only little attention: research has been limited to the works of Bakhtin et al.~\cite{bakhtin2014space, bakhtin2016inviscid,bakhtin2019thermodynamic, bakhtin2018zero, bakhtin2022dynamic} on the $d=n=1$ case in relation with the Burgers equation, to the works of Berger and Torri~\cite{berger2019entropy,berger2021beyond} on the $d=n=1$ case in continuous space with a Poisson environment and to studies with $d=1$ and general $n$ of related (positive temperature) models of directed polymers based on Gaussian random walk or on Brownian motion (see, e.g., the review by Comets--Cosco~\cite{comets2018brownian}). The recent work of Ben Arous--Bourgade--McKenna~\cite{ben2024landscape} considers the model~\eqref{eq:formal Hamiltonian} (termed \emph{the elastic manifold} there) for fixed $d$ and $n\to\infty$ but focuses on different questions, of landscape complexity (following Fyodorov--Le Doussal~\cite{fyodorov2020manifolds}; see also~\cite{fyodorov2018exponential,fyodorov2020bmanifolds}). Our work appears to be the first mathematically-rigorous investigation of the surfaces minimizing~\eqref{eq:formal Hamiltonian} in $d\ge 2$ dimensions.

\subsubsection{The transversal fluctuations, scaling relation and ground energy fluctuations in the physics literature}\label{sec:physics results background}

The height fluctuations of the minimal surface and the ground energy fluctuations, as captured (at least for $d\le 3$) by the transversal fluctuation exponent $\xi$ and the energy fluctuation exponent $\chi$, received great attention in the physics literature. Already the seminal work of Huse--Henley~\cite{huse1985pinning}, dealing with general $d$ and $n=1$, put forth the scaling relation~\ref{eq:scaling relation}, predicted that the minimal surface is localized for $d>4$ (following arguments of Imry--Ma~\cite{imry1975random}), delocalized for $d<4$ and predicted, based on numerical simulations, the exponents $\chi=1/3$ and $\xi=2/3$ for $d=n=1$. Theoretical arguments for the $d=n=1$ exponents were then presented by Kardar--Nelson~\cite{kardar1985commensurate} and Huse--Henley--Fisher~\cite{huse1985huse}. The exponents reported in Table~\ref{table:transversal fluctuations n=1} and Table~\ref{table:energy fluctuations n=1} for dimensions $d=2,3$ were obtained by Middleton~\cite{middleton1995numerical} using numerical simulations based on an exact min-cut algorithm. Fisher~\cite{fisher1986interface} developed a functional renormalization group method to obtain an expansion of $\xi$ to first order in $\varepsilon>0$ in dimension $d=4-\varepsilon$ when $n=1$ (later extended to second~\cite{le2004functional} and third orders~\cite{husemann2018field}). The expression presented in Table~\ref{table:transversal fluctuations n=1} for $d=4$ is taken from the renormalization group results of Emig--Natterman (\cite{emig1998roughening} and ~\cite[following (79)]{emig1999disorder}). We are not aware of a prediction for the ground energy fluctuations for $d=4$ (we are also unaware of an interpretation of the scaling relation for $d=4$; see also the remarks after Theorem~\ref{theorem:half scaling relation} and Section~\ref{sec:scaling relation interpretations}).

A replica-theoretic approach, aided by a variational method, was developed by M{\'e}zard--Parisi~\cite{mezard1990interfaces,mezard1991replica,mezard1992manifolds} for the study of $\xi$ for $n$ large. 
Balents--Fisher~\cite{balents1993large} extend the functional renormalization group method to $d=4-\varepsilon$ with $n$ large; further extensions are in~\cite{emig1999disorder, le2002functional, le2003functional} and reviewed in~\cite{wiese2022theory}. 

It is believed that the transversal fluctuation exponent $\xi$ is non-increasing with the number of components $n$ (for a fixed $d$). For $d=1$, it is also believed that $\xi\ge1/2$, as we prove here. It is an open question in the physics literature whether there exists $n_0<\infty$ for which $\xi=1/2$ exactly; see~\cite[equation (783) and discussion in Section 7.11]{wiese2022theory}.

\subsubsection{Comparison with mathematical literature} We summarize here some of the ways in which our results complement the existing mathematical literature.
\smallskip
\paragraph{\emph{Transversal fluctuations}.}\label{sec:transversal fluctuations} Our work justifies the prediction that minimal surfaces in ``independent disorder'' are delocalized in dimensions $1\le d\le 4$ and localized in dimensions $d\ge 5$. Moreover, we prove delocalization with power-law fluctuations in dimensions $1\le d\le 3$ and sub-power-law fluctuations in dimension $d=4$. Only the $d=1$ case was previously proved, by Licea--Newman--Piza~\cite{licea1996superdiffusivity}. Moreover, even for $d=1$, our lower bounds $\xi\ge 3/5$ for $d=n=1$ and $\xi\ge 1/2$ for $d=1, n\ge 2$ are new, though the following related results were known: W{\"u}thrich~\cite{wuthrich1998superdiffusive} proves $\xi\ge 1/2$ for all $n$ and $\xi\ge 3/5$ for $n=1$ for \emph{point-to-plane} Brownian motion in a Poissonian potential (a \emph{``positive temperature''} model), mentioning that the arguments do not adapt to the point-to-point setting. Bakhtin--Wu~\cite{bakhtin2019transversal} show $\xi\ge 3/5$ for $n=1$ for a point-to-line first-passage percolation model in a Poissonian environment. Licea--Newman--Piza~\cite{licea1996superdiffusivity} also establish versions of both bounds, but for definitions of $\xi$ which consider ``almost geodesics in some direction''.
    The upper bound $\xi\le 3/4$ for $d=1$ was shown by Newman--Piza~\cite{newman1995divergence} for the length of geodesics in directions of curvature in first-passage percolation (see also~\cite{bakhtin2019thermodynamic,bakhtin2022dynamic}).

    Balents--Fisher~\cite[Section VI]{balents1993large} predict that for all $d\le 3$ and $n\ge 1$ the exponent $\xi$ is between $\frac{4-d}{4+n}$ and $\frac{4-d}{4}$. The prediction $\xi\ge\frac{3}{4+n}$ for $d=1$ was also made by Le Doussal and Machta~\cite[Last displayed equation in Section 6.1]{le1991self}. These predictions are established rigorously in this work. The exponent $\frac{4-d}{4+n}$ is termed a Flory-type estimate~\cite{mezard1991replica, giamarchi2009disordered, wiese2022theory}.

    The law of the iterated logarithm (LIL)~\eqref{eq:loglog chi concentration} upper bound ($d=1$) and the law of the fractional logarithm~\eqref{eq:log chi concentration} upper bound ($d\ge 2$) for the \emph{heights} of the minimal surface are new, even for the $d=n=1$ case. In the latter case, plugging the predicted values $\chi=1/3$ and $\alpha=3/2$ in~\eqref{eq:chi conc}, our LIL upper bound~\eqref{eq:loglog chi concentration} has the scaling $(\log\log x)^{1/3}$. A related LIL for the \emph{ground energy} on a sequence of growing length scales is established, for integrable directed last-passage percolation ($d=n=1$), by Ledoux~\cite{ledoux2018law}. In fact, the scaling of the LIL of~\cite{ledoux2018law} differs for the upper and lower deviations of the ground energy, being $(\log\log x)^{2/3}$ for the former and $(\log\log x)^{1/3}$ for the latter (the last claim is completed by Basu et al.~\cite{basu2021lower}). The $(\log\log x)^{1/3}$ scaling that our LIL exhibits for $d=n=1$ may indeed be sharp, as it seems plausible that it is a height counterpart to the $(\log\log x)^{2/3}$ ground-energy LIL of Ledoux via the scaling relation~\eqref{eq:scaling relation}. We note also that an LIL for the solution of the KPZ equation with narrow wedge initial data is established by Das--Ghosal~\cite{das2023law}.
    
    \smallskip
    \paragraph{\emph{Scaling relation}.} The inequality $\chi'\ge 2\xi-1$ of the $d=1$ scaling relation was established by Newman--Piza~\cite{newman1995divergence} for all $n$ for geodesics in ``directions of curvature'' in first-passage percolation, with $\chi'$ the exponent for the rate of convergence to the limit shape (the so-called ``non-random fluctuations'') in the slowest direction. The full scaling relation for $d=1$ and arbitrary $n$ was established under assumptions which are still unverified by W{\"u}thrich~\cite{wuthrich1998scaling} in the context of Brownian motion in a Poissonian potential and by Chatterjee~\cite{chatterjee2013universal} for lattice first-passage percolation, with simplifications of the latter and an extension to a directed, positive-temperature version by Auffinger--Damron~\cite{auffinger2014simplified,auffinger2013scaling}. The scaling relations proved in these works are on the level of exponents, allowing for possible sub-power-law corrections. Just before the first version of this paper appeared, an impressive work of Basu--Sidoravicius--Sly~\cite{basu2023rotationally} showed that the scaling relation holds when $d=n=1$, under assumptions which will be verified for a class of rotationally-invariant first-passage percolation models in a sequel paper. For all choices of $d,n$ other than $d=n=1$, our work is the first to establish versions of the scaling relation without relying on unverified assumptions. For $d=n=1$ it appears concurrently with~\cite{basu2023rotationally}. Moreover, our results apply for a wide class of disorder distributions and show that the scaling relation holds without sub-power-law corrections and even on the level of the tails of the distributions.

Alexander~\cite{alexander2023properties} considers $d=1$ and general $n$ and initiates a mathematical exploration of the consequences of having $\chi=0$ (equivalent to $\xi=1/2$ by the scaling relation). In relation to this, note that Theorem~\ref{thm:delocalization intro} and~\eqref{eq:std lower bound} show that, for $\eta=\eta^\white$ and fixed $\lambda$, if the ground energy has standard deviation of order $1$ then the (typical) maximal transversal fluctuation is exactly of order $L^{1/2}$.

\smallskip
    \paragraph{\emph{Ground energy fluctuations}.} Our lower bounds on the ground energy fluctuations~\eqref{eq:best lower bounds on energy fluctuations} are all new. Earlier results concern the $d=n=1$ case: Close to our $L^{1/5}$ lower bound is the result of Newman--Piza~\cite{newman1995divergence} showing that $\max\{\chi,\chi'\}\ge 1/5$, where $\chi$ is the ground energy fluctuation exponent maximized over all directions and $\chi'$ is the exponent for the ``non-random fluctuations'' mentioned above. It is also shown by Newman--Piza~\cite{newman1995divergence} for general first-passage percolation models with $d=n=1$ that the ground energy fluctuations are at least $\sqrt{\log L}$ (also obtained by Pemantle--Peres~\cite{pemantle1994planar} for exponential weights), which is improved to $L^{1/8}$ in directions of curvature (using the bound $\xi\le3/4$ and the Wehr--Aizenman~\cite{wehr1990fluctuations} inequality $\chi\ge \frac{1-\xi}{2}$).

\subsection{Reader's guide} Section~\ref{sec:main identity and consequences} presents the \emph{main identity} (Proposition~\ref{prop:main identity}), an identity satisfied by the Hamiltonian~\eqref{eq:formal Hamiltonian} which lies at the heart of all of our proofs. First consequences of the identity including a basic relation~\eqref{eq:lower bound on fluctuations} between the heights of the minimal surface and its ground energy are also presented. 

Localization and the direction $\chi\ge2\xi+d-2$ of the scaling relation are the subject of Section~\ref{sec:loc}, while delocalization and the direction $\chi\le2\xi+d-2$ of the scaling relation are the subject of Section~\ref{sec:deloc}.

In Section~\ref{sec:assumptions for eta white} we verify that the disorder $\eta^\white$ satisfies all the assumptions of Section~\ref{sec:disorder assumptions}.

Section~\ref{sec:disorder types} discusses universality classes of disorders, the predicted behavior of the minimal surfaces with such disorders and connections with disordered spin systems. 

Section~\ref{sec:open questions and extensions} discusses extensions of our results and several of the outstanding open questions.

\section{Main identity, first consequences and overview of proof of main results}\label{sec:main identity and consequences}

This section introduces the main identity, which is the special feature of the model~\eqref{eq:formal Hamiltonian} on which our analysis is based. We start with required notation, proceed to describe the identity and then present some of its first consequences.

\subsection{Notation}
\subsubsection*{Shifts} Given a function $s:\Z^d\to\R^n$ we let
\begin{equation}
    \supp(s) := \{v\in\Z^d\colon s_v\neq 0\}
\end{equation}
be the \emph{support} of $s$. Given additionally an environment $\eta:\Z^d\times\R^n\to(-\infty,\infty]$ we 
recall from~\eqref{eq:eta s def} that the $s$-shifted environment $\eta^s$ is defined by \begin{equation*}
    \eta^s_{v,t}:=\eta_{v,t-s_v}.
\end{equation*}

\subsubsection*{Inner products and Laplacian} For $x,y\in\R^n$ we use the standard notation
    \begin{equation}\label{eq:inner product R n}
        x\cdot y:=\sum_{i=1}^n x_i y_i
    \end{equation}
    so that $\|x\|^2=x\cdot x$. We use the notation $(\cdot ,\cdot )$ for inner products of functions on the vertices:  
    Given $\varphi,\psi:\Z^d\to\R^n$ and $\Lambda\subset\Z^d$, we let
    \begin{equation}
        (\varphi,\psi)_\Lambda:=\sum_{v\in\Lambda} \varphi_v \cdot  \psi_v
    \end{equation}
    and also set $\|\varphi\|_\Lambda^2:=(\varphi,\varphi)_\Lambda$.
    In addition, we write
    \begin{equation}
        (\nabla \varphi, \nabla\psi)_\Lambda:= \sum_{\substack{u\sim v\\\{u,v\}\cap\Lambda\neq\emptyset}}(\varphi_u - \varphi_v)\cdot ( \psi_u - \psi_v)
    \end{equation}
    and also set $\|\nabla\varphi\|_\Lambda^2:=(\nabla\varphi,\nabla\varphi)_\Lambda$ (the Dirichlet energy of $\varphi$ on $\Lambda$). We use the abbreviations $(\varphi,\psi):=(\varphi,\psi)_{\Z^d}$ and $(\nabla \varphi,\nabla\psi):=(\nabla \varphi,\nabla\psi)_{\Z^d}$.

    We note the following useful \emph{discrete Green's identity},
    \begin{equation}\label{eq:Green's identity}
        (\nabla\varphi,\nabla\psi)_\Lambda=(\varphi,-\Delta_\Lambda\psi)=(-\Delta_\Lambda\varphi,\psi),
    \end{equation}
    which holds whenever one of the sums (and then all others) converges absolutely. Here,
    $\Delta_\Lambda$ is the discrete Laplacian which acts on $\varphi:\Z^d\to\R^n$ by
    \begin{equation}\label{eq:Laplacian definition}
        (\Delta_\Lambda\varphi)_v:=\sum_{u\,:\, \substack{u\sim v,\\ \{u,v\}\cap\Lambda\neq\emptyset}} (\varphi_u - \varphi_v).
    \end{equation}
    In particular, letting
    \begin{equation}
        \Lambda^+:=\{v\in\Z^d\colon \exists u, u\sim v\text{ and }\{u,v\}\cap\Lambda\neq \emptyset\}
    \end{equation}
    we have $\Delta_\Lambda\varphi=0$ off $\Lambda^+$ with this definition, so that, e.g., $(-\Delta_\Lambda\varphi,\psi) = (-\Delta_\Lambda\varphi,\psi)_{\Lambda^+}$. Again, we abbreviate $\Delta:=\Delta_{\Z^d}$.

    \subsubsection*{Harmonic extension and Dirichlet energy} Given a finite $\Lambda\subset\Z^d$ and $\tau:\Z^d\to\R^n$, let $\overline{\tau}^\Lambda$ be the harmonic extension of $\tau$ to $\Lambda$. Precisely, $\overline{\tau}^\Lambda:\Z^d\to\R^n$ is the unique function satisfying
    \begin{equation}
    \begin{split}
        &\overline{\tau}^\Lambda = \tau\quad\text{on $\Z^d\setminus\Lambda$},\\
        &\Delta_\Lambda(\overline{\tau}^\Lambda) = 0\quad\text{on $\Lambda$}.
    \end{split}
    \end{equation}
    This is equivalent, as is well known, to setting $\bar{\tau}^\Lambda_v = \E(\tau_{X_{T_\Lambda}})$ with $(X_t)_{t=0}^\infty$ a simple random walk on $\Z^d$ started at $v$ and $T_\Lambda:=\min\{t\ge 0\colon X_t\notin\Lambda\}$. In addition, we let
    \begin{equation}
    \|\tau\|_{\DE(\Lambda)}^2:=\|\nabla\bar{\tau}^\Lambda\|_\Lambda^2
    \end{equation}
    be the Dirichlet energy of the harmonic extension of $\tau$ to $\Lambda$.

\subsection{Main identity} The starting point for our analysis of the disordered random surface model~\eqref{eq:formal Hamiltonian} is the following deterministic identity, which controls the change in the energy of surfaces under a shift of both the disorder and the surface. 
\begin{prop}[Main identity]\label{prop:main identity} Let $\eta:\Z^d\times\R^n\to(-\infty,\infty]$, $\lambda\in\R$ and finite $\Lambda\subset\Z^d$. For each $s:\Z^d\to\R^n$ and $\varphi:\Z^d\to\R^n$ we have
\begin{equation}\label{eq:main identity}
    H^{\eta^s,\lambda,\Lambda}(\varphi+s) - H^{\eta,\lambda,\Lambda}(\varphi) = (\varphi,-\Delta_\Lambda s) + \frac{1}{2}\|\nabla s\|_\Lambda^2.
\end{equation}
\end{prop}

\begin{proof}
We first observe that, by the definition~\eqref{eq:eta s def} of $\eta^s$, for each $v\in\Z^d$,
\begin{equation}
    \eta^s_{v,\varphi_v+s_v} = \eta_{v,\varphi_v}.
\end{equation}
Therefore, using the discrete Green's identity~\eqref{eq:Green's identity},
\begin{equation}
\begin{split}
    H^{\eta^s,\lambda,\Lambda}(\varphi+s) - H^{\eta,\lambda,\Lambda}(\varphi) &= \frac{1}{2}(\|\nabla (\varphi+s)\|_\Lambda^2 - \|\nabla\varphi\|_\Lambda^2) + \lambda \sum_{v\in\Lambda} (\eta^s_{v,\varphi_v+s_v}-\eta_{v,\varphi_v})\\
    &=\frac{1}{2}\big((\nabla(\varphi+s),\nabla(\varphi+s))_\Lambda - (\nabla\varphi, \nabla\varphi)_\Lambda\big)\\
    &=(\nabla\varphi,\nabla s)_\Lambda + \frac{1}{2}(\nabla s, \nabla s)_\Lambda=(\varphi,-\Delta_\Lambda s) + \frac{1}{2}\|\nabla s\|_\Lambda^2.\qedhere
\end{split}    
\end{equation}
\end{proof}
We remark that the idea of shifting the configurations of a disordered spin system together with the disorder was also used by Aizenman--Wehr~\cite{aizenman1990rounding} in their study of the random-field spin $O(n)$ models (and its quantitative extension~\cite{dario2024quantitative}).

\subsection{First consequences}
The above identity is our main tool in proving the results of this paper. We collect here several of its immediate consequences. 
\subsubsection{Effect of boundary condition and limit shape}
Among the consequences of the main identity is a useful description of the effect of boundary conditions as given by the following corollary.
\begin{cor}[Effect of boundary conditions. \ref{as:exiuni}]\label{cor:effect of boundary conditions} Let $\Lambda\subset \Z^d$ and $\tau:\Z^d \to \R^n$. Then
\begin{alignat}{1}
\varphi^ {\eta,\lambda,\Lambda,\tau }&=\varphi^ {\eta^ {-\bar{\tau}^\Lambda},\lambda,\Lambda}+\bar{\tau}^\Lambda,\\
\GE^{\eta,\lambda,\Lambda,\tau}&= \GE^ {\eta^ {-\bar{\tau}^\Lambda},\lambda,\Lambda }+\frac{1}{2}\|\tau\|_{\DE(\Lambda)}^2.
\end{alignat}
In particular, if the disorder $\eta$ also satisfies~\ref{as:stat} then
\begin{equation}
    (\varphi^ {\eta,\lambda,\Lambda,\tau}, \GE^{\eta,\lambda,\Lambda,\tau})\eqd(\varphi^ {\eta,\lambda,\Lambda}, \GE^{\eta,\lambda,\Lambda}) + \Big(\bar{\tau}^\Lambda,\frac{1}{2}\|\tau\|_{\DE(\Lambda)}^2\Big).
\end{equation}
\end{cor}

\begin{proof}
    First note that $\varphi+{\bar{\tau}^\Lambda}\in \Omega^ {\Lambda,\tau}$ if and only if $\varphi\in\Omega^ {\Lambda}$.
        Second, Proposition \ref{prop:main identity} implies that for each $\varphi\in\Omega^ {\Lambda}$,
        \begin{equation*}
H^{\eta,\lambda,\Lambda}(\varphi+{\bar{\tau}^\Lambda}) - H^{\eta^{-\bar{\tau}^\Lambda},\lambda,\Lambda}(\varphi) = (\varphi,-\Delta_\Lambda {\bar{\tau}^\Lambda}) + \frac{1}{2}\|\nabla \bar{\tau}^\Lambda\|_\Lambda^2= \frac{1}{2}\|\tau\|_{\DE(\Lambda)}^2
\end{equation*}
using that $\Delta_\Lambda(\overline{\tau}^\Lambda) = 0$ on $\Lambda$ and $\varphi=0$ off $\Lambda$. Thus, $\varphi$ minimizes $H^{\eta^{-\bar{\tau}^\Lambda},\lambda,\Lambda}$ over $\Omega^ {\Lambda}$ if and only if $\varphi+{\bar{\tau}^\Lambda}$ minimizes $H^{\eta,\lambda,\Lambda}$ over $\Omega^ {\Lambda,\tau}$. In the case of multiple minimizers, recall that $\varphi^{\eta,\lambda,\Lambda}$ is selected by lexicographic order and note that shifting preserves the ordering.
\end{proof}
The corollary generalizes the, so called, \emph{shear invariance} familiar from the $d=1$ case (see, e.g., \cite{bakhtin2014space,bakhtin2016inviscid} and \cite[Lemma~10.2]{dauvergne2022directed}) to \emph{harmonic invariance} (invariance under shifts by an harmonic function). We emphasize, however, that we will often make use of the main identity~\eqref{eq:main identity} with non-harmonic shift functions $s$, for which the term $(\varphi,-\Delta_\Lambda s)$ plays an important role.

Corollary~\ref{cor:effect of boundary conditions} shows that the boundary conditions $\tau$ ``cost'' exactly $\frac{1}{2}\|\tau\|_{\DE(\Lambda)}^2$ more in energy than zero boundary conditions. For $d=1$, this is a strong (finite-volume) version of the strict convexity of the limit shape (a property believed to hold in many first-passage percolation models). We briefly make this connection explicit (though it will not be used in the rest of the paper): For $d=1$, let $L\ge 2$, $h\in\R^n$ and $I_L := \{1,\ldots, L-1\}$. Let $\tau^{L,h}:\Z\to\R^n$ be any function with $\tau^{L,h}_0 = 0$ and $\tau^{L,h}_L=h$. The harmonic extension of $\tau^{L,h}$ to $I_L$ is linear, whence $\frac{1}{2}\|\tau^{L,h}\|_{\DE(I_L)}^2=\frac{\|h\|^2}{2L}$. The minimal surface $\varphi^{\eta,\lambda,I_L,\tau^{L,h}}$ is analogous to a first-passage percolation geodesic which starts at the origin at time $0$ and ends at $h$ at time $L$. Define the time constant function (see \cite[Equation~(2.4)]{50years}) by 
\begin{equation}\label{eq:time constant function}
    \mu(x):=\lim _{L\to \infty }\frac{1}{L}\E\left(\GE^{\eta,\lambda,I_L,\tau^{L,xL}}\right)\quad\text{for }x\in\R^n,
\end{equation}
assuming the expression is well defined.
Then, by Corollary~\ref{cor:effect of boundary conditions} (under~\ref{as:exiuni}+\ref{as:stat}),
\begin{equation}\label{eq:d=1 limit shape}
    \mu(x)=\mu (0) +\frac{1}{2}\|x\|^2\quad\text{for }x\in\R^n.
\end{equation}

Corollary~\ref{cor:effect of boundary conditions} also implies directly the following statement, yielding a stationary ``boundary conditions process''. The statement will not be used in the rest of the paper.
\begin{cor}[Boundary conditions process. \ref{as:exiuni}+\ref{as:stat}]\label{cor:stationary process} Let $\Lambda\subset\Z^d$ be finite. Then the stochastic process $\tau\mapsto(\varphi^{\eta,\lambda,\Lambda,\tau} - \overline{\tau}^\Lambda,\, \GE^{\eta,\lambda,\Lambda,\tau} - \frac{1}{2}\|\tau\|_{\DE(\Lambda)}^2)$ is stationary. In other words, for each integer $k\ge1$, $\tau_1,\ldots,\tau_k:\Z^d\to\R^n$ and $s:\Z^d\to\R^n$ the joint distribution of the pairs
\begin{equation}
    (\varphi^{\eta,\lambda,\Lambda,\tau_i+s}-\overline{\tau_i+s}^\Lambda,\,\GE^{\eta,\lambda,\Lambda,\tau_i+s} - \frac{1}{2}\|\tau_i+s\|_{\DE(\Lambda)}^2  
    )\quad\text{for $1\le i\le k$}
\end{equation}
does not depend on $s$. 
\end{cor}

\subsubsection{Concentration for linear functionals}
The next lemma, which follows from the main identity, is our key tool for proving the direction $\chi\ge2\xi+d-2$ of the scaling relation and for providing upper bounds on the height fluctuations of the minimal surface.

Recall that the notation $\eta_A$ refers to $(\eta_{v,\cdot})_{v\in A}$ (as defined before assumption~\ref{as:conc}).

\begin{lemma}[Concentration for linear functionals of the minimal surface. \ref{as:exiuni}+\ref{as:stat}]\label{lem:fluctuation and concentration} Let $\lambda>0$, $\Lambda\subset\Z^d$ be finite, $s:\Z^d\to\R^n$ and $r>0$. Then, almost surely,
\begin{equation}\label{eq:lower bound on fluctuations}
\P\left(\left|(\varphi^{\eta,\lambda,\Lambda},-\Delta_\Lambda s)\right|\ge r\mid \eta_{\supp(s)^c}\right)\le 3\inf_{\gamma\in\R}\P\left(\left|\GE^{\eta,\lambda,\Lambda}-\gamma\right|\ge \frac{r^2}{4\|\nabla s\|_\Lambda^2}\mid \eta_{\supp(s)^c}\right).
\end{equation}

Consequently, 
\begin{enumerate}
    \item if $\eta$ additionally satisfies assumption~\ref{as:conc} then, almost surely,
    \begin{equation}\label{eq:linear functional prob upper bound conc}
\P(|(\varphi^{\eta,\lambda,\Lambda}, -\Delta_\Lambda s)|\ge r\mid \eta_{\supp(s)^c})\le 3K\exp\bigg(-\frac{\kappa r^4}{16\lambda^2\|\nabla s\|_\Lambda^4|\supp(s)|}\bigg)
\end{equation}
where $K,\kappa$ are of~\eqref{eq:tal}.

\item if $\eta$ additionally satisfies the concentration estimate~\eqref{eq:chi conc}, $\Lambda =\Lambda _L$ for some $L\ge 1$ and ${\rm Supp}(s)\subseteq \Delta$ for a box $\Delta$ of the form $\Delta=(v+\Lambda_\ell) \cap \Lambda _L$ for some $0\le \ell\le L$ and $v\in\Lambda_L$, then, almost surely,
\begin{equation}\label{eq:linear functional prob upper bound chiconc}
\P(|(\varphi^{\eta,\lambda,\Lambda}, -\Delta_\Lambda s)|\ge r\mid \eta_{\Delta ^c})\le 3K\exp\bigg(-\frac{\kappa r^{2\alpha }}{4^\alpha f(\lambda )^\alpha \|\nabla s\|_\Lambda^{2\alpha }\ell^{\alpha \chi}}\bigg)
\end{equation}
where $\chi,\alpha,K,\kappa$ are of~\eqref{eq:chi conc}.
\end{enumerate}
\end{lemma}

\begin{proof}
    For brevity, we write $H^{\eta}$, $\varphi^\eta$ and $\GE^\eta$ for $H^{\eta,\lambda,\Lambda}$, $\varphi^{\eta,\lambda,\Lambda}$ and $\GE^{\eta,\lambda,\Lambda}$, respectively.
    
    By the main identity (Proposition~\ref{prop:main identity}), for each $\rho\in\R$, 
    \begin{equation}
        H^{\eta^{\rho s}}(\varphi^{\eta}+\rho s) - H^{\eta}(\varphi^{\eta}) = \rho(\varphi^{\eta},-\Delta_\Lambda s) + \frac{\rho^2}{2}\|\nabla s\|_\Lambda^2.
    \end{equation}
    It follows that for $\rho = -\frac{r}{\|\nabla s\|_\Lambda^2}$ we have the containment
    \begin{equation}
        \left\{(\varphi^{\eta}, -\Delta_\Lambda s)\ge r\right\} \subset\left\{H^{\eta^{\rho s}}(\varphi^{\eta}+\rho s) - H^{\eta}(\varphi^{\eta})\le -\frac{r^2}{2\|\nabla s\|_\Lambda^2}\right\}.
    \end{equation}
    Noting now that $\GE^{\eta}=H^{\eta}(\varphi^{\eta})$ while $\GE^{\eta^{\rho s}}\le H^{\eta^{\rho s}}(\varphi^{\eta}+\rho s)$ we conclude that
    \begin{equation}
    \left\{(\varphi^{\eta}, -\Delta_\Lambda s)\ge r\right\} 
    \subset\left\{\GE^{\eta^{\rho s}}\le \GE^{\eta}-\frac{r^2}{2\|\nabla s\|_\Lambda^2}\right\}.
    \end{equation}
    In particular, for any $\gamma\in\R$,
    \begin{equation}\label{eq:bound for +r}
    \left\{(\varphi^{\eta}, -\Delta_\Lambda s)\ge r\right\} \subset\left\{|\GE^{\eta^{\rho s}} - \gamma|\ge \frac{r^2}{4\|\nabla s\|_\Lambda^2}\right\}\cup\left\{|\GE^{\eta} - \gamma|\ge \frac{r^2}{4\|\nabla s\|_\Lambda^2}\right\}.
    \end{equation}
    Repeating the above derivation with $r$ replaced by $-r$, we also conclude that, for any $\gamma\in\R$,
    \begin{equation}\label{eq:bound for -r}
    \left\{(\varphi^{\eta}, -\Delta_\Lambda s)\le -r\right\} \subset\left\{|\GE^{\eta^{-\rho s}} - \gamma|\ge \frac{r^2}{4\|\nabla s\|_\Lambda^2}\right\}\cup\left\{|\GE^{\eta} - \gamma|\ge \frac{r^2}{4\|\nabla s\|_\Lambda^2}\right\}.
    \end{equation}

    Finally, our assumption~\ref{as:stat} implies that conditionally on $\eta_{\supp(s)^c}$, the ground energies $\GE^{\eta^{\rho s}}$, $\GE^{\eta^{-\rho s}}$ and $\GE^\eta$ all have the same distribution. Thus, inequality  ~\eqref{eq:lower bound on fluctuations} follows from~\eqref{eq:bound for +r} and~\eqref{eq:bound for -r} (noting that they hold for all $\gamma\in\R$).
    
    Now, inequality~\eqref{eq:linear functional prob upper bound conc} follows from~\eqref{eq:lower bound on fluctuations} by taking $\gamma=\E\big(\GE^{\eta}\mid \eta_{\supp(s)^c}\big)$ and applying assumption~\ref{as:conc} with $\rho = \frac{r^2}{4\lambda \|\nabla s\|_\Lambda^2\sqrt{|\supp(s)|}}$, while inequality~\eqref{eq:linear functional prob upper bound chiconc} follows from~\eqref{eq:lower bound on fluctuations} by taking $\gamma=\E\big(\GE^{\eta}\mid \eta_{\Delta^c}\big)$ and applying~\eqref{eq:chi conc} with $\rho = \frac{r^2}{4f(\lambda) \|\nabla s\|_\Lambda^2\ell^\chi}$.
\end{proof}

We remark that the idea of using conditional concentration inequalities as an ingredient in proving long-range order (similar to our upper bound on the height fluctuations) was used by Ding--Zhuang~\cite{ding2024long} in their recent work on the random-field Ising and Potts models.

\subsection{Overview of proof of main results}

\subsubsection{Localization and the scaling inequality $\chi\ge2\xi+d-2$}
We briefly overview the proof of the localization estimates for the minimal surface $\varphi^\eta$ and the scaling inequality $\chi\ge2\xi+d-2$, referencing the full details in Section~\ref{sec:loc}. For simplicity, we only discuss the $n=1$ case.

Dimension $d=1$ admits a simple approach: The pointwise localization estimate~\eqref{eq:localization estimate} at a vertex $v$, as well as the scaling inequality~\eqref{eq:one dimension pointwise lower bound}, follow from Lemma~\ref{lem:fluctuation and concentration} with a function $s$ whose Laplacian is $1$ on $v$ and $0$ everywhere else in $\Lambda $ (see the first part of Lemma~\ref{lem:46}). This approach also works in dimension $d=2$ to prove the pointwise scaling inequality~\eqref{eq:d=2 pointwise scaling relation}.

In higher dimensions, in order to obtain a pointwise bound at $v$, a multiscale approach is required (see Lemma~\ref{lem:loc}). In this proof we use Lemma~\ref{lem:fluctuation and concentration} with a sequence of functions $s_j$ such that the Laplacian of $s_j$ is supported on the two boundaries of an annulus at scale $2^j$ around $v$. This allows us to bound the difference between the averages of $\varphi^\eta$ on the boundary of a box of side length $2^j$ around $v$ and on the boundary of a box of side length $2^{j+1}$ around $v$. A bound on the height at $v$ then follows by a union bound over the scales.

In order to obtain a bound on the maximal height of $\varphi^\eta$ (as in, e.g.,~\eqref{eq:loglog},~\eqref{eq:locintro} and~\eqref{eq:std lower bound}), we first show that the surface satisfies a H\"older-type bound (Lemma~\ref{lem:holder} and the second part of Lemma~\ref{lem:46}). Namely, we bound $\varphi^\eta_u-\varphi^\eta_v$ in terms of a power of $\|u-v\|$. In dimension $d=1$ this follows from Lemma~\ref{lem:fluctuation and concentration} with a function whose Laplacian is $1$ on $u$, $-1$ on $v$ and $0$ elsewhere in $\Lambda$. In higher dimensions, this is done with a multiscale approach similar to the one explained above. Then, to bound the maximal height, we use a chaining argument (see Corollary~\ref{cor:21} and Corollary~\ref{cor:loc2}. This is reminiscent of the proof of Dudley's theorem on the maximum of a Gaussian process).

Finally, the scaling inequality~\eqref{eq:ge direction of scaling relation} involving the average height of $\varphi^\eta$ is derived in Section~\ref{sec:average height localization} from Lemma~\ref{lem:fluctuation and concentration} using a function $s$ whose Laplacian is constant in $\Lambda_L$.

\subsubsection{Delocalization and the scaling inequality $\chi\le2\xi+d-2$}

Section~\ref{sec:deloc} is devoted to the scaling inequality $\chi\le2\xi+d-2$ and delocalization results. To obtain a unified approach to the different results, the section begins with a deterministic abstract proposition (Proposition~\ref{prop:Deloc2}) from which the delocalization results are derived. The proposition relies on the main identity (Proposition~\ref{prop:main identity}): The basic idea is to show that the minimal energy among delocalized minimal surfaces is lower than the minimal energy among localized ones. This is derived by showing that the energy cost for shifting a localized minimal surface (to make it delocalized) may be outweighed by the energy saving from visiting a new environment.

Nevertheless, among our results, the proof of Theorem~\ref{theorem:other half of scaling relation} (the scaling inequality $\chi\le2\xi+d-2$) is simple enough that establishing it directly, without Proposition~\ref{prop:Deloc2}, may be enlightening. Thus, we have chosen to include a direct proof of Theorem~\ref{theorem:other half of scaling relation} here (up to a reference to the simple Proposition~\ref{prop:const_laplace}), while in Section~\ref{sec:the scaling inequality chi less than xi} we include a stronger result, Theorem~\ref{theorem:local other half of scaling relation}, establishing a kind of \emph{local} (applicable to subdomains) scaling inequality, proved with Proposition~\ref{prop:Deloc2}. Theorem~\ref{theorem:local other half of scaling relation} is used to derive the one-dimensional scaling inequality ~\eqref{eq:std upper bound} in Theorem~\ref{thm: scaling relation d=1} and the scaling inequality of Theorem~\ref{thm:scalingrelation<}.

\begin{proof}[Proof of Theorem \ref{theorem:other half of scaling relation}]
 For brevity, we omit $\lambda$ and $\Lambda_L$ from the notation, writing, e.g., $H^\eta$ for $H^{\eta,\lambda,\Lambda_L}$. Write 
\[\GE^\eta_h\coloneqq\inf_{\varphi: \max_{v\in\Lambda_L}|\varphi_v\cdot e|\le h }H^\eta(\varphi)\]
and let $\varphi^{\eta,h}$ be the corresponding minimizer (using~\ref{as:exiuni}).
Set $\zeta:=\eta[\Lambda_{\lfloor L/2\rfloor}] $. Using Proposition~\ref{prop:const_laplace}, let $s = \pi\cdot e$ with $\pi:\Z^d\to\R$ satisfying $|\pi_v\cdot e|\ge 3h$ on $\Lambda_{\lfloor L/2\rfloor}$, $\pi=0$ outside $\Lambda_L$, $\|\nabla \pi\|^2\le Ch^2L^{d-2}$ and $\sum_v|\Delta \pi_v|\le ChL^{d-2}$.
By the main identity (Proposition \ref{prop:main identity}), using that $\eta = (\eta^s)^{-s}$,
\begin{equation}\label{eq:sketch1}
    \GE^\eta\le H^\eta(\varphi^{\eta^s,h}-s)=   \GE^{\eta^s}_h+(\varphi^{\eta^s,h},\Delta s)+ \frac 12 \|\nabla s\|^2\le \GE^{\eta^s}_h+ 2Ch^2L^{d-2}
\end{equation}
and similarly (using $-s$ instead of $s$)
\begin{equation}\label{eq:sketch2}
    \GE^{\zeta^s}\le \GE^\zeta_h +2Ch^2L^{d-2}.
\end{equation}

Let $\mathcal F$ be the sigma algebra of the restriction of $\eta $ to $\Lambda_{\lfloor L/2\rfloor }^c\times\R^n$.
Conditionally on $\mathcal F $, the events
\[
A_1:=\{\GE^\eta_h\ge \GE^\zeta_h+5Ch^2L^{d-2}\}\quad\text{and}\quad A_2:=\{\GE^{\eta^s}_h\le \GE^{\zeta^s}_h\}
\]
are independent by \ref{as:indep}, since $A_1$ depends on $\eta$ and $\zeta$ restricted to $\Lambda_{\lfloor L/2\rfloor }\times \{t\colon |t\cdot e|\le h\}$ whereas $A_2$ depends on $\eta$ and $\zeta$  restricted to $\Lambda_{\lfloor L/2\rfloor}\times \{t\colon |t\cdot e|\ge 2h\}$ (using that $h\ge1$ and $|s_v\cdot e|\ge 3h$ on $\Lambda_{\lfloor L/2\rfloor})$. Conditionally on $\mathcal F$, we have by symmetry that $\P(A_2\,|\,\mathcal F)\ge1/2$.

On the event $A_1\cap A_2\cap \{\GE ^{\zeta^s}_h= \GE ^{\zeta^s}\}$,
we have, using \eqref{eq:sketch1} and \eqref{eq:sketch2},
\[\GE^\eta \le \GE^{\eta^s}_h+ 2Ch^2L^{d-2}\le \GE^{\zeta^s}_h+ 2Ch^2L^{d-2} \le  \GE^{\zeta}_h+ 4Ch^2L^{d-2}\le \GE^{\eta}_h -Ch^2L^{d-2}   <\GE^\eta_h\]
implying that $\max_{v\in\Lambda_L}|\varphi^\eta _v\cdot e|> h $.
It follows that
\begin{multline*}
\P(\max_{v\in\Lambda_L}|\varphi^\eta _v\cdot e|\ge h\,|\,\mathcal F)\ge \P(A_1\cap A_2\cap \{\GE ^{\zeta^s}_h= \GE ^{\zeta^s}\}\,|\,\mathcal F)\\
\ge \P(A_1\cap A_2\,|\,\mathcal F) - \P(\max_{v\in\Lambda_L}|\varphi^{\zeta^s} _v\cdot e|\ge h\,|\,\mathcal F) \ge \frac{1}{2}\P(A_1\,|\,\mathcal F) - \P(\max_{v\in\Lambda_L}|\varphi^{\zeta^s} _v\cdot e|\ge h\,|\,\mathcal F)
\end{multline*}
The result follows by taking the expectation of both sides and using the fact that $\zeta^s$ has the same distribution as $\eta$ by~\ref{as:stat} and the definition of $\zeta$.
\end{proof}

Theorem \ref{theorem:other half of scaling relation} can be used to establish an analogue of Theorem~\ref{thm:delocalization intro} 
in which the \emph{maximum height} of the minimal surface is shown to delocalize. This analogue will follow from a lower bound on the fluctuations of $\GE^\eta-\GE^{\eta[\Lambda_{L/2}]}$. We sketch two different strategies to obtain such a lower bound for the disorder $\eta^\white$.

First strategy: As the goal is to show delocalization, we may assume, to obtain a contradiction, that the maximum height of the minimal surface does not exceed a given $h$ with high probability. Now, consider perturbing the disorder $\eta$ by adding a constant $\varepsilon$ to all the disorder variables $\eta_{v,t}$ with $(v,t)\in \Lambda_{L/2}\times[-h,h]^n$ (also changing some variables outside this range, due to the bump function in the definition of $\eta^\white$). This perturbation does not change significantly the disorder distribution, as long as the order of $\varepsilon$ does not exceed $\lambda (L^dh^n)^{-1/2}$. As this additive perturbation necessarily adds $L^d\varepsilon$ to the energy of every surface with maximal height less than $h$, one may conclude that the typical fluctuations of $\GE^\eta-\GE^{\eta[\Lambda_{L/2}]}$ are at least of order $\lambda \sqrt{L^d/h^n}$ (this is also the idea behind Theorem~\ref{thm:lower bounds on energy fluctuations by localization}). However, our assumption that the maximum height is localized to height $h$, implies, by Theorem~\ref{theorem:other half of scaling relation}, that the typical fluctuations of $\GE^\eta-\GE^{\eta[\Lambda_{L/2}]}$ are at most of order $h^2 L^{d-2}$. To avoid a contradiction, we conclude that the maximum height of the minimal surface delocalizes at least to order $\lambda^{\frac{2}{4+n}}L^{\frac{4-d}{4+n}}$.

Second strategy: Here, we obtain a lower bound on the fluctuations of $\GE^\eta-\GE^{\eta[\Lambda_{L/2}]}$ by additively decreasing $\eta$ on a tube of width $1$ around the actual minimal surface. As the minimal surface is at a random location, correlated with the disorder, the exact amount of additive perturbation that is possible without significant change to the disorder distribution seems difficult to study. Nevertheless, using that this perturbation preserves the minimal surface and that the energy of the minimal surface on $\Lambda_L$ is typically of order $|\Lambda_L|$, we are able to use this strategy to obtain a lower bound of order $c_\lambda$ (i.e., depending only on $\lambda$) on the typical fluctuations of $\GE^\eta-\GE^{\eta[\Lambda_{L/2}]}$. This leads to a non-trivial delocalization result in dimension $d=1$, translating via Theorem~\ref{theorem:other half of scaling relation} to the delocalization of the maximum height to at least order $c_\lambda' \sqrt{L}$ (for all codimensions $n$).

The actual proof of Theorem~\ref{thm:delocalization intro} in Section~\ref{sec:delocalization in dimension 4} and Section~\ref{sec:delocalization in low dimensions} differs from the above description in several ways. First, as we show that a positive fraction of the heights exceed $h$ (instead of just the maximum), we cannot rely on Theorem~\ref{theorem:other half of scaling relation} and instead use Proposition~\ref{prop:Deloc2} directly. Second, our perturbations, while following the above two strategies, are made to the white noise underlying $\eta^\white$ rather than to $\eta^\white$ itself; the perturbations and the resulting energy fluctuation lower bounds are described in Section~\ref{sec:two specific mappings} and Proposition~\ref{lem:bounty_construction}. Lastly, our proof of delocalization in the critical dimension $d=4$ (Section~\ref{sec:delocalization in dimension 4}) uses additional arguments, employing fractal percolation.

\section{Localization}\label{sec:loc}

In this section we prove most of the localization results from the introduction. Throughout this section we often use the Green's function $G_\Lambda ^v$ which is defined as follows. Let $\Lambda\subset \mathbb Z^d$. Let $\partial\Lambda$ denote the exterior boundary of $\Lambda$, that is $\partial \Lambda:=\Lambda^+\setminus \Lambda$. The Green's function in the domain $\Lambda $ for $v\in\Lambda$ is defined by 
\[\forall x\in\Z^ d\qquad G_\Lambda^v(x):=\frac{1}{2d} \cdot \E_x\left[ \big| \big\{ t\in [0, \tau _{\Lambda }] : X_t=v\big\} \big| \right],\]
where $(X_t)_{t\ge0}$ is a simple discrete-time random walk on $\Z^d$ with $X_0=x$ and $\tau _\Lambda $ is the first exit time of $\Lambda $,  $\tau_{\Lambda }:=\min\{t\ge 0:X_t\notin \Lambda \}$. When $\Lambda =\Lambda _L$ we write $G_L ^v:=G_{\Lambda _L}^v$. Note that $\Delta G_\Lambda ^v (v)=-1$ and $\Delta G_\Lambda ^v (u)=0$ for any  $u\notin \partial \Lambda \cup \{v\}$.

Throughout this section, for brevity, we write $H^{\eta}$, $\varphi^\eta$ and $\GE^\eta$ for $H^{\eta,\lambda,\Lambda}$, $\varphi^{\eta,\lambda,\Lambda}$ and $\GE^{\eta,\lambda,\Lambda}$, respectively.

\subsection{Localization when $d\in\{1,2\}$}
In this section we prove the lower bound \eqref{eq:std lower bound} in Theorem~\ref{thm: scaling relation d=1} and the pointwise estimates \eqref{eq:one dimension pointwise lower bound} and \eqref{eq:d=2 pointwise scaling relation} in Theorem~\ref{theorem:half scaling relation}.

We will need the following estimates on the Green's function whose proofs are postponed to Appendix \ref{appendix:greenfunction}. 

\begin{lem}\label{lem:88}
    Suppose that $d=1$. Then, for any $v\in \Lambda _L$ we have that $G_{\Lambda _L}^v(v)\le C r_v$. Moreover, for any pair of vertices $u,v\in \Lambda _L$ we have that $0\le G^v_L (v)-G^v_L (u)\le C|u-v|$.  
\end{lem}

\begin{lem}\label{lem:89}
    Suppose that $d=2$. Then, for any $v\in \Lambda _L$ we have that $G^v_L (v) \le C \log (1+r_v)$.
\end{lem}

Lemma~\ref{lem:46} and Lemma~\ref{lem:47} below correspond, respectively, to dimensions $d=1$ and $d=2$. They imply~\eqref{eq:one dimension pointwise lower bound} and~\eqref{eq:d=2 pointwise scaling relation} by choosing $\gamma=\med(\GE^{\eta})$. 

\begin{lem}[\ref{as:exiuni}+\ref{as:stat}]\label{lem:46}
Suppose that $d=1$. There exists $c>0$ such that the following holds for any $\gamma \in \mathbb R$ and any unit vector $e\in \mathbb R ^n$:
\begin{enumerate}
    \item 
    For all $v\in \Lambda _L$ we have  
    \begin{equation}
        \mathbb P \big( | \varphi^ \eta _v  \cdot e| \ge h \big)\le 3 \mathbb P \big( |\GE ^\eta -\gamma | \ge c h^2/r_v \big).
        \end{equation}
        \item 
        For all $u,v\in \Lambda _L$ we have
\begin{equation}
        \mathbb P \big( | (\varphi^ \eta _v-\varphi^ \eta _u)\cdot e | \ge h \big)\le 3 \mathbb P \big( |\GE ^\eta -\gamma | \ge c h^2/|u-v| \big).
\end{equation}
\end{enumerate}
\end{lem}

\begin{proof}[Proof of Lemma \ref{lem:46}]
Define the function $s(x):=G_{\Lambda _L}^v(x)e$. We would like to use Lemma~\ref{lem:fluctuation and concentration} with the function $s$. To this end, note that $\Delta s(v)=-e$ and that $\Delta s(w)=0$ for all $w\in \Lambda _L \setminus \{v\}$. Thus, using that $\varphi^ \eta |_{\Lambda _L^c}=0$ we obtain that $\varphi^ \eta _v\cdot e=(\varphi^ \eta ,-\Delta s)$. Moreover, using that $s|_{\Lambda _L^c}=0$ we have  
\begin{equation}
        \|\nabla s\|_{\Lambda _L}^2=(\nabla s,\nabla s)_{\Lambda _L} =(s,-\Delta s)=G_{\Lambda _L}^v(v)\le Cr_v,
    \end{equation}
    where the last inequality is by Lemma~\ref{lem:88}. Thus, by Lemma~\ref{lem:fluctuation and concentration}, for all $h\ge 0$ and $\gamma\in\R$
    \begin{equation}
        \mathbb P \big( |\varphi^ \eta _v\cdot e | \ge h \big)=\mathbb P \big( |(\varphi^ \eta ,-\Delta s) | \ge h \big)\le 3\mathbb P \big( \big| \GE^\eta-\gamma \big| \ge ch^2/r_v \big).
        \end{equation}

    We turn to prove the second part. To this end, we use Lemma~\ref{lem:fluctuation and concentration} with the function $s(x):=(G_{\Lambda _L}^v(x)-G_{\Lambda _L}^{u}(x))e$. We have that $\Delta s(v)=-e$, $\Delta s(u)=e$ and $\Delta s(w)=0$ for all $w\in \Lambda _L \setminus \{u,v\}$. Thus, we obtain 
    $(\varphi^ \eta _v-\varphi^ \eta _u)\cdot e=(\varphi^ \eta ,-\Delta s)$ using again that $\varphi^ \eta |_{\Lambda _L^c}=0$.
    Moreover, by Lemma~\ref{lem:88} we have
    \begin{equation}
        \|\nabla s\|_{\Lambda _L}^2=(s,-\Delta s)=(G_{\Lambda _L}^v(v)-G_{\Lambda _L}^u(v))-(G_{\Lambda _L}^v(u)-G_{\Lambda _L}^u(u))\le C|u-v|.
    \end{equation}
    Thus, by Lemma~\ref{lem:fluctuation and concentration}, for all $h\ge 0$
    \begin{equation}
        \mathbb P \big( \big|(\varphi^ \eta _v-\varphi^ \eta _u)\cdot e\big| \ge h\big)\le 3\mathbb P \big( \big| \GE^\eta-\gamma \big| \ge ch^2/|u-v| \big)
    \end{equation}
    as needed.
\end{proof}

Similarly, we obtain the following lemma in dimension $d=2$.

\begin{lem}[\ref{as:exiuni}+\ref{as:stat}]\label{lem:47}
Suppose that $d=2$. There exists $c>0$ such that for any $\gamma \in \mathbb R$, a unit vector $e\in \mathbb R ^n$ and $v\in \Lambda _L$ we have 
    \begin{equation}
        \mathbb P \big( | \varphi^ \eta _v  \cdot e| \ge h \big)\le 3 \mathbb P \big( |\GE ^\eta -\gamma | \ge c h^2/\log (1+r_v) \big).
    \end{equation}
\end{lem}

The proof of Lemma~\ref{lem:47} is almost identical to the proof of Lemma~\ref{lem:46}. The only difference is that when $d=2$ the function $s(x):=G^v_{\Lambda _L}(x)e$ satisfies $\|\nabla s\|_{\Lambda _L}^2 =G_{\Lambda _L}^v(v) \le C \log (1+r_v)$ by Lemma~\ref{lem:89}.

The lower bound in Theorem~\ref{thm: scaling relation d=1} follows from the next corollary using the fact that $\max _{L-K\le v\le L } |\varphi^ \eta _v\cdot e |$ has the same distribution as $\max _{-L\le v\le -L+K } |\varphi^ \eta _v\cdot e |$ for any $K\le 2L$ and any vector $e\in \mathbb R ^n$.

\begin{cor}[\ref{as:exiuni}+\ref{as:stat}]\label{cor:21}
    Suppose that $d=1$ and let $\sigma _L:=\sqrt{{\rm Var} (\GE ^\eta )}$. Then, for any unit vector $e\in \mathbb R ^n$ and $K\le 2L$ we have 
    \begin{equation}\label{eq:jf}
        \mathbb E \Big[ \max _{L-K\le v\le L } |\varphi^ \eta _v\cdot e | \Big] \le C\sqrt{K\sigma _L}.
    \end{equation}
\end{cor}

\begin{proof}
    Fix a unit vector $e\in \mathbb R ^n$. Let $m:=\lfloor \log _2K \rfloor +1$. For any $h>0$ define the event 
    \begin{equation}
        \mathcal E _h :=\bigcap _{j=0}^m \big\{ \forall u\in 2^j \mathbb Z \cap [L-K,\infty ), \ \   |(\varphi^ \eta _u-\varphi^ \eta _{u+2^j}) \cdot e | \le h2^{-4-(m-j)/6} \big\}.
    \end{equation}
    First, we claim that on the event $\mathcal E _h$ we have $\max_{v\ge L-K} |\varphi^ \eta _v\cdot e | \le h$. Indeed, let $L-K\le v \le L$ and note that there is a sequence of vertices $v=v_0\le v_1\le \cdots \le v_{\ell +1}$ with $\ell \le m$ such that  $v_{\ell +1}\notin \Lambda _L$ and for all $0\le j\le \ell $ we have $v_j\in 2^j\mathbb Z$ and $v_{j+1}-v_j \in \{0,2^j\}$. Thus, using that $\varphi^ \eta _{v_{\ell +1}}=0$ we obtain that on $\mathcal E _h$ 
    \begin{equation}
        |\varphi^ \eta _v\cdot e | \le \sum _{j=0}^\ell  |(\varphi^ \eta _{v_{j+1}}-\varphi^ \eta _{v_j})\cdot e | \le h \sum _{j=0}^\ell  2^{-4-(m-j)/6}\le h.
    \end{equation}
   Thus, by Lemma~\ref{lem:46} we have
   \begin{equation}
   \begin{split}
       \mathbb P &\Big( \max _{L-k\le v\le L} |\varphi^ \eta _v\cdot e | \ge h \Big) \le \mathbb P (\mathcal E _h^c ) \le \sum _{j=0}^m \sum _{\substack{u\in 2^j\mathbb Z \\ |u-L|\le 4K}} \mathbb P \big(  |( \varphi^ \eta _u-\varphi^ \eta _{u+2^j})\cdot e| \ge h2^{-4-(m-j)/6} \big) \\
       &\le C\sum _{j=0}^m  K2^{-j} \mathbb P \big( |\GE ^\eta -\mathbb E[\GE^\eta ] |\ge ch^2 2^{-j-(m-j)/3} \big) 
       \le C\frac{K^{5/3} \sigma _L^2}{h^4}\sum _{j=0}^m  2^{j/3} \le C\frac{K^2\sigma _L^2}{h^4},
   \end{split}
   \end{equation}
    where in the fourth inequality we used Chebyshev's inequality. Note that the last bound is nontrivial only when $h\ge C\sqrt{K\sigma _L}$.  Integrating the last inequality over $h\ge \sqrt{K\sigma _L}$ finishes the proof of the corollary.
\end{proof}

\subsection{Pointwise height and maximal height localization}
In this section we prove Theorem~\ref{thm:localization} and Theorem~\ref{thm:scaling relation>}. Throughout this section the constants $C$ and $c$ are allowed to depend on $d$. When we make the assumption \eqref{eq:chi conc} the constants are also allowed to depend on $\chi , \alpha , K $ and $\kappa$ from the assumption.

We start with an estimate on the Green's function whose proof is given in Appendix \ref{appendix:greenfunction}. As usual, for a box $\Lambda $ and $v\in \Lambda $ we write $r_v=d(v,\Lambda ^c)$.

\begin{lem}\label{lem:Green}
    Let $\Lambda :=[a_1,b_1]\times \cdots \times [a_d,b_d]$ be a box. Then:
    \begin{enumerate}
        \item 
        For all $v,x\in \Lambda $ such that $r_v \le 2\|v-x\|$ we have 
        \begin{equation}
           G_\Lambda ^v(x)\le Cr_xr_v \|x-v\|^{-d}. 
        \end{equation}
        \item 
        For all $v,u,x\in \Lambda $ such that $r_v\le 2\|x-v\|$ we have 
        \begin{equation}
           |G_\Lambda ^v(x)-G_\Lambda ^u(x) |\le Cr_x\|u-v\|\cdot \|x-v\|^{-d}. 
        \end{equation}
    \end{enumerate}
\end{lem}
The two following lemmas will be the two key inputs to prove Theorem~\ref{thm:scaling relation>}.

\begin{lem}[\ref{as:exiuni}+\ref{as:stat}]\label{lem:loc}
    Suppose that \eqref{eq:chi conc} holds for some $\chi < d$. Then, the following holds for any $v\in \Lambda _L$, a unit vector $e\in \mathbb R ^n$ and $t>0$:
    \begin{itemize}
        \item 
    Suppose that $\chi >d-2$ and let $\xi :=(\chi -d+2)/2$ according to the scaling relation \eqref{eq:scaling relation}. Then
    \begin{equation}\label{3}
        \mathbb P \big( |\varphi^ \eta _v\cdot e| \ge t \sqrt{f(\lambda)}r_v^{\xi}  \big) \le Ce^{-ct^{2\alpha}}.
    \end{equation}
    \item
    Suppose that $\chi =d-2$. Then 
    \begin{equation}\label{4}
        \mathbb E |\varphi^ \eta _v\cdot e| \le C\sqrt{f(\lambda)}\log r_v.
    \end{equation}
    \item 
    Suppose that $\chi <d-2$. Then 
    \begin{equation}\label{5}
        \mathbb P \big( |\varphi^ \eta _v\cdot e| \ge t \sqrt{f(\lambda)} \big) \le Ce^{-ct^{2\alpha}}.
    \end{equation}
    \end{itemize}
\end{lem}

Note that the estimates \eqref{eq:localization estimate} from Theorem~\ref{thm:localization} and \eqref{eq:pointwise localization under exponent assumption} from Theorem~\ref{thm:scaling relation>} follow immediately from Lemma~\ref{lem:loc} using a union bound over the coordinates of $\varphi _v^\eta $. 

Next, we obtain the following bound on the height difference between pairs of points. 

\begin{lem}[\ref{as:exiuni}+\ref{as:stat}]\label{lem:holder}
    Suppose that \eqref{eq:chi conc} holds for some $d-2 < \chi < d$. Let $\xi :=(\chi -d+2)/2$. Then, for all $v,u\in \mathbb Z ^d$, a unit vector $e\in \mathbb R ^n$ and $t>0$ we have  
    \begin{equation}\label{eq:holder}
        \mathbb P \big( | (\varphi^ \eta _v-\varphi^ \eta _u)\cdot e | > t\sqrt{f(\lambda)} \|v-u\|^{\xi}  \big) \le Ce^{-ct^{2\alpha }}.
    \end{equation}
\end{lem}

Note that the last statement is not empty only when either $u$ or $v$ are in $\Lambda _L$.

\begin{proof}[Proof of Lemma~\ref{lem:loc}]
    Let $v\in \Lambda_L $ and let $e\in \mathbb R ^n$ be a unit vector. Define the boxes $\Lambda _j:=\Lambda_L  \cap (v+(-2^j,2^j)^d)$ and let $m$ be the first integer for which $\Lambda _m=\Lambda_L $. Let $G_j^v(x)=G_{\Lambda _j}^v(x)$ and for $j>1$ let $s_j(x):=(G_j^v(x)-G_{j-1}^v(x) )e$. Finally, let $s_1(x):=G_1^v(x)e$. We have that $G_m^ve=\sum _{j=1}^ms_j$ and therefore, using that $\Lambda _m=\Lambda _L$ and that  $\varphi^ \eta |_{\Lambda _L^c}=0$ we obtain 
    \begin{equation}\label{eq:phi i}
        \varphi^ \eta _v\cdot e=(\varphi ^\eta ,-\Delta G_m^ve)= \sum _{j=1}^m (\varphi^ \eta ,- \Delta s_j).
    \end{equation}

    Next, let $m':=\lfloor \log _2r_v \rfloor -1$. We claim that for all $r>0$ 
    \begin{equation}\label{1}
        \mathbb P \big( | (\varphi^ \eta , \Delta s_j)| \ge r \big) \le C\begin{cases}\exp \big( -c \big( r^{2}2^{j (d-2-\chi)}/f(\lambda) \big)^ {\alpha }\big) & \ \, 1\le j\le m' \\ \exp \big( -c\big( r^{2}r_v^{-2}2^{j (d-\chi) } /f(\lambda) \big) ^ {\alpha }\big) & m'<j \le m\end{cases}.  
    \end{equation}
    To this end, we estimate the Dirichlet energy of $s_j$ and use Lemma~\ref{lem:fluctuation and concentration}. For all $1<j\le m$ we have that 
    \begin{equation}\label{eq:laplacian norm}
        \|\nabla s_j\|^2_{\Lambda } = (s_j,-\Delta s_j)=\big( G_j^v-G_{j-1}^v,\Delta (G_j^v-G_{j-1}^v) \big)=-\!\!\! \sum _{x\in \partial \Lambda _{j-1}\setminus\partial \Lambda } \! G_j^v(x) \Delta G_{j-1}^v(x),
    \end{equation}
    where in the last equality we used that $\Delta (G_j^v-G_{j-1}^v)$ is supported on $\partial \Lambda _j \cup \partial \Lambda _{j-1}$ and that $G_j$ is supported on $\Lambda _j$ and $G_{j-1}$ is supported on $\Lambda _{j-1} $. 

    Next, we use Lemma~\ref{lem:Green} in order to bound the right hand side of \eqref{eq:laplacian norm}. We consider separately the cases $1<j\le m'$ and $m'<j \le m$. When $1<j\le m'$, by Lemma~\ref{lem:Green}, for all $x\in \partial \Lambda _{j-1} $ we have $|G_j^v(x)| \le C2^{j(2-d)}$ and $|\Delta G_{j-1}(x)| = | G_{j-1}(y)|\le C2^{j(1-d)}$, where $y\sim x$ is the unique neighbour of $x$ in $\Lambda _{j-1}$. Similarly, when $m'<j\le m$, for all $x\in \partial \Lambda _{j-1}\setminus \partial \Lambda $ we have $|G_j^v(x)| \le Cr_v2^{j(1-d)}$ and $|\Delta G_{j-1}(x)| = | G_{j-1}(y)|\le Cr_v2^{-dj}$. Substituting these estimates into \eqref{eq:laplacian norm} we obtain
    \begin{equation}
        \|\nabla s_j\|_\Lambda ^2 \le C\begin{cases}2^{j(2-d)} & \ \, 1\le j\le m' \\ r_v^22^{-dj} & m'<j \le m\end{cases},  
    \end{equation}
    where here it is clear that this estimate also holds in the case $j=1$. Thus, using Lemma~\ref{lem:fluctuation and concentration} and the fact that $\text{Supp}(s_j)\subset \Lambda_j$ we obtain the estimate \eqref{1}.
    
    Next, suppose that $d-2<\chi < d$ and recall that $\xi := (\chi+2-d)/2$ and hence $0<\xi<1 $. Let $\delta :=\frac{1}{2}\min(\xi ,1-\xi )$ and define the event 
    \begin{equation}
        \mathcal E := \bigcap _{j=1}^m \big\{  |(\varphi^ \eta ,\Delta s_j)| \le t\sqrt {f(\lambda)} r_v^{\xi} 2^{-\delta |j-m'|}  \big\}.
    \end{equation}
    Clearly, on the event $\mathcal E $ we have that $|\varphi^ \eta _v\cdot e|\le Ct\sqrt {f(\lambda)}r_v^{\xi}$ as most of the contribution to the sum in \eqref{eq:phi i} will come from constantly many terms corresponding to $j\approx m'$. Moreover, by \eqref{1} we have 
    \begin{equation*}
    \begin{split}
        \mathbb P (\mathcal E ^c )&\le \sum _{j=1}^{m'} \exp \big( -c\big( t^{2}r_v^{2 \xi}2^{-2j\xi}2^{-2\delta (m'-j)} \big)^\alpha \big)+
        \sum _{j=m'+1}^m \exp \big( -c\big( t^{2} r_v^{2(\xi-1)} 2^{-2j(\xi-1)-2\delta (j-m')} \big)^\alpha  \big) \\
        &\le \sum _{j=1}^{m'} \exp \big( -c\big( t^{2} 2^{2\xi (m'-j) -2\delta (m'-j)}  \big)^\alpha \big)+
        \sum _{j=m'+1}^m \exp \big( -c\big( t^{2 } 2^{2(1-\xi )(j-m')-2\delta (j-m')} \big)^\alpha  \big) \\
        &\le \sum _{j=1}^{m'} \exp \big( -ct^{2\alpha} 2^{\delta \alpha (m'-j)} \big)+
        \sum _{j=m'+1}^m \exp \big( -ct^{2\alpha} 2^{\delta \alpha (j-m')} \big)
        \le Ce^{-ct^{2\alpha}},
    \end{split}
    \end{equation*}
     where in the second inequality we substituted $2^{m'+1}\ge r_v\ge 2^{m'}$. In here the constants $C,c$ may depend additionally on $\chi, \alpha $ and $\kappa $ from assumption~\ref{eq:chi conc}. This finishes the proof of \eqref{3}.
     
   When $\chi=d-2$ that is $\xi=0$, by \eqref{1} we have that $\mathbb E |(\varphi^ \eta , \Delta s_j)|\le C\sqrt {f(\lambda)} $ for $j\le m'$ and $\mathbb E |(\varphi^ \eta , \Delta s_j)|\le C\sqrt {f(\lambda)} 2^{m'-j}$ for $j> m'$. By \eqref{eq:phi i}, it follows that $\mathbb E |\varphi^ \eta _v\cdot e| \le C\sqrt {f(\lambda)} \log r_v$. This finishes the proof of \eqref{4}.

    Finally, suppose that $\chi<d-2$, that is $\xi<0$. Define the event 
    \begin{equation}
        \mathcal E := \bigcap _{j=1}^m \big\{  |(\varphi^ \eta ,\Delta s_j)| \le t\sqrt {f(\lambda)} 2^{\xi j/2}  \big\}.
    \end{equation}
    As in the case $\xi>0$, on $\mathcal E$ we have that $|\varphi^ \eta _v\cdot e|\le Ct\sqrt {f(\lambda)} $ and by \eqref{1} we have 
    \begin{equation*}
        \mathbb P (\mathcal E ^c )\le \sum _{j=1}^{m} \exp \big( -c\big( t^{2}2^{-\xi j} \big)^\alpha \big)\le Ce^{-ct^{2\alpha}}.
    \end{equation*}
    Note that in \eqref{1}, the bound when $j>m'$ is better than $\exp \big( -c \big( r^{2}2^{j (d-2-\chi)}/f(\lambda) \big)^ {\alpha }\big)$. However, this improvement is not needed in the case $\xi <0$.
\end{proof}

\begin{proof}[Proof of Lemma~\ref{lem:holder}]
    The proof is similar to the proof of Lemma~\ref{lem:loc} and some of the details are omitted. We may assume that $u,v\in \Lambda_L $ and that $\|u-v\|\le r_v/10$. Indeed, if $\|u-v\|\ge  r_v/10$ then the inequality in \eqref{eq:holder} follows from Lemma~\ref{lem:loc}. 

    Next, as in the proof of Lemma~\ref{lem:loc} we let $\Lambda _j:=\Lambda_L  \cap (v+(-2^j,2^j)^d)$ and let $m$ be the first integer for which $\Lambda _{m}=\Lambda _L$. We denote $G_j^v:=G_{\Lambda _j}^v$ and $G_j^u:=G_{\Lambda _j}^u$. Finally, let $m':=\lceil  \log _2 \|u-v\| \rceil +2$ and note that both $u,v$ are in $\Lambda _{m'}$. By the same argument as in the proof of Lemma~\ref{lem:loc} we have that 
    \begin{equation}\label{eq:839}
        \mathbb P \big( \big| (\varphi^ \eta , (\Delta  G_{m'}^v)e ) \big| \ge t \sqrt {f(\lambda)} \|u-v\|^{\xi} \big)  \le Ce^{-ct^{2\alpha }}
    \end{equation}
    and 
    \begin{equation}\label{eq:840}
        \mathbb P \big( \big| (\varphi^ \eta , (\Delta  G_{m'}^u)e ) \big| \ge t \sqrt {f(\lambda)} \|u-v\|^{\xi} \big)  \le Ce^{-ct^{2\alpha }}.
    \end{equation}
    Indeed, this is done by replacing the box $\Lambda _L$ with the box $\Lambda _{m'}$ and using the fact that both $d(u,\partial \Lambda _{m'})$ and $d(v,\partial \Lambda _{m'})$ are of order $\|u-v\|$ (the estimates \eqref{eq:839} and \eqref{eq:840} correspond to \eqref{3} since $|\varphi _v^\eta \cdot e|=
| (\varphi ^\eta ,(\Delta G_m^v)e)|$). 

    Define the functions $s_j$ for $m'\le j \le m$ by 
    \begin{equation}
        s_j(x):=\begin{cases}
        (G_{m'}^v(x)-G_{m'}^u(x))e, \quad \quad \quad  &j=m'\\
            (G_j^v(x)-G_j^u(x))e-(G_{j-1}^v(x)-G_{j-1}^u(x))e, \quad \quad &j>m'
        \end{cases}.
    \end{equation}
 We have that 
    \begin{equation}
        (\varphi^ \eta _v-\varphi^ \eta _u)\cdot e =- \sum _{j=m'}^{m} (\varphi^ \eta ,\Delta s_j).
    \end{equation}

    For all $m'<j \le m$ we have 
    \begin{equation}
        \|\nabla s_j\| _\Lambda ^2 = - \sum _{x\in \partial \Lambda _{j-1}\setminus\partial \Lambda } \! (G_j^v(x)-G_j^u(x)) \Delta (G_{j-1}^v(x)-G_{j-1}^u(x))\le C\|u-v\|^22^{-jd}
    \end{equation}
    where in here we used that, by Lemma~\ref{lem:Green}, for all $x\in \partial \Lambda _{j-1}\setminus\partial \Lambda_L$ we have $|G_j^v(x)-G_j^u(x)|\le C\|u-v\|2^{j(1-d)}$ and $|\Delta (G_{j-1}^v(x)-G_j^u(x))|=| G_{j-1}^v(y)-G_j^u(y)|\le C\|u-v\|2^{-jd}$ where $y$ is the unique neighbour of $x$ in $\Lambda _{j-1}$. Thus, letting $\delta :=(1-\xi )/2$ and using that $\text{Supp}(s_j)\subseteq \Lambda _j$ and Lemma~\ref{lem:fluctuation and concentration} we obtain
    \begin{equation}\label{eq:838}
        \mathbb P \big( |(\varphi^ \eta ,\Delta s_j)| \ge t\sqrt {f(\lambda)}\|u-v\|^{\xi} 2^{-\delta (j-m')} \big) \le C\exp \big(-c t^{2\alpha }2^{\alpha \delta (j-m')}  \big).
    \end{equation}
    On the intersection of the complements of the events in \eqref{eq:839}, \eqref{eq:840} and \eqref{eq:838} for all $m'<j\le m$, we have that $|(\varphi^ \eta _v-\varphi^ \eta _u)\cdot e|\le Ct\sqrt {f(\lambda)}\|u-v\|^{\xi}$.
\end{proof}

In the following corollary we establish the bounds \eqref{eq:locintro} from Theorem~\ref{thm:localization} and \eqref{eq:log chi concentration} from Theorem~\ref{thm:scaling relation>}.

\begin{cor}[\ref{as:exiuni}+\ref{as:stat}]\label{cor:loc2}
    Suppose that \eqref{eq:chi conc} holds for some $d-2< \chi < d$. Then, the following holds for any $v\in \Lambda _L$, a unit vector $e\in \mathbb R ^n$ and $t>0$:
\begin{equation}
    \P\left(\forall v\in\Lambda_L, \   | \varphi^{\eta}_v\cdot e |\le t\sqrt{f(\lambda)}\left(\log \frac{2L}{r_v}\right)^ {1/(2\alpha)} r_v^{\xi } \right)\ge 1-Ce^ {-ct^{2\alpha}}.
\end{equation}
\end{cor}

\begin{proof}
    The proof is similar to the proof of Corollary~\ref{cor:21}. Let $m:=\lfloor \log _2 L \rfloor +1$ and $t>0$ sufficiently large. Define the event
    \begin{equation}
        \mathcal E := \bigcap _{j=1}^{m} \Big\{ \forall u,v\in 2^j \mathbb Z ^d , \|u-v\|_\infty \le 2^{j}, \  |(\varphi^ \eta _u -\varphi^ \eta _v)\cdot e| \le t\sqrt {f(\lambda)} 2^{j\xi } (m-j+1)^{1/(2\alpha )} \Big\}.
    \end{equation}
    By Lemma~\ref{lem:holder} and a union bound over the pairs of vertices $u,v$ in the event $\mathcal E$ for which either $u$ or $v$ are in $\Lambda _L$ we have that 
    \begin{equation}
        \mathbb P (\mathcal E ^c) \le C\sum  _{j=1}^m (L2^{-j})^d e^{-ct^{2\alpha }(m-j+1)}\le Ce^{-ct^{2\alpha }},
    \end{equation}
    where the last inequality holds when $t$ is sufficiently large. Next, we let $v\in \Lambda _L$ and bound $|\varphi^ \eta _v\cdot e|$ on the event $\mathcal E$. Without loss of generality, suppose that all the coordinates of $v$ are non-negative. Let $k:=\lceil \log _2r_v \rceil $. There exists a sequence of vertices $v=v_0,v_1,\dots ,v_k$ that is non decreasing in each coordinate such that the following holds. For all $j<k$ we have that $v_j\in 2^j\mathbb Z ^d$ and $v_{j+1}-v_j\in \{ 0,2^j \}^d$ and $v_k\notin \Lambda _L$. Thus, using that $\varphi^ \eta _{v_k}=0$ we have on the event $\mathcal E $
    \begin{equation}
        | \varphi^ \eta _v \cdot e | \le \sum _{j=1}^k |(\varphi^ \eta _{v_{j}}-\varphi^ \eta _{v_{j-1}} )\cdot e| \le t\sum _{j=1}^k 2^{j\xi } (m-j+1)^{1/(2\alpha )} \le Ct r_v^{\xi }  (m-k+1)^{1/(2\alpha )}
    \end{equation}
    as needed.
\end{proof}

In the following corollary we establish the bounds \eqref{eq:loglog} from Theorem~\ref{thm:localization} and \eqref{eq:loglog chi concentration} of Theorem~\ref{thm:scaling relation>}.

\begin{cor}[\ref{as:exiuni}+\ref{as:stat}]
    Suppose that $d=1$ and that \eqref{eq:chi conc} holds for some $-1< \chi < 1$. Then, the following holds for any $v\in \Lambda _L$, a unit vector $e\in \mathbb R ^n$ and $t>0$:
\begin{equation}
    \P\left(\forall v\in\Lambda_L, \   | \varphi^{\eta}_v\cdot e |\le t\sqrt{f(\lambda)}\left(\log \log  \frac{3L}{r_v}\right)^ {1/(2\alpha)} r_v^{\xi } \right)\ge 1-Ce^ {-ct^{2\alpha}}.
\end{equation}
\end{cor}
 
\begin{proof}
  The proof is almost identical to the proof of Corollary~\ref{cor:loc2}.  Let $m:=\lfloor \log _2 L \rfloor $ and $t>0$ sufficiently large. Define the event
    \begin{equation*}
        \mathcal E := \bigcap _{j=1}^{m} \Big\{ \forall u,v\in 2^j \mathbb Z  , |u-v|= 2^{j}, \  |(\varphi^ \eta _u -\varphi^ \eta _v)\cdot e| \le t\sqrt {f(\lambda)}\max(r_v,r_u)^ {\frac \xi2} 2^{j\frac{\xi}{2}} \log (m-j+1)^{\frac{1}{2\alpha }} \Big\}.
    \end{equation*}
    Note that for each $j\in \{1,\dots,m\}$ and $k\in\{1,\dots , 2^{m+1-j}\}$ there is at most two distinct vertices in $\Lambda_L\cap 2^j \Z$ such that $(k-1)2^j\le r_u< k 2^j$
    By Lemma~\ref{lem:holder} we have that 
    \begin{equation}
        \mathbb P (\mathcal E ^c) \le C\sum  _{j=1}^m  \sum_{k=1}^ {2^ {m+1-j}}\exp \big( -ct^{2\alpha} k^{\alpha \xi}\log (m-j+1) \big) \le Ce^{-ct^{2\alpha}},
    \end{equation}
    where the last inequality holds when $t$ is sufficiently large. Next, we let $v\in \Lambda _L$ and bound $|\varphi^ \eta _v\cdot e|$ on the event $\mathcal E$. Let $k:=\lceil \log _2r_v \rceil $. Without loss of generality, we assume that $v\ge 0$. There exists a sequence of vertices $v_0,v_1,\dots ,v_k$ such that $v_0=v$, for all $j<k$, $v_j\in 2^j\mathbb Z $ and $v_{j+1}-v_j\in \{0, 2^j\}$ and $v_k\in\Lambda_L^c$. In particular, we have $\varphi_{v_k}^\eta =0$ and $r_{v_0}\ge \dots\ge r_{v_k}$. Thus, we have on the event $\mathcal E $
    \begin{equation}
    \begin{split}
        | \varphi^ \eta _v \cdot e | \le \sum _{j=1}^k |(\varphi^ \eta _{v_{j}}-\varphi^ \eta _{v_{j-1}})\cdot e |  &\le t\sqrt {f(\lambda)}\sum _{j=1}^kr_{v}^ {\frac \xi 2} 2^{\frac{(j-1)\xi}{2}} \log (m-j+2)^{1/(2\alpha)} \\&\le Ct \sqrt {f(\lambda)}r_v^{\xi} \left(\log \log \frac {3L}{r_v}\right)^ {\frac 1{2\alpha}},
        \end{split}
    \end{equation}
    as needed.
\end{proof}

\subsection{Average height localization}\label{sec:average height localization}
In this section we prove the bound \eqref{eq:ge direction of scaling relation} in Theorem~\ref{theorem:half scaling relation}.
Let $L\ge 1$ and $e\in \mathbb R^ n$.
We define
\[\forall x\in\Z^ d\quad s(v):=\frac {e}{2d|\Lambda_L|}\E_v\left[\tau_{\Lambda_L}\right]\]
where $\tau _\Lambda $ is the first exit time of $\Lambda $ for a simple discrete time random walk on $\Z^d$. By conditioning on the first step of the random walk we have for all $v\in \mathbb Z ^d$
\[s(v)=\frac{e}{2d|\Lambda_L|}+\frac{1}{2d}\sum_{u:u\sim v}s(u)\]
and therefore
\[\Delta s(v)=\sum_{u:u\sim v}(s(u)-s(v))=-\frac{e}{|\Lambda_L|}.\]
Moreover, it is easy to check that
\begin{equation}
\E_v[\tau_{\Lambda_L}]\le \sum_{i=1}^d\E_{v_i}[\tau^i_{\Lambda_L}]
\end{equation}
where $\tau^i_{\Lambda_L}$ denotes the exit time of $\{-L,\dots,L\}$ for a simple discrete time random walk on $\Z$ starting at $v_i$.
Using Lemma~\ref{lem:88}, we have
\begin{equation*}
\|s(v)\|= \frac{1}{2d|\Lambda_L|}\E_v[\tau_{\Lambda_L}]\le CL^{2-d}.
\end{equation*}
It yields
\begin{equation*}
    \|\nabla s\|_\Lambda^2=-\sum_{v\in\Lambda_L} s(v)\cdot \Delta s(v) \le C L^ {2-d}.
\end{equation*}
We conclude by applying Lemma \ref{lem:fluctuation and concentration} (inequality~\eqref{eq:lower bound on fluctuations}) with $\gamma=\med(\GE^{\eta,\lambda,\Lambda_{L}}) $.

\section{Delocalization}\label{sec:deloc}
\global\long\def\indic#1{\mathbf{1}_{#1}}
This section provides the proofs of the delocalization results in Theorems \ref{thm:delocalization intro}, \ref{theorem:other half of scaling relation}, \ref{thm:scalingrelation<} and \ref{thm:lower bounds on energy fluctuations by localization}. We fix $\lambda>0$ throughout and remove it from the notation.
\subsection{Preliminary results}

We start by proving some preliminary key results. Let $\Lambda\subset \Z^d$. Recall that $H^{\eta,\lambda,\Lambda}$ and $\Omega^\Lambda$ were defined in Section \ref{subsection : the model}. In this subsection, $\Lambda$ is fixed and we remove it from the notation.
\subsubsection{Deterministic results}
Let $\Pi\subset\Omega^{\Lambda}$ be a closed set of surfaces
and let $\eta$ be an environment such that $H^{\eta}$ admits a minimum on $\Pi$. Denote $\varphi^{\eta,\Pi}\coloneqq\arg\min_{\varphi\in\Pi}H^{\eta}(\varphi)$
and $\GE_{\Pi}^{\eta}\coloneqq H^{\eta}(\varphi^{\eta,\Pi})$. The following proposition, which builds upon the main identity (Proposition~\ref{prop:main identity}), provides the main deterministic input of our proofs of delocalization. 

\begin{prop}\label{prop:Deloc2}
Let $\Lambda\subset\Z^{d}$. Let $\Pi\subset\Omega^{\Lambda}$ be a closed set.
Let $s:\Z^{d}\to\R^{n}$ be supported on $\Lambda$. Let $\eta$ and $\zeta$
be two environments satisfying that $H^{\xi}$ admits a finite minimum on $\Omega ^{\Lambda,\tau}$ and on $\Pi$ when $\xi\in\{\eta,\zeta,\eta^s,\zeta^s,\eta ^{-s},\zeta^{-s}\}$. Denote 
\begin{align*}
\Delta_{1} & :=\GE_{\Pi}^{\eta}-\GE_{\Pi}^{\zeta}\\
\Delta_{2} & :=\inf_{\varphi\in (\Pi+s)\cup(\Pi-s)}H^\zeta(\varphi)-H^\eta(\varphi)\\
\Delta_{3} & :=\frac{1}{2}(\varphi^{\zeta^{s}}-\varphi^{\zeta^{-s}},-\Delta s)\\
\Delta_{2}' & :=\frac{1}{2}\big(\GE_{\Pi}^{\zeta^{s}}+\GE_{\Pi}^{\zeta^{-s}}-\GE_{\Pi}^{\eta^{s}}-\GE_{\Pi}^{\eta^{-s}}\big)\\
\Delta_{3}' & :=\frac{1}{2}(\varphi^{\eta^{s},\Pi}-\varphi^{\eta^{-s},\Pi},-\Delta s).
\end{align*}
Then
\begin{equation}\label{eq:deterministic delocalization}
\{\Delta_1+\Delta_2+\Delta_3>\|\nabla s\|_{\Lambda}^{2}\}\cup \{\Delta_1+\Delta_2'+\Delta_3'>\|\nabla s\|_{\Lambda}^{2}\}\subset \{\varphi^{\zeta^{s}}\notin\Pi\}\cup \{\varphi^{\zeta^{-s}}\notin\Pi\}\cup  \{\varphi^{\eta}\notin\Pi\}.
\end{equation}
\end{prop}

In our use of the proposition, the set $\Pi$ corresponds to localized surfaces, in a suitable sense, and the disorders $\eta$ and $\zeta$ correspond to the original disorder and a resampled (or coupled) new disorder, respectively.

The containment~\eqref{eq:deterministic delocalization} describes two mechanisms leading to delocalization: either $\Delta_1+\Delta_2+\Delta_3$ is large or $\Delta_1+\Delta_2'+\Delta_3'$ is large. Our delocalization results will then use probabilistic arguments to lower bound the probability that the desired inequality is satisfied.

One advantage of the quantity $\Delta_3$ is that it may be controlled via Markov's inequality for general sets $\Pi$; see Corollary~\ref{cor:markov}. This is used (for $\Delta_3$) in our proof of Theorem~\ref{thm:delocalization intro} in order to show that not only the maximum of the minimal surface is large but also that a positive proportion of its heights are large.

In several places, we will make use of the shift function of the following proposition.
\begin{prop}
\label{prop:const_laplace}Fix $d\in\N$, $\ep>0$. Let $L\ge 1$.
Write $\Lambda_{L}^{-}:=\Lambda_{\left\lceil \left(1-\frac{\ep}{2d}\right)L\right\rceil }$
and note that $|\Lambda_{L}^{-}|/|\Lambda_{L}|\ge1-\ep$. There is
a function $\pi:\Z^{d}\to\R$ satisfying the following conditions:
\begin{align}
\pi_{v} & =0 & \,\,\,\,\forall v\notin\Lambda_{L}\label{eq:s_0}\\
\Delta\pi_{v} & =O_{\ep,d}(\frac{1}{L^{2}}) & \,\,\,\,\forall v\in\Z^{d}\label{eq:s_laplace}\\
\pi_{v} & \ge1, & \,\,\,\,\forall v\in\Lambda_{L}^{-}\label{eq:s_min}\\
\|\nabla\pi\|^{2} & =O_{\ep,d}(L^{d-2}).\label{eq:s_energy}
\end{align}
\end{prop}

\begin{proof}[Proof of Proposition \ref{prop:Deloc2}]
    We have
\begin{equation*}
\begin{split}
    \GE^{\eta}\le H^{\eta}(\varphi^{\zeta^s,\Pi}-s)&\le  H^{\zeta}(\varphi^{\zeta^s,\Pi}-s)-\Delta_2= \GE^{\zeta^s}_\Pi+ (\varphi^{\zeta^s,\Pi},\Delta s)+ \frac{1}{2}\|\nabla s\|_\Lambda^2-\Delta_2
    \end{split}
\end{equation*}
where we used Proposition~\ref{prop:main identity} in the last equality.
Similarly,
\begin{equation*}
\begin{split}
    \GE^{\eta}\le  \GE^{\zeta^{-s}}_\Pi+ (\varphi^{\zeta^{-s},\Pi},-\Delta s)+ \frac{1}{2}\|\nabla s\|_\Lambda^2-\Delta_2.
    \end{split}
\end{equation*}
Summing the two previous inequalities, we get on the event $\{\varphi^{\zeta^{s}}\in\Pi\}\cap \{\varphi^{\zeta^{-s}}\in\Pi\}$
\begin{equation}\label{eq:000}
    2\GE^{\eta}\le \GE^{\zeta^s}_\Pi+\GE^{\zeta^{-s}}_\Pi+ \|\nabla s\|_\Lambda^2-2\Delta_2-2\Delta_3.
\end{equation}
Moreover, we have 
\begin{equation}
\GE^ {\zeta^{s}}\le H^ {\zeta^ s}(\varphi^ {\zeta ,\Pi}+s)= \GE^ {\zeta}_\Pi+(\varphi^ {\zeta ,\Pi},-\Delta s)+\frac{1}{2}\|\nabla s\|_\Lambda^2
\end{equation}
where we used Proposition~\ref{prop:main identity} in the last equality.
Similarly,
\begin{equation}
\GE^ {\zeta^{-s}}\le \GE^ {\zeta}_\Pi+(\varphi^ {\zeta ,\Pi},\Delta s)+\frac{1}{2}\|\nabla s\|_\Lambda^2.
\end{equation}
So summing the two previous inequalities yields
\begin{equation}\label{eq:001}
\GE^ {\zeta ^{ s}}+ \GE^ {\zeta ^{- s}}\le 2\GE^ {\zeta}_\Pi+\|\nabla s\|_\Lambda^2\le 2\GE^ {\eta}_\Pi- 2\Delta_1 +\|\nabla s\|_\Lambda^2.
\end{equation}
Finally, on the event $\{\varphi^{\zeta^{s}}\in\Pi\}\cap \{\varphi^{\zeta^{-s}}\in\Pi\}$ we get, combining inequalities \eqref{eq:000} and \eqref{eq:001},
\begin{equation}
    2\GE^\eta\le 2\GE^\eta_\Pi +2\|\nabla s\|_\Lambda^2-2(\Delta_1+\Delta_2+\Delta_3).
\end{equation}
Note that $ \GE^\eta< \GE^\eta_\Pi$ implies that $\varphi^{\eta}\notin\Pi$. 
It follows that 
\[\{\Delta_1+\Delta_2+\Delta_3>\|\nabla s\|_{\Lambda}^{2}\}\subset \{\varphi^{\zeta^{s}}\notin\Pi\}\cup \{\varphi^{\zeta^{-s}}\notin\Pi\}\cup  \{\varphi^{\eta}\notin\Pi\}.\]
We are left to prove 
\[\{\Delta_1+\Delta_2'+\Delta_3'>\|\nabla s\|_{\Lambda}^{2}\}\subset \{\varphi^{\zeta^{s}}\notin\Pi\}\cup \{\varphi^{\zeta^{-s}}\notin\Pi\}\cup  \{\varphi^{\eta}\notin\Pi\}.\]
By Proposition~\ref{prop:main identity} and it holds that
\begin{align*}
\GE^{\eta} & \le H^{\eta}(\varphi^{\eta^{s},\Pi}-s)=H^{\eta^{s}}(\varphi^{\eta^{s},\Pi})-(\varphi^{\eta^{s},\Pi},-\Delta s)+\frac{1}{2}\|\nabla s\|_{\Lambda}^{2},\\
\GE^{\eta} & \le H^{\eta}(\varphi^{\eta^{-s},\Pi}+s)=H^{\eta^{-s}}(\varphi^{\eta^{-s},\Pi})+(\varphi^{\eta^{-s},\Pi},-\Delta s)+\frac{1}{2}\|\nabla s\|_{\Lambda}^{2}
\end{align*}
By summing the two previous inequalities, we get
\begin{equation}\label{eq:deloc1}
    2 \GE^{\eta}\le \GE_{\Pi}^{\eta^{s}}+ \GE_{\Pi}^{\eta^{-s}}-2\Delta_3'+ \|\nabla s\|_{\Lambda}^{2}=  \GE_{\Pi}^{\zeta^{s}}+ \GE_{\Pi}^{\zeta^{-s}}-2(\Delta_2'+\Delta_3')+ \|\nabla s\|_{\Lambda}^{2}.
\end{equation}
By inequality \eqref{eq:001}, we have
\begin{equation*}
\GE^ {\zeta ^{ s}}+ \GE^ {\zeta ^{- s}}\le 2\GE^ {\eta}_\Pi- 2\Delta_1 +\|\nabla s\|_\Lambda^2.
\end{equation*}
Assume $\{\varphi^{\zeta^{s}}\in\Pi\}\cap \{\varphi^{\zeta^{-s}}\in\Pi\}$, then by inequality \eqref{eq:deloc1}, we get
\begin{equation*}
    2 \GE^{\eta}\le   \GE^{\zeta^{s}}+ \GE^{\zeta^{-s}}- 2(\Delta_2'+\Delta_3')+ \|\nabla s\|_{\Lambda}^{2}.
\end{equation*}
Combining the two previous inequalities yields
\begin{align*}
    2 \GE^{\eta}\le 2\GE^{\eta}_\Pi -2(\Delta_1+\Delta_2'+\Delta_3')+ 2\|\nabla s\|_{\Lambda}^{2}.
\end{align*}
It follows that 
\[\{\varphi^{\zeta^{s}}\in\Pi\}\cap \{\varphi^{\zeta^{-s}}\in\Pi\}\cap\{\Delta_1+\Delta'_2+\Delta'_3>\|\nabla s\|_{\Lambda}^{2}\}\subset  \{\varphi^{\eta}\notin\Pi\}.  \]
Hence,
\[\{\Delta_1+\Delta_2'+\Delta_3'>\|\nabla s\|_{\Lambda}^{2}\}\subset \{\varphi^{\zeta^{s}}\notin\Pi\}\cup \{\varphi^{\zeta^{-s}}\notin\Pi\}\cup  \{\varphi^{\eta}\notin\Pi\}.\]
The result follows.
\end{proof}
\begin{proof}[Proof of Proposition \ref{prop:const_laplace}]
Fix $d\in\N,\ep>0$. Denote $A\coloneqq(-1+\frac{\ep}{3d},1-\frac{\ep}{3d})^{d}$.
Fix a smooth bump function $p:\R^{d}\to [0,1]$ such that
\begin{equation*}
    p(x)=\begin{cases}
        1, \quad x\in A\\
        0,\quad  x\notin (-1,1)^d.
    \end{cases}
\end{equation*}
Let $L\ge 1$. For small values of $L$, we may define $\pi_{v}=1$ for $v\in\Lambda_{L}^{-}$
and $\pi_{v}=0$ otherwise.

For sufficiently large $L$ define $\pi_{v}=p(v/L)$. Then (\ref{eq:s_0})
holds. Since $L$ is sufficiently large, $\Lambda_{L}^{-}\subset LA$
holds, and implies (\ref{eq:s_min}).

Using a second-order Taylor expansion of $p$ we obtain that for all $v\in \mathbb Z ^d$ and $1\le i\le d$ we have  
\begin{equation*}
    \pi _{v\pm e_i}-\pi _v=\pm \frac{1}{L} \partial _i p (v/L)+O_{\ep,d}(L^{-2}).
\end{equation*}
It follows that $(\Delta \pi )_ v = O_{\ep,d}(L^{-2})$ as the first order cancels out. This finishes the proof of \eqref{eq:s_laplace}. Finally, \eqref{eq:s_energy} follows from \eqref{eq:s_laplace} and Green's identity
\begin{equation*}
    \|\nabla \pi \|^2=(\nabla \pi ,\nabla \pi )=(\pi ,-\Delta \pi ) \le O_{\ep,d}(L^{d-2}).\qedhere
\end{equation*} 
\end{proof}

\subsubsection{Probabilistic results}
For $s\in\Omega^ \Lambda$,
define $H_{s}^{\eta}(\varphi)\coloneqq H^{\eta^{-s}}(\varphi-s)$.
Accordingly, let $\varphi_{s}^{\eta}\coloneqq\varphi^{\eta^{-s}}+s$ be the
minimizer of $H_{s}^{\eta}(\varphi)$. The next proposition reveals a non-obvious positivity statement and harnesses it via Markov's inequality.
\begin{prop}
[\ref{as:exiuni}]\label{prop:close_surf_markov}Fix $\lambda,\Lambda$ and $\eta$ be a random environment.
Let $s\in\Omega^{\Lambda}$. Denote $D:=(\varphi_s^{\eta}-\varphi^{\eta},-\Delta s\,)$.
Then, $0\le H^{\eta}(\varphi_{s}^{\eta})-H^{\eta}(\varphi^{\eta})\le D$
and for every $a>0$ if  the process $(\varphi^ {\eta^ {-ks}})_{k\in\Z}$ is stationary, then

\begin{align*}
\P(\,D & \le a\|\nabla s\|^{2}\,)\ge1-\frac{1}{a}.
\end{align*}
\end{prop}
The following corollary follows easily by applying the previous proposition to the disorder~$\eta^s$ and the shift $2s$. As mentioned, it is useful for bounding the quantity $\Delta_3$ of Proposition~\ref{prop:Deloc2}.
\begin{cor}[\ref{as:exiuni}]\label{cor:markov}Fix $\lambda,\Lambda$ and $\eta$.
Let $s\in\Omega^{\Lambda}$. If  the process $(\varphi^ {\eta^ {-ks}})_{k\in\Z}$ is stationary then
\begin{align*}
\P \big( (\varphi^{\eta^{-s}}-\varphi^{\eta^s},-\Delta s\,) & \le (a-2)\|\nabla s\|^{2}\, \big)\ge1-\frac{2}{a}.
\end{align*}
    
\end{cor}

We will also need the following Gaussian decomposition for $\eta^\white$. Let
\begin{equation}\label{eq:ball def}
B_{r}(s):=\{t\in\R^n\colon\|t-s\|\le r\}
\end{equation}
denote a ball of radius $r$ around $s$ in $\R^{n}$.
\begin{lemma}[Decomposition of $\eta^{white}$]\label{lem : decomposition noise-1}Let
${\rm S\subset\Lambda\times\mathbb{R}^{n}}$ be a bounded region of
the form $\Lambda'\times B_{r}(0)$ for $r\ge1$ and $\Lambda'\subset\Lambda$.
Set 
\[
\eta^{{\rm S}}:=\frac{\int b(t)dt}{{\rm Vol}(S)}\sum_{v\in\Lambda'}W_{v}(\indic{B_{r}(0)}).
\]
For all $v\in\Lambda$ and $t\in\R^{n}$ denote $\kappa_{v,t}\coloneqq\frac{\cov(\eta_{v,t},\eta^{{\rm S}})}{\var(\eta^{{\rm S}})}$
and define 
\[
\eta_{v,t}^{\perp}:=\eta_{v,t}-\kappa_{v,t}\eta^{{\rm S}}.
\]
Then, $\eta^{{\rm S}}$ and $(\eta_{v,t}^{\perp})_{v\in\Lambda,t\in\R^{n}}$
are independent, $\eta^{{\rm S}}\sim\mathcal{N}(0,\frac{\left(\int b(t)dt\right)^{2}}{{\rm {Vol}({\rm S)}}})$ and
\begin{equation}\label{eq:decomp_bound}
    \mathds{1}\{(v,t)\in\Lambda'\times B_{r-1}(0)\}\le\kappa_{v,t}\le\mathds{1}\{(v,t)\in\Lambda'\times B_{r+1}(0)\}.
\end{equation} 

\end{lemma} 
\begin{proof}[Proof of Lemma \ref{lem : decomposition noise-1}]
Since $\eta$ is a Gaussian process, the random variables $\eta^{{\rm S}}$
and $\eta_{v,t}^{\perp}$ are centered gaussian random variables.
We have 
\[
\E[\eta_{v,t}^{\perp}\eta^{{\rm S}}]=\cov(\eta_{v,t},\eta^{{\rm S}})-\cov(\eta_{v,t},\eta^{{\rm S}})=0.
\]
Then, by a property of Gaussian processes, it follows that $\eta^{{\rm S}}$
and $(\eta_{v,t}^{\perp})_{v\in\Lambda,t\in\R^{n}}$ are independent.
Besides, we have 
\[
\begin{split}\E[(\eta^{{\rm S}})^{2}] & =\E\left[\sum_{v\in\Lambda'}\left(W_{v}\left(\frac{\indic{B_{r}(0)}\int b(t)dt}{{\rm Vol}({\rm S)}}\right)\right)^{2}\right]=\frac{\left(\int b(t)dt\right)^{2}}{{\rm {Vol}({\rm S)^{2}}}}\int(\indic S)^{2}=\frac{\left(\int b(t)dt\right)^{2}}{{\rm {Vol}({\rm S)}}}\end{split}
\]
For $v\notin\Lambda'$, we have $\E[\eta^{{\rm S}}\eta_{v,t}]=0$.
Otherwise, we have 
\[
\begin{split}\E[\eta^{{\rm S}}\eta_{v,t}]=\E\left[W_{v}\left(\frac{\indic{B_{r}(0)}\int b(t)dt}{{\rm Vol}({\rm S)}}\right)W_{v}(b(\cdot-t))\right]=\frac{\int_{B_{r}(0)}b(s-t)ds\int b(s)ds}{\rm {Vol}({\rm S)}}\end{split}
\]
Note that since $b$ is non-negative, $0\le\int_{B_{r}(0)}b(s-t)ds\le\int b(s)ds$.
Also since $b(t)=0$ when $\|t\|\ge1$, it holds that 
\[
\int_{B_{r}(0)}b(s-t)ds=\begin{cases}
0 & \|t\|\ge r+1\\
1 & \|t\|\le r-1
\end{cases}.
\]
The last assertion of the lemma follows.
\end{proof}
\begin{proof}[Proof of Proposition \ref{prop:close_surf_markov}] Denote $\eta(\varphi)\coloneqq\sum_{v\in\Lambda}\eta_{v,\varphi_{v}}$
so that $H^{\eta}(\varphi)=\frac{1}{2}\|\nabla\varphi\|^{2}+\lambda\eta(\varphi)$
and more generally $H_{s}^{\eta}(\varphi)=\frac{1}{2}\|\nabla(\varphi-s)\|^{2}+\lambda\eta(\varphi)$.
Consequently, we may rewrite the main identity as 
\[
H_{s'}^{\eta}(\varphi)-H_{s}^{\eta}(\varphi)\coloneqq\frac{1}{2}\|\nabla s'\|^{2}-\frac{1}{2}\|\nabla s\|^{2}+(\varphi,-\Delta(s-s')\,).
\]
Hence
\[H^ {\eta}(\varphi)=H^\eta _s(\varphi)-\frac 12 \|\nabla s\|^ 2+(\varphi,-\Delta s).\]
This yields
\begin{align*}
0 & \le H^{\eta}(\varphi_{s}^{\eta})-H^{\eta}(\varphi^{\eta})\\
 & =H_{s}^{\eta}(\varphi_{s}^{\eta})-H_{s}^{\eta}(\varphi^{\eta})+(\varphi_{s}^{\eta}-\varphi^{\eta},-\Delta s\,)\\
 & \le0+D
\end{align*}
where we used in the last inequality that $\varphi^ \eta_s$ is the minimizer of $H_{s}^{\eta}$. By assumption, the process $(\varphi^ {\eta^ {-ks}})_{k\in\Z}$ is stationary. It yields that 
\[\E(\varphi^ \eta,-\Delta s)= \E(\varphi^ {\eta^ {-s}},-\Delta s)\]
and
\[\E D= (s,-\Delta s)=\|\nabla s \|^ 2.\]
The result follows by applying Markov's inequality.
\end{proof}

\subsection{The scaling inequality $\chi\le2\xi+d-2$}\label{sec:the scaling inequality chi less than xi}
In this section, we prove a stronger \emph{local} version of Theorem \ref{theorem:other half of scaling relation}. We then deduce Theorem \ref{thm:scalingrelation<} and inequality \eqref{eq:std upper bound} of Theorem \ref{thm: scaling relation d=1}. 

Denote by $\Lambda_{w,\ell}$ the box $w+\Lambda_{\lfloor \ell\rfloor}$. Recall the definition of $\eta[\Delta]$ from before Theorem~\ref{theorem:other half of scaling relation}.

\begin{theorem}[Local version of $\chi\le2\xi+d-2$. \ref{as:exiuni}+ \ref{as:stat}+\ref{as:indep}]\label{theorem:local other half of scaling relation} 
There exists $C>0$, depending only on the dimension $d$, such that the following holds for each $\Lambda\subset \Z^ d$, each $\lambda>0$ and each unit vector $e\in\R^n$. 
For all $h\ge 1$, $\ell\ge 1$, $w\in \Lambda$ such that $\Lambda_{w,2\ell}\subset \Lambda$, for any $ \Delta \subset\Lambda_{w,\ell} $, we have
\begin{equation}\label{eq:le direction of local scaling relation}
\P\left(\big|\GE^{\eta,\lambda,\Lambda}-\GE^{\eta[\Delta],\lambda,\Lambda}\big|\ge Ch^2\ell^{d-2}\,\big|\, \eta_{\Lambda\setminus \Lambda_{w,2\ell}}\right) \le C\P\left(\max_{v\in\Lambda_{w,2\ell}} \left|\varphi^{\eta,\lambda,\Lambda}_v\cdot e\right| \ge h\,\big|\, \eta_{\Lambda\setminus \Lambda_{w,2\ell}}\right).
\end{equation}
Furthermore, we have for all $h\ge 1$ and $w\in\Lambda$ 
\begin{equation}\label{eq:le direction of local scaling relation l=0}
\P\left(\big|\GE^{\eta,\lambda,\Lambda}-\GE^{\eta[\{w\}],\lambda,\Lambda}\big|\ge Ch^2\,\big|\, \eta_{\Lambda\setminus \Lambda_{w,1}}\right) \le C\P\left(\max_{v\in\Lambda_{w,1}} \left|\varphi^{\eta,\lambda,\Lambda}_v\cdot e\right| \ge h\,\big|\, \eta_{\Lambda\setminus \Lambda_{w,1}}\right).
\end{equation}
\end{theorem}

\begin{proof}[Proof of Theorem \ref{theorem:local other half of scaling relation}]
Let $d,n\in\N$, $\Lambda\subset \Z^ d$, $\lambda>0$, $e\in\R^{n}$, $h\ge1$. Let us first assume $\ell\ge1$ and prove \eqref{eq:le direction of local scaling relation}.   Let $\zeta=\eta[\Delta]$,
i.e. $\zeta$ is obtained from $\eta$ by resampling $\eta_{\Delta\times\R^n}$.
Let $\pi$ be the function from Proposition \ref{prop:const_laplace} applied to the box $\Lambda_{w,2\ell}$ and $\ep=1/2$,
let $C_{0}>0$ and define the surface $s\coloneqq C_{0}h\pi e$. Denote
\[\Pi\coloneqq\{\varphi\in\Omega^\Lambda:\,\forall v\in\Lambda_{w,2\ell}\quad|\varphi_{v}\cdot e|\,\le h\,\}\]
and
\[
S_{i}\coloneqq\left\{ (v,\varphi_{v}):v\in\Delta,\varphi\in\Pi+is\right\} .
\]
By (\ref{eq:s_min}), taking $C_{0}$ large enough the sets $S_{i}$
for $i\in\Z$ are pairwise separated by a distance of at least $2$.
Let $\Delta_{1},\Delta_2',\Delta_3'$ be as in Proposition \ref{prop:Deloc2},
for $\Pi,\eta,\zeta$ and $\Lambda=\Lambda_{L}$. For $\varphi\in\Pi$,
it holds that $|(\varphi,-\Delta s)|=\sum_{v\in\Lambda_{w,2\ell}}(\varphi_{v}\cdot e) C_{0}h\Delta\pi_{v}=O(\ell^{d-2}h^{2})$
by (\ref{eq:s_laplace}). Thus $\Delta_3'\ge-O(\ell^{d-2}h^{2})$ almost
surely.

Let $\mathcal{F}$ be the sigma algebra generated by the restriction
of $\eta$ to $(\Lambda\setminus\Delta)\times\R^{n}$.
Conditionally on $\mathcal{F}$, the variable $\Delta_{1}$ is a function
of $\eta|_{S_{0}}$ and $\zeta|_{S_{0}}$, while $\Delta_2',\Delta_3'$
are functions of $\eta|_{S_{-1}\cup S_{1}}$ and $\zeta|_{S_{-1}\cup S_{1}}$. In particular, conditionally on $\mathcal{F}$, $\Delta_1$ is independent of $\Delta_2'$.
Conditionally on $\mathcal{F}$, the disorders $\eta$ and $\zeta$ are
exchangeable, thus $\P(\Delta_2'\ge0\,|\,\mathcal{F})\ge1/2$ almost
surely. By (\ref{eq:s_energy}), and the bound on $\Delta_3'$,
there is a constant $C>0$ such that almost surely 
\[
B\coloneqq C\ell^{d-2}h^{2}>\|\nabla s\|_{\Lambda}^{2}-\Delta'_{3}.
\]
Since, conditionally on $\mathcal{F}$, $\Delta_{1}$ is independent
of $\Delta_2'$, we have that
\begin{equation*}
    \begin{split}
        \frac{1}{2}\P(\Delta_{1}>B\,|\,\mathcal F)&\le\P(\Delta_2'\ge0|\mathcal F)\P(\Delta_{1}>B|\mathcal F)\\&\le\P(\Delta_{1}+\Delta_2'+\Delta_3'>\|\nabla s\|_{\Lambda}^{2}|\mathcal F)\\&\le\P(\varphi^{\eta}\notin\Pi|\mathcal F)+ \P(\varphi^{\zeta^s}\notin\Pi|\mathcal F)+ \P(\varphi^{\zeta^{-s}}\notin\Pi|\mathcal F)
    \end{split}
\end{equation*}
where we used Proposition \ref{prop:Deloc2} in the last inequality.
Note that conditionally on $ \eta_{\Lambda\setminus \Lambda_{w,2\ell}}$ since $s$ has its support included in $\Lambda_{w,2\ell}$ and by assumption \ref{as:stat}, the surfaces $\varphi^{\eta},\varphi^{\zeta^s}$ and $\varphi^{\zeta^{-s}}$ have the same distribution.
Taking expectation over $\eta$ above $\Lambda_{w,2\ell}\setminus \Delta$ and using the exchangeability,
we get \[\P(|\Delta_{1}|>B| \eta_{\Lambda\setminus \Lambda_{w,2\ell}})=2\P(\Delta_{1}>B| \eta_{\Lambda\setminus \Lambda_{w,2\ell}})\le 12\P(\varphi^{\eta}\notin\Pi |\eta_{\Lambda\setminus \Lambda_{w,2\ell}}).\] Finally
\[
\{|\GE^{\eta}-\GE^{\zeta}|\ge B\}\subset\{|\Delta_{1}|\ge B\}\cup\{\varphi^{\eta}\notin\Pi\}\cup\{\varphi^{\zeta}\notin\Pi\}
\]
and by a union bound \[\P(|\GE^{\eta}-\GE^{\zeta}|\ge B|\eta_{\Lambda\setminus \Lambda_{w,2\ell}})\le 14\P(\varphi^{\eta}\notin\Pi|\eta_{\Lambda\setminus \Lambda_{w,2\ell}})\]
which concludes the proof in the case $\ell\ge 1$.
To prove \eqref{eq:le direction of local scaling relation l=0}, it is easy to check that the same proof holds by considering the shift function $\pi'=\mathbf{1}_ve$ instead of $\pi$.
\end{proof}

\begin{proof}[Proof of Theorem \ref{thm:scalingrelation<}] 
We start by building a partition of $\Lambda_L$ into blocks in such a way that each block is at distance at least its side length from $\Lambda_L^c$. We will then interpolate between two different disorders by resampling the disorder block after block.

For all $j\ge 0$, let $\mathcal A_j$ be the set of blocks of the form $w+[0,2^j)^d$ for all $w\in 2^j\mathbb Z ^d$ that intersect $\Lambda _L$. Let $m:=\lfloor \log_2 L \rfloor $. Let $\mathcal B_{m}\subseteq \mathcal A _{m}$ be the set of good blocks in $\mathcal A _{m}$ whose distance from  $\Lambda _L^c$ is at least $2^{m}$ and let $\mathcal C_{m}$ be the rest of the blocks in $\mathcal A _{m}$. We now proceed inductively to smaller and smaller $j$. Suppose $\mathcal B_{j+1},\mathcal C_{j+1}$ were defined. Define $\mathcal B_j$ to be the set of blocks in $\mathcal A _j$ that are contained in some block in $\mathcal C _{j+1}$ that are at distance at least $2^j$ from $\Lambda _L^c$ and let $\mathcal C_j$ be the rest of the blocks in $\mathcal A _j$ that are contained in some block in $\mathcal C _{j+1}$. Finally let $\mathcal B :=\bigcup _{j\ge 0} \mathcal B _j$ be the partition. Clearly, the number of blocks in $\mathcal B _j$ is at most $C(L2^{-j})^{d-1}$ since all of these blocks are in $\mathcal A _j$ with distance at most $2^{j+2}$ from the boundary of $\Lambda _L$. Let us denote by $\Lambda_1,\dots,\Lambda_{k'}$ this family of blocks.

Let $\eta'$ be an independent disorder with the same distribution as $\eta$.
We define a sequence of disorders $\eta^0,\dots,\eta^{k'}$ where $\eta^0=\eta$ and $\eta^k=\eta$ on $\Lambda_L\setminus \cup_{i\le k}\Lambda_i$ and $\eta^k=\eta'$ on $\cup_{i\le k}\Lambda_i$  $\eta^{k'}=\eta'$. Denote by $L_k$ the side length of $\Lambda_k$. Using the Efron--Stein inequality, we get
\begin{equation}\label{eq:efronstein}
\var(\GE^\eta)=\frac{1}{2}\E[(\GE^ \eta-\GE^{\eta'})^2]\le\frac{1}{2}\sum_{k=0}^{k'-1}\E[(\GE^ {\eta^k}-\GE^{\eta^{k+1}})^2].
\end{equation}
If $L_k>1$, pick $w_k\in\Lambda_L$ such that $\Lambda_k\subset (w_k+\Lambda_{L_k/2})$ and set $\Lambda'_k\coloneqq w_k+\Lambda_{L_k}$. Since $\Lambda_k$ is at distance at least $L_k$ from $\Lambda_L^c$, we have $\Lambda'_k\subset \Lambda_L$. If $L_k=1$ with $\Lambda_k=\{w\}$ then $\Lambda'_k=\Lambda_{w,1}$. In particular, thanks to our choice of the partition we can apply Theorem~\ref{theorem:local other half of scaling relation}, it yields for any $k$ and $h\ge 1$
\begin{equation*}
    \P(|\GE^ {\eta^k}-\GE^{\eta^{k+1}}|\ge Ch^2L_{k+1}^{d-2})\le C\P(\max_{v\in\Lambda_{k+1}'}\|\varphi_v^\eta\|\ge h) 
\end{equation*}
and 
\begin{equation}\label{eq:ineqefronstein}
\begin{split}
    \E[(\GE^ {\eta^k}-\GE^{\eta^{k+1}})^2]
&\le \int_{t\ge 0}\P(|\GE^ {\eta^k}-\GE^{\eta^{k+1}}|^2\ge t)\\&\le C^2 L_{k+1}^{2(d-2)}\!+  C \int_{t\ge C^2 L_{k+1}^{2(d-2)}}\P \big( \max_{v\in\Lambda_{k+1}'}\|\varphi_v^\eta\|\ge L_{k+1}^{-d/2+1}t^{1/4}C^{-1/2} \big)\\
&\le C \left(1+\E[\max_{v\in\Lambda_{k+1}'}\|\varphi_v^\eta\|^4]\right)L_{k+1}^{2(d-2)}.
\end{split}
\end{equation}
By construction, any point in the box of side length $L_k$ is at distance at least $L_k$ and at most $5 L_k$ from the boundary of $\Lambda_L$.
The definition~\eqref{eq:kappa def} of $\kappa$ (or the alternative definition~\eqref{eq:kappa alternative def}) then implies that
\begin{equation*}
    \E[(\GE^ {\eta^k}-\GE^{\eta^{k+1}})^2]
\le C\kappa  \left(\log \frac L {L_{k+1}} \right)^{4\beta} L_{k+1}^{2(2\xi+(d-2))}\end{equation*}
where we used that $\xi\ge 0$.
The number of cubes among the $\Lambda_1,\dots,\Lambda_k$ of side length $2^j$ is at most $O((L2^ {-j})^{d-1})$.
Hence, combining the previous inequality together with inequality \eqref{eq:efronstein}, 
\begin{equation}
\begin{split}
    \var(\GE^\eta)
&\le \sum_{j=0} ^m C\kappa(L2^ {-j})^{d-1}\left(\log \frac L {2^j} \right)^{4\beta} 2^{2j(2\xi+(d-2))}\\& \le C\kappa L^{d-1} \sum_{j=0} ^m \left(m+1-j\right)^{4\beta} 2^{j(2(2\xi+(d-2))-(d-1))}\\&\le C\kappa L ^{2(2\xi+(d-2))}\sum_{j=1}^{m+1}j^ {4\beta}2^{j[(d-1)- 2(2\xi+d-2)]}.
\end{split}
\end{equation}
We distinguish three cases:
\begin{itemize}
    \item If $d-1=2(2\xi+d-2)$ then $\var(\GE^\eta)
\le C\kappa (\log L)^{4\beta +1} L ^{2(2\xi+(d-2))}$.
\item If $d-1<2(2\xi+d-2)$ then $\var(\GE^\eta)
\le C \kappa L ^{2(2\xi+(d-2))}$.
\item If $d-1>2(2\xi+d-2)$ then $\var(\GE^\eta)
\le C\kappa (\log L)^{4\beta } L ^{d-1}$.\qedhere
\end{itemize}
\end{proof}

\begin{proof}[Proof of Theorem \ref{thm: scaling relation d=1}, inequality \eqref{eq:std upper bound}] Consider the same decomposition as in the proof of Theorem \ref{thm:scalingrelation<} in the case of $d=1$. Note that there is at most a constant number of boxes in the partition such that their side length is $2^j$. By construction, if a box $\Lambda_k$ has side length $2^j$, then the corresponding box $\Lambda'_k$ is included in $\{v:L-2^{j+2}<|v|\le L\}$.
It yields, combining \eqref{eq:efronstein} and \eqref{eq:ineqefronstein}, that 
\begin{equation}
   \std(\GE^{\eta,\lambda,\Lambda_L})\le C\sum_{0\le j\le \lceil\log_2 L\rceil} 2^{-j}\sqrt{1+ \E M_{2^{j+2}}^4},
\end{equation}
concluding the proof.
\end{proof}

\subsection{Delocalization when $d=4$}

\label{sec:delocalization in dimension 4} In this section, we prove
Theorem \ref{thm:delocalization intro} for $d=4$ and $n\ge1$. The
proof method is inspired by the quantitative bound of~\cite{dario2024quantitative}
on the magnetization decay of the $d=4$ random-field spin $O(m)$
model with $m\ge2$ (following the non-quantitative result of Aizenman--Wehr~\cite{aizenman1990rounding}).
It is a simple application of fractal (Mandelbrot) percolation.

In all this section, $\eta=\eta^{\white}$. Let $L\ge3$ and $\lambda>0$.
We will consider ``delocalizatoin events'' on ``blocks'' $\Lambda$
taken from a nested structure of blocks within $\Lambda_{L}$. Let
$\Lambda\subset\Lambda_{L}$. Define to $D_{\Lambda}=D_{\Lambda }(h)$ be the number
of vertices $v\in\Lambda$ on which the surface delocalizes to height
$h$. That is, 
\begin{equation}
D_{\Lambda}:=\big|\big\{ v\in\Lambda:\|\varphi_{v}\|>h\big\}\big|.
\end{equation}
Define $\hat{\eta}_{\Lambda}$ to be the integral of $W$ in the region
$\Lambda\times B_{h+1}(0)$, where $B_{r}(s)$ is the ball of radius $r$ defined in~\eqref{eq:ball def}.
That is 
\begin{equation}
\hat{\eta}_{\Lambda}:=\sum_{v\in\Lambda}W_{v}(\indic{B_{h+1}(0)}).
\end{equation}
Next, let $\mathcal{F}_{\Lambda^{c}}$ be the sigma algebra generated
by the disorder in $\Lambda^{c}\times\mathbb{R}^{n}$.
\begin{lem}[$\eta=\eta^{\white}$]
\label{lem:83-1}There is a constant $C_{0}>0$ (depending on $\mathcal{L}$)
such that the following holds for all $h\ge\max\{1,\lambda^{\frac{2}{n+4}}\}$.
Let $\Lambda=v+\Lambda_{\ell}\subseteq\Lambda_{L}$
for some $\ell\ge1$. Denote $a(\ell)=C_{0}\lambda^{-1}\ell^{2}h^{n+2}$.
Then 
\begin{equation}
\mathbb{E}\big[D_{\Lambda}\mid\hat{\eta}_{\Lambda},\mathcal{F}_{\Lambda^{c}}\big]\ge\frac{1}{5}|\Lambda|\quad\text{on the event}\quad\big\{\hat{\eta}_{\Lambda}\ge a(\ell)\big\}.\label{eq:block_deloc_cond}
\end{equation}
\end{lem}

\begin{proof}[Proof of Theorem~\ref{thm:delocalization intro} for $d=4$] 
Suppose that $L$ is sufficiently large. We would like to apply Lemma~\ref{lem:83-1} with $h:=c_0 \lambda ^{\frac{2}{4+n}} (\log \log L)^{\frac{1}{4+n}}$ for a sufficiently small constant $c_0>0$ to be determined later. We say that the block $\Lambda=v+ \Lambda _\ell  \subseteq \Lambda _L$ is good if $\hat \eta _\Lambda \ge a(\ell) $ with $a(\ell) =C_0\lambda ^{-1}\ell ^2 h^{n+2}$ as in Lemma~\ref{lem:83-1}. Note that $\hat\eta _\Lambda$ is normally distributed with zero expectation and $\var(\hat\eta _\Lambda)\ge c\ell ^4h^n \ge 10a(\ell )/\log \log L$, as long as $c_0$ is sufficiently small depending on $C_0$, and therefore $\mathbb P (\Lambda \text{ is good}) \ge (\log L)^{-1/3}$ for large $L$.

Our goal will be to construct a set $\mathcal Q$ of disjoint good blocks that cover a positive fraction of $\Lambda _L$ with high probability. To this end, we consider a sequence of hierarchical partitions. Let $j_1:=\lfloor \tfrac{1}{2} \sqrt{\log L} \rfloor $ and for any integer $j\in[0,j_1]$ let $m_j:=\lfloor \log _3L \rfloor -j\lfloor \sqrt{\log L}\rfloor  $ and $R_j:=3^{m_j}$. Define $\Lambda ^0:=[0,R_0)^d \subseteq \Lambda _L$ and let
\begin{equation}
\mathcal T _j:= \Big\{  \Lambda =\big( z+[0,R_j)^d \big) \cap \mathbb Z ^d : z\in R_j \mathbb Z ^d \text{ and } \Lambda \subseteq \Lambda ^0 \Big\}.
\end{equation} 
Note that $\mathcal T _0=\{\Lambda ^0\}$ and that $\mathcal T _{j+1}$ is a refinement of the partition $\mathcal T _j$ such that any block $\Lambda \in \mathcal T _j$ is the union of $3^{d\lfloor \sqrt{\log L} \rfloor }$ blocks in $\mathcal T _{j+1}$. Moreover, by construction, the side length of each block in $\mathcal T _j$ is odd and therefore it is of the form $v+\Lambda _{\ell _j} $ with $\ell _j:=(R_j-1)/2$.

We define $\mathcal Q$ in the following way. A block $\Lambda \in \mathcal T _j$ belongs to $\mathcal Q$ if and only if $\Lambda $ is good and for all $j'<j$ the unique block $\Lambda ' \in \mathcal T _{j'}$ with $\Lambda \subseteq \Lambda '$ is not good. We say that a vertex $v\in \Lambda _L$ is covered if $v\in \Lambda $ for some $\Lambda \in \mathcal Q $. We claim that for all $v\in \Lambda ^0$
\begin{equation}\label{eq:bound for v}
    \mathbb P (v\text{ is covered}) \ge 1/2.
\end{equation}
Indeed, letting $\Lambda ^j=\Lambda ^j(v) $ be the unique block in $\mathcal T _j$ containing $v$ we have
\begin{equation}\label{eq:44}
    \Big\{ \!\! \begin{array}{c} v \text{ is not} \\ \text{covered}\end{array} \!\! \Big\} \subseteq \bigcap _{j=0}^{j_1-1} \!\big\{ \hat \eta _{\Lambda ^j} \le  a(\ell _j) \big\} \subseteq \bigcup _{j=0}^{j_1-1} \! \big\{\hat \eta _{\Lambda ^{j+1}} \le  -a(\ell _j) \big\} \cup \bigcap _{j=0}^{j_1-1} \!\big\{ \hat \eta _{\Lambda ^j}-\hat \eta _{\Lambda ^{j+1}} \le  2a(\ell _j) \big\},
\end{equation}
The random variable $\hat \eta  _{\Lambda ^{j+1}}$ is normal with $\var (\hat \eta  _{\Lambda ^{j+1}})\le C\ell _{j+1}^{4}h^{n} \le a(\ell _j)^2\exp (-\sqrt{\log L})$ and therefore $\mathbb P(\hat \eta  _{\Lambda ^{j+1}} \le -a(\ell _j)) \ll 1/L $. Moreover, the variables $\{\hat \eta  _{\Lambda ^j}-\hat \eta  _{\Lambda ^{j+1}}\}_{j=0}^{j_1-1}$ are independent and normal with $\var (\hat \eta  _{\Lambda ^j}-\hat \eta  _{\Lambda ^{j+1}})\ge c\ell _j^{4}h^{n}\ge 10a(\ell _j)^2/\log \log L$, as long as $c_0$ is sufficiently small, and therefore the probability of the intersection on the right hand side of \eqref{eq:44} is at most $(1-(\log L)^{-1/3})^{j_1} \le 1/4$. Substituting these bounds in \eqref{eq:44} completes the proof of \eqref{eq:bound for v}.

Next, note that the for all $j'\le j$ and any $\Lambda \subseteq \Lambda '$ with $\Lambda \in \mathcal T_j$ and $\Lambda '\in \mathcal T _{j'}$, the random variable $\hat \eta _{\Lambda '}$ is measurable in $ \sigma (\hat \eta _\Lambda , \mathcal F _{\Lambda ^c} )$. Thus, the event $\{\Lambda \in \mathcal Q \}$ is measurable in $ \sigma (\hat \eta _\Lambda , \mathcal F _{\Lambda ^c} )$ and therefore by Lemma~\ref{lem:83-1} we have 
\begin{equation}
    \mathbb E \big[ 
D_{\Lambda }\mathds 1 \{\Lambda \in \mathcal Q \} \big] =\mathbb E \big[ \mathds 1\{\Lambda \in \mathcal Q \} \cdot \mathbb E [D_\Lambda \mid \hat \eta _\Lambda , \mathcal F _{\Lambda ^c} ]  \big]\ge \tfrac{1}{5} |\Lambda | \cdot  \mathbb P (\Lambda \in \mathcal Q ).
\end{equation}
Hence, we obtain
    \begin{equation}
    \begin{split}
        \mathbb E [D_{\Lambda _L}] \ge \mathbb E \Big[ \sum _{\Lambda \in \mathcal Q } D_\Lambda \Big]\ge \sum _{j=0}^{j_1} \sum _{\Lambda \in \mathcal T _j} \mathbb E \big[ 
D_{\Lambda }\mathds 1 \{\Lambda \in \mathcal Q \} \big]\ge \frac{1}{5} \sum _{j=0}^{j_1} \sum _{\Lambda \in \mathcal T _j} |\Lambda |\cdot \mathbb P (\Lambda \in \mathcal Q \big)\\
=\frac{1}{5} \cdot \mathbb E \Big[ \sum _{\Lambda \in \mathcal Q } |\Lambda | \Big] = \frac{1}{5} \sum _{v\in \Lambda ^0 } \mathbb P \big( v \text{ is covered} \big) \ge \frac{1}{10} |\Lambda ^0| \ge c|\Lambda _L|,
\end{split}
\end{equation}
where in the second to last inequality we used \eqref{eq:bound for v}. This finishes the proof.
\end{proof}

It remains to prove Lemma~\ref{lem:83-1}.

\begin{proof}[Proof of Lemma \ref{lem:83-1}]
Fix $\Lambda$ as in the lemma, a function $\eta_{0}:\Lambda^{c}\times\R^{n}\to\R$
and $a\ge a(\ell)$. For the scope of this proof, we let $\zeta$
be distributed according to the law $\overline{\P}\coloneqq{\eta^{\white}}|_{\eta|_{\Lambda^{c}\times\R^{n}}=\eta_{0}}$.
We apply Lemma \ref{lem : decomposition noise-1} to $\zeta$, with
${\rm S}=\Lambda\times B_{h+1}(0)$. We get that $\hat{\zeta}_{\Lambda}$
and $\zeta^{\perp}$ are independent since $\zeta^{S}=\frac{\int b(t)dt}{{\rm Vol}(S)}\hat{\zeta}_{\Lambda}$.
Write $\hat{\zeta}_{\Lambda}=\hat{\zeta}$ for short. Let $\kappa_{v,t}$
be as in the lemma. Define $\eta:=\zeta+\kappa\frac{\int b(t)dt}{{\rm Vol}(S)}(a-\hat{\zeta})$
i.e. $\hat{\eta}=a$ and $\eta^{\perp}=\zeta^{\perp}$. Thus the marginal
law of $\eta$ under $\overline{\P}$ is ${\eta^{\white}}|_{\eta|_{\Lambda^{c}\times\R^{n}}=\eta_{0},\hat{\eta}=a}$
as in (\ref{eq:block_deloc_cond}), and so it remains to prove $\overline{\E}\big[D_{\Lambda}\big]\ge\frac{1}{5}|\Lambda|$
where $\overline{\E}$ is the expectation corresponding to $\overline{\P}$.

Recall that $h\ge\max\{1,\lambda^{\frac{2}{n+4}}\}$ and $a(\ell)=C_{0}\lambda^{-1}\ell^{2}h^{n+2}$
thus 
\[
\frac{a(\ell)}{\sqrt{\var(\hat{\zeta})}}=\Omega(C_{0}\frac{\lambda^{-1}\ell^{2}h^{n+2}}{\ell^{d/2}h^{n/2}})=\Omega(C_{0}\frac{\ell^{2}}{\ell^{2}}\frac{h^{\frac{n+4}{2}}}{\lambda})=\Omega(C_{0}),
\]
(here and in the remainder of the proof the notations $\Omega,O$ hide constants depending on $n,d$ and $\mathcal{L}$). Denote 
\global\long\def\goodev{\mathcal{E}}%
 $\mathcal{\goodev}\coloneqq\{\hat{\zeta}\le a/2\}$. Then $\overline{\P}(\goodev^{c})=O(C_{0}^{-2})$.

Note that by (\ref{eq:decomp_bound}),
\begin{align*}
\kappa_{v,t} & =0\,\text{ for }\|t\|\ge h+2,\\
\kappa_{v,t} & =1\,\text{ for }\|t\|\le h,v\in\Lambda,\\
0\le\kappa_{v,t} & \le1\,\text{ for all }v,t.
\end{align*}
Let $c\coloneqq C_{0}^{-1}$, and define
\[
\Pi\coloneqq\Big\{\varphi\in\Omega^{\Lambda_{L}}:\sum_{v\in\Lambda}\kappa_{v,\varphi_{v}}\ge(1-c)|\Lambda|\Big\}.
\]
Thus $\overline{\E}\big[D_{\Lambda}\big]\ge(1-c)\overline{\P}(\varphi^{\eta}\notin\Pi)$
and it suffices to prove $\overline{\P}(\varphi^{\eta}\notin\Pi)\ge1/6$.
For $\varphi\in\Omega_{L}^{\Lambda}$, denote: 
\begin{align}
B_{0} & \coloneqq\lambda\frac{\int b(t)dt}{{\rm Vol}(S)}|\Lambda|\frac{a}{2}\\
B(\varphi) & \coloneqq H^{\eta}(\varphi)-H^{\zeta}(\varphi)=\lambda\frac{\int b(t)dt}{{\rm Vol}(S)}(a-\hat{\zeta})\sum_{v\in\Lambda}\kappa_{v,\varphi_{v}}
\end{align}
and notice that on $\goodev$ for all $\varphi\in\Omega_{L}^{\Lambda}$ we have $B(\varphi)\ge(1-c)B_{0}\iff\varphi\in\Pi$.

Let $\pi$ be the function from Proposition \ref{prop:const_laplace}
applied to the box $\Lambda$ with $\ep=c$. Recall from the proposition
that $|\Lambda^{-}|/|\Lambda|\ge1-\ep$, thus $|\Lambda\setminus\Lambda^{-}|\le c|\Lambda|$.
Define the surface $s\coloneqq4h\pi e_{1}$. We record the crucial
fact that 
\begin{equation}
B_{0}=C_{0}\Omega(\|\nabla s\|^{2}).\label{eq:bountyisright}
\end{equation}
Indeed, we have that
\begin{equation}\label{eq:int_b_lb}
\int b(t)dt\ge\int\frac{(b(t))^{2}}{\max\{|b(t)|\}}dt\ge\frac{1}{\mathcal{L}}=\Omega(1),
\end{equation}
and therefore 
\[
B_{0}=\lambda\frac{\int b(t)dt}{{\rm Vol}(S)}\frac{a}{2}|\Lambda|=\Omega\left(\frac{\lambda a(\ell)}{{\rm Vol}(B_{h+1}(0))}\right)=\Omega\left(\frac{\lambda C_{0}\lambda^{-1}\ell^{2}h^{n+2}}{h^{n}}\right)=C_{0}\Omega(\ell^{2}h^{2}).
\]
On the other hand, by (\ref{eq:s_energy}), $\|\nabla s\|^{2}=(4h)^{2}O(\ell^{d-2})=O(h^{2}\ell^{2})$.
Also, by (\ref{eq:s_min}) and $h\ge1$ we have 
\begin{equation}
\forall v\in\Lambda^{-}\quad\forall t\in\R^{n}\quad\kappa_{v,t}+\kappa_{v,t\pm s_{v}}\le1.\label{eq:boundsumkappa}
\end{equation}

Let $\Delta_{1},\Delta_{2},\Delta_{3}$ be as in Proposition \ref{prop:Deloc2}.
We first lower-bound $\Delta_{1}$ by:
\begin{equation}
\begin{split}\Delta_{1}=\GE_{\Pi}^{\eta}-\GE_{\Pi}^{\zeta}\ge\inf_{\varphi\in\Pi}(H^{\eta}(\varphi)-H^{\zeta}(\varphi)) & =\inf_{\varphi\in\Pi}B(\varphi)\overset{\text{on \ensuremath{\goodev}}}{\ge}(1-c)B_{0}.\end{split}
\label{eq:bounddelta1}
\end{equation}
For $\varphi\in\Pi$, we have thanks to \eqref{eq:boundsumkappa}
\[
\sum_{v\in\Lambda}\kappa_{v,\varphi_{v}\pm s_{v}}\le|\Lambda^{-}|+2|\Lambda\setminus\Lambda^{-}|-\sum_{v\in\Lambda}\kappa_{v,\varphi_{v}}\le|\Lambda|(1+c-(1-c))=2c|\Lambda|.
\]
It yields that,
\begin{equation}
\begin{split}\Delta_{2}=\inf_{\varphi\in(\Pi+s)\cup(\Pi-s)}H^{\zeta}(\varphi)-H^{\eta}(\varphi) & \overset{\text{on \ensuremath{\goodev}}}{\ge}-\frac{B_{0}}{|\Lambda|}\sup_{\varphi\in(\Pi+s)\cup(\Pi-s)}\sum_{v\in\Lambda}\kappa_{v,\varphi_{v}}\ge-2cB_{0}.\end{split}
\label{eq:bounddelta2}
\end{equation}
Thus combining \eqref{eq:bounddelta1} and
\eqref{eq:bounddelta2} yields 
\begin{equation}
\goodev\subset\{\Delta_{1}+\Delta_{2}\ge(1-3c)B_{0}\}\label{eq:bounddelta12}
\end{equation}

By Corollary \ref{cor:markov} with $a=6$ (unrelated to the $a$
used in this proof), using that the sequence $(\zeta^{ks})_{k\in\Z}$
is stationary under the law $\overline{\P}$ we get 
\begin{equation}
\overline{\P}(\,\Delta_{3}\ge-4\|\nabla s\|^{2}\,)\ge\frac{2}{3}.\label{eq:bounddelta3}
\end{equation}
Combining (\ref{eq:bounddelta12}), (\ref{eq:bounddelta3}), (\ref{eq:bountyisright})
and Proposition \ref{prop:Deloc2}, gives
\begin{equation}
\overline{\P}(\varphi^{\eta}\notin\Pi)+\overline{\P}(\varphi^{\zeta^{s}}\notin\Pi)+\overline{\P}(\varphi^{\zeta^{-s}}\notin\Pi)\ge\frac{2}{3}-\overline{\P}(\goodev^{c}).\label{eq:deloc4dl1}
\end{equation}
We claim that 
\begin{equation}
\goodev\cap\{\varphi^{\zeta}\notin\Pi\}\subset\{\varphi^{\eta}\notin\Pi\}.\label{eq:claim monotony}
\end{equation}
Indeed assume by contradiction that $\varphi^{\eta}\in\Pi,\varphi^{\zeta}\notin\Pi$,
and $\goodev$ holds. Then 
\begin{align*}
B(\varphi^{\eta}) & \ge(1-c)B_{0}>B(\varphi^{\zeta})\\
 & \implies H^{\eta}(\varphi^{\eta})-H^{\zeta}(\varphi^{\eta})>H^{\eta}(\varphi^{\zeta})-H^{\zeta}(\varphi^{\zeta})\\
 & \implies0\ge\GE^{\eta}-H^{\eta}(\varphi^{\zeta})>H^{\zeta}(\varphi^{\eta})-\GE^{\zeta}\ge0.
\end{align*}
Using the stationarity of the sequence $(\zeta^{ks})_{k\in\Z}$ under
the law $\overline{\P}$ and (\ref{eq:claim monotony}) we get 
\[
\overline{\P}(\varphi^{\zeta^{s}}\notin\Pi)=\overline{\P}(\varphi^{\zeta^{-s}}\notin\Pi)=\overline{\P}(\varphi^{\zeta}\notin\Pi)\le\overline{\P}(\varphi^{\eta}\notin\Pi)+\overline{\P}(\goodev^{c})
\]
Combining the previous inequality together with inequality \eqref{eq:deloc4dl1}
gives 
\[
2\overline{\P}(\goodev^{c})+3\overline{\P}(\varphi^{\eta}\notin\Pi)\ge\frac{2}{3}-\overline{\P}(\goodev^{c}).
\]
Taking $C_{0}$ large enough gives
\[
\overline{\P}(\varphi^{\eta}\notin\Pi)\le\frac{2}{9}-\overline{\P}(\goodev^{c})\le\frac{1}{6}.
\]
\end{proof}

\subsection{Delocalization when $d\le 3$}\label{sec:delocalization in low dimensions} In this section, we prove Theorem~\ref{thm:delocalization intro} for $d\le 3$. In all of this section, $\eta = \eta^\white$ and $\Lambda=\Lambda_L$.
\global\long\def\minloc#1#2{\varphi^{#1,\overline{\underline{#2}}}}%
\global\long\def\GEloc#1#2{\GE^{#1,\overline{\underline{#2}}}}%
\global\long\def\loc#1#2#3{\Omega^{#1,#2,\overline{\underline{#3}}}}%

\global\long\def\indic#1{\mathbf{1}_{#1}}%
\global\long\def\condon{\,|\,}%
\global\long\def\N{\mathbb{N}}%
\global\long\def\minloc#1#2{\varphi^{#1,\overline{\underline{#2}}}}%
\global\long\def\GEloc#1#2{\GE^{#1,\overline{\underline{#2}}}}%
\global\long\def\loc#1#2#3{\Omega^{#1,\overline{\underline{#2}},#3}}%

For $h\in(0,\infty)$ and $p\in[0,1]$, say that a surface $\varphi$
is \emph{$p$-localized} to height $h$ if $\varphi\in\loc{\Lambda}hp$
where 
\begin{equation}
\loc{\Lambda}h{p\coloneqq}\left\{ \varphi\in\Omega^{\Lambda}:\frac{|\{v\in\Lambda:\|\varphi_{v}^{\eta}\|\le h\}|}{|\Lambda|}\ge p\right\} .\label{eq:loc_set}
\end{equation}
Otherwise, if $\varphi\notin\loc{\Lambda}hp$, then we say that $\varphi$
is \emph{$(1-p)$-delocalized}. The following theorem is a slightly stronger version of Theorem \ref{thm:delocalization intro} in the case where $d\in\{1,2,3\}$. It also applies in dimensions $d\ge 4$ where it provides additional information when $\lambda$ is large.
\begin{thm}[Delocalization, $\eta\sim\eta^{\white}$]\label{thm: deloc}
There exists $c_0,c_1>0$, depending only on $d,n$ and the Lipschitz constant $\mathcal{L}$ of the bump function in~\eqref{eq:eta white definition}, such that for $\lambda>0,L\ge 1$
and 
\begin{equation}
h\le c_0\max\{\lambda^{\frac{2}{4+n}}L^{\frac{4-d}{4+n}},\min(\lambda^ {-1},\lambda)^{1/2}L^{1-d/2}\}\label{eq:h_deloc}
\end{equation}
it holds that 
\[
\P\left(\varphi^{\eta}\notin\loc{\Lambda}h{51\%}\right)\ge c_1
\]
as long as $h\ge1$.
\end{thm}

We first prove Theorem \ref{thm: deloc} by combining Proposition \ref{prop:Deloc2}, Proposition \ref{prop:const_laplace}, Corollary \ref{cor:markov} and Proposition \ref{lem:bounty_construction}. Then, in the remainder of the current subsection, we prove Proposition \ref{lem:bounty_construction}. This proposition shows that there is a perturbation of $\eta^\white$ having the same distribution, such that with constant probability there is a ``ground energy gap'' between the original and the perturbed environment. This perturbation is obtained by reducing the average noise in a tube around the geodesic or in the region $\Lambda\times B_h(0)$, where $B_r(s)$ is the ball of radius $r$ defined in~\eqref{eq:ball def}. In this proposition and its proof, the notations $O,\Omega$, sometimes with parameters in subscript, hide constants that may depend on $d,n$, the Lipschitz constant $\mathcal L$
of the bump function in~\eqref{eq:eta white definition} and the parameters in subscript (if present).

\begin{prop}[$\eta\sim\eta^{\white}$]
\label{lem:bounty_construction} There exists $q>0$ such that for every large enough
$L\ge 1$,  every $h\ge 1$,
there is a coupling $(\eta,\zeta )$ of noise functions such that the marginals are distributed as $\eta^ {white}$ and
\begin{equation}
    \P\left(\GE_{\Pi}^{\eta}-\GE_{\Pi}^{\zeta}+\inf_{\varphi\in\Pi'}[H^\zeta(\varphi)-H^\eta(\varphi) ] \ge B\right)\ge q
\end{equation}
where $\Pi\coloneqq\loc{\Lambda}h{51\%}$, $\Pi'\coloneqq(\loc{\Lambda}{2h}{50\%})^c $ and
\begin{equation}
B=\Omega\left(\lambda\max\left\{\frac 1 {1+\lambda^2},\sqrt{\frac{L^{d}}{h^{n}}}\right\}\right).\label{eq:bounty_value}
\end{equation}
\end{prop}

\begin{proof}[Proof of Theorem \ref{thm: deloc}] Let $L\ge 1$ and let $h$ to be chosen later depending on $L$. Denote $\Pi\coloneqq\loc{\Lambda}h{51\%}$ and $\Pi'\coloneqq(\loc{\Lambda}{2h}{50\%})^c $. Let $q>0$ and $B$ as in the statement of Proposition \ref{lem:bounty_construction}.
Let $(\eta,\zeta)$ be the corresponding coupling of noise.
Let $e\in\R^n$.
Let $s=Ch\pi e$ where $\pi$ is as in Proposition \ref{prop:const_laplace} applied to the box $\Lambda=\Lambda_L$ and $\ep=1/100$ where $C$ is such that 
\[\min_{v\in\Lambda_L^-}Ch\pi_v> 3h.\]
Let $\Delta_1,\Delta_2,\Delta_3$ be as in Proposition \ref{prop:Deloc2}.
We claim that for every $\varphi\in\Pi$, we have $\varphi+s,\varphi- s\in \Pi'$.
Indeed, for every $v\in \Lambda_L^-=\Lambda^-$ such that $\|\varphi_v\|\le h$, we have $\|\varphi_v+s_v\|>2h$.
It yields 
\begin{equation*}
    \begin{split}
        |\{v\in\Lambda:\|\varphi_{v}+s_v\|\le 2h\}|&\le |\Lambda|- |\{v\in\Lambda^-:\|\varphi_{v}\|\le h\}|\\&\le |\Lambda|- |\{v\in\Lambda:\|\varphi_{v}\|\le h\}|+ |\Lambda\setminus \Lambda^-|\\
        &\le  \frac {101}{100}|\Lambda|-  \frac {51}{100}|\Lambda|=\frac 12 |\Lambda| 
    \end{split}
\end{equation*}
where we used in the last inequality that $|\Lambda\setminus \Lambda^-|\le \frac 1 {100}|\Lambda|$ and that $\varphi\in\Pi$.
Hence, if $\varphi\in\Pi$ then $\varphi+s\in\Pi'$ concluding the claim.
As a result, we get
\begin{equation*}
    \Delta_1+\Delta_2\ge \GE_{\Pi}^{\eta}-\GE_{\Pi}^{\zeta}+\inf_{\varphi\in\Pi'}[H^\zeta(\varphi)-H^\eta(\varphi) ]
\end{equation*}
and
\begin{equation}\label{eq:delta1}
    \P(\Delta_1+\Delta_2\ge B)\ge q.
\end{equation}
Let $a>0$ such that $\frac 1 a <\frac q 2$. By applying Corollary \ref{cor:markov} to $2a$, we get
\begin{equation*}
    \P(\Delta_3\ge -2(a-1)\|\nabla s\|^ 2)\ge 1- \frac q 2.
\end{equation*}
Let us now choose $h$ such that
\begin{equation}\label{eq:choice h}
    B-2(a-1)\|\nabla s\|^ 2\ge \|\nabla s\|^ 2.
\end{equation}
There exists $C>0$ such that $\|\nabla s\|^ 2\le Ch^ 2L^ {d-2}$.
It yields that 
\begin{equation*}
h\le c_0\max\{\lambda^{\frac{2}{4+n}}L^{\frac{4-d}{4+n}},\min(\lambda,\lambda^{-1})^{1/2}L^{1-d/2}\}
\end{equation*}
Let us assume that 
\[\P\left(\varphi^{\eta}\notin\Pi\right)\le \frac q 8.\]
If the previous inequality does not hold then the result follows. Note that we have the following inclusion
\[\{\varphi^{\zeta^ s}\in\Pi\}\cap\{\varphi^{\zeta^ {-s}}\in\Pi\} \cap \{(\varphi^ {\zeta^s}-\varphi^{\zeta^{-s}},-\Delta s)\ge -2(a-1)\|\nabla s\|^ 2\}\subset \{\Delta_3\ge -2(a-1)\|\nabla s\|^ 2\}. \]
It yields that
\[\P(\Delta_3\ge -2(a-1)\|\nabla s\|^ 2)\ge  1-\frac  {3}4q.\]
Combining the previous inequality together with \eqref{eq:delta1} and \eqref{eq:choice h} yields
\begin{equation*}
    \P(\Delta_1+\Delta_2+\Delta_3>\|\nabla s\|^ 2)\ge \frac q 4.
\end{equation*}
Thanks to Proposition \ref{prop:Deloc2}, it yields that 
\[\P\left(\varphi^{\eta}\notin\Pi\right)\ge \frac q {12}.\]
This concludes the proof.
\end{proof}

\subsubsection{Proof of Proposition \ref{lem:bounty_construction}
}
\global\long\def\law{\mathcal{L}}

We want to build a new noise by reducing the noise in a given region, without changing its distribution by much. To do so we need to control the Radon--Nikodim derivative between the distributions of these noises. Let
$h\ge 1$. Let $L\ge 1$. Let $\Pi\coloneqq\loc{\Lambda}h{51\%}$ and
$\Pi'\coloneqq\Omega^{\Lambda}\setminus\loc{\Lambda}{2h}{49\%}$.
\smallskip

\paragraph{Standard measure theoretic lemmas} In this section, we prove measure theoretic results useful for the proof of Proposition \ref{lem:bounty_construction}.
Denote by $\law_{(\cdot)}$ the law of a process and by $\frac{d(\cdot)}{d(\cdot)}$ the Radon--Nikodym derivative.
\begin{fact}
\label{prop:radnik_indep}Let $X,X'$ and $Y$ be random variables
and assume that $(X,X')$ is independent from $Y$. Then 
\[
\frac{d\law_{(X',Y)}}{d\law_{(X,Y)}}(X',Y)=\frac{d\law_{X'}}{d\law_{X}}(X')
\]
holds almost surely, assuming that the RHS is well defined. 
\end{fact}

\begin{prop}
\label{prop:translation_deriv}Let $W$ be distributed as a white
noise on a measure space $E$ and let $f\in L^{2}(E)$. Then
\[
\frac{d\law_{W+f}}{d\law_{W}}(w+f)=e^{w(f)+\frac{1}{2}\int f^{2}}.
\]

\end{prop}
\begin{proof}
Let $L^{2}(E)^{\#}$ be the set of linear functionals of $L^{2}(E)$.
Define $h:\R\to\R$ by $h(x)\coloneqq e^{x-\frac{1}{2}\int f^{2}}$.%

Write $X\coloneqq W(f)\sim N(0,\int f^{2})$. Thus 
\[
\frac{d\law_{X+\int f^{2}}}{d\law_{X}}(x)=\frac{\exp\left(-\frac{(x-\int f^{2})^{2}}{2\int f^{2}}\right)}{\exp\left(-\frac{x^{2}}{2\int f^{2}}\right)}=h(x).
\]
Write $f^{\perp}$ for the orthogonal complement of $f$ in $L^{2}(E)$,
and $p_{f}$ for the orthogonal projection onto $f^{\perp}$. By properties
of white noise, $W(f)$ is independent from $W|_{f^{\perp}}$. Also
note that $(W+f)|_{f^{\perp}}=W|_{f^{\perp}}$. Applying Fact
\ref{prop:radnik_indep} to $X,X+\int f^{2}$ and $W|_{f^{\perp}}$,
and considering the bijection $w\mapsto(w(f),w|_{f^{\perp}})$, the proposition follows: 
\begin{align*}
\frac{d\law_{W+f}}{d\law_{W}}(w+f) & =\frac{d\law_{(X+\int f^{2},W|_{f^{\perp}})}}{d\law_{(X,W|_{f^{\perp}})}}(w(f)+\int f^{2},w|_{f^{\perp}})\\
 & =h(w(f)+\int f^{2})=e^{w(f)+\int f^{2}-\frac{1}{2}\int f^{2}}.
\end{align*}

\end{proof}
Let $(\Omega_1,\mathcal{F}_1,\P_1)$ be a measure space and let $(\Omega_2,\mathcal{F}_2)$ be a measurable space. Let $m:\Omega_1\to\Omega_2$ be measurable. Denote by $m_*\P$ the push-forward measure defined on $(\Omega_2,\mathcal{F}_2)$ by $m_*\P(A)\coloneqq\P(m^{-1}(A))$. The following lemma controls the total variation distance between a measure and its push-forward by $m$ assuming that the Radon--Nikodym derivative is close to $1$ with high probability.
\begin{lem}
\label{prop:deriv_TV}Let $(\Omega,\mathcal{F},\P)$ be a probability
space and let $m:\Omega\to\Omega$ measurable. Suppose for some $c>0$
that $m_*\P(\frac{dm_*\P}{d\P}>1+c)<c$. Then $d_{\mathrm{TV}}(\P,m_*\P)\le2c$. 
\end{lem}

\begin{proof}
Let $B=\{\frac{dm_*\P}{d\P}>1+c\}$. Then $m_*\P(B)<c$. Let $A$
be an event. Then 
\begin{align*}
m_*\P(A) & =m_*\P(A\cap B^{c})+m_*\P(A\cap B)\\
 & \le\int_{A\cap B^{c}}\frac{dm_*\P}{d\P}d\P+m_*\P(B)\\
 & \le(1+c)\P(A)+c\le\P(A)+2c
\end{align*}
By considering $A^{c}$ we get also $m_*\P(A)\ge\P(A)-2c$. 
\end{proof}
\begin{lem}
\label{prop:injective_union}Let $m$ be an injective function from
a subset of a probability space $\Omega$ into itself, and assume
that $m$ is a union of a countable set of bijections 
$\{m_{i}:E_{i}\to\Omega\}_{i\in\N}$
where $E_{i}$ is an event and $m_{i}$ and its inverse are measurable.
If $\frac{dm_{i*}\mu}{d\mu}\le a$ almost surely for each $i$,
then $\frac{dm_*\mu}{d\mu}\le a$. 
\end{lem}

\begin{proof}
For an event $F$, we are given that $\mu(m_{i}^{-1}(F))=m_{i*}\mu(F)=\int_{F}\frac{dm_*\mu}{d\mu}\le a\mu(F)$.
The sets $F_{i}\coloneqq m(m_{i}^{-1}(F))$ are measurable, and since
$m$ is injective, they are pairwise disjoint. 
\[
m_*\mu(F)=\mu(m^{-1}(F))=\mu(\bigcup m_{i}^{-1}(F))\le\sum\mu(m_{i}^{-1}(F_{i}))\le\sum a\mu(F_{i})\le a\mu(F).
\]
\end{proof}

\paragraph{Two specific mappings}\label{sec:two specific mappings}
In this section, we present two mappings the first one corresponding to reducing the noise on a region of the form $\Lambda\times B_{h+1}(0)$ the other corresponding to reducing the noise around the minimal surface.
\begin{defn}
Let $\alpha>0$ be a parameter. Let $f\in\{f_{1},f_{2}\}$ be a random
function where: 
\begin{align*}
f_{1} & =-\frac{\alpha}{\sqrt{L^{d}h^{n}}}\indic{\Lambda\times B_{h+1}(0)}\\
f_{2} & =-\frac{\alpha}{L^{d}}\sum_{v\in\Lambda}\indic{\{v\}\times B_{2}(\left\lfloor \varphi^{\eta}_{v}\right\rfloor )}
\end{align*}
where $\left\lfloor t\right\rfloor \coloneqq\left(\left\lfloor t_{i}\right\rfloor \right)_{i=1}^{n}$.

Let $m_f$ be a mapping from the probability space of white noise
on $\Z^{d}\times\R^{n}$ to itself, defined by $m_f(W)=W+f$. Equivalently
\[
(m_f\eta)_{v,t}=\eta_{v,t}+\int_{\R^{n}}f(t-x)b(x)dx.
\]
\end{defn}
The purpose of the following lemma is to control the Radon--Nikodym derivative of the push-forward measure by $m_f$ with respect to the original measure when $f=f_1$ and $f=f_2$. The main difficulty is in the case $f=f_2$ in which we reduce the noise near the minimal surface. For that purpose, we bound the integral of the noise in the neighborhood of the minimal surface (corresponding to $W(f_2)$). This is done by union bounding on all the possible values for $f_2$.
\begin{lem}
\label{lem:deriv_bound}Let $a\in(0,1)$. There is $\alpha_{0}=\begin{cases}
\Omega_{a}(1) & f=f_{1}\\
\Omega_{a}(\frac {1}{1+\lambda^2}) & f=f_{2}
\end{cases}$ such that for $\alpha\le\alpha_{0}$, 
\[
\P\left(W(f)+\frac{1}{2}\int f^{2}>a\right)<a.
\]
\end{lem}
\begin{proof}
Write $\sigma^{2}=\int f^{2}$. For $f=f_{1}$: 
\[
\sigma^{2}=\int f^{2}=|\Lambda||B_{h+1}(0)|\left(\frac{\alpha}{\sqrt{L^{d}h^{n}}}\right)^{2}=O(\alpha^{2})\cdot
\]
Thus $W(f)\sim N(0,\sigma^{2})$ and the lemma follows for this case.
Thus by Chebyshev's inequality, for $\alpha\ll a\le1$, 
\[
\P\left(W(f)+\frac{1}{2}\int f^{2}>a\right)=\P\left(W(f)>(\frac{a}{\sigma}-\frac{\sigma}{2})\sigma\right)\le\frac{1}{\frac{a}{\sigma}-\frac{\sigma}{2}}=O(\frac{\alpha}{a}).
\]

For $f=f_{2}$:
For every $\varphi\in\Z^d\to \Z^n$
\[\P\left (\sum_{v\in\Lambda}\int_{t\in B_2(\varphi_v)}W(t)dt<-k|\Lambda|\right)\le e^{-\Omega(ck^{2}|\Lambda|)}\]
By Corollary \ref{lem:nabla}, we have for $k\ge \lambda^2K_0$
\[ \mathbb P \big(  \| \nabla \varphi ^{\eta } \|_{\Lambda }^2 \ge k|\Lambda | \big) \le \exp \big( -c\frac{k^2}{\lambda^2}|\Lambda | \big).\]
We bound $\P(W(f)>\alpha k)$ by borrowing some ideas and notations from the proof of Lemma \ref{lem:conc_lem}. Denote
\begin{align*}
    \Omega_k &\coloneqq \big\{ \varphi \in \Omega ^{\Lambda,0}: \|\nabla \varphi \| _{\Lambda}^2 \le k|\Lambda| \big\} \\
    \widetilde \Omega_k &\coloneqq \big\{ \varphi \in \Omega_k: \varphi:\Z^d\to\Z^n \big\} \\
    \tilde{\varphi} &\coloneqq \lfloor\varphi ^{\eta } \rfloor.
\end{align*}
Assume $k>C$. By inequality \eqref{eq:norml2discretephi}, we have that $\varphi^\eta\in\Omega_k$, implies 
$\| \nabla \tilde{\varphi} \|_{\Lambda }^2\le 2k|\Lambda|$.
Thanks to \eqref{eq:contomegak}, we have
\[|\widetilde \Omega_{2k}|\le e^{4k|\Lambda|}.\]
It follows that 
\begin{equation*}
    \begin{split}
        \P(W(f)>\alpha k)&\le \P\big(  \| \nabla \varphi ^{\eta } \|_{\Lambda }^2 \ge k|\Lambda | \big) +\P\left(\exists\varphi\in \widetilde\Omega_{2k},\quad \sum_{v\in\Lambda}\int_{t\in B_2(\varphi_v)}W(t)dt<-k|\Lambda|\right)\\
        &\le  e^{ -c\frac{k^2}{\lambda^2}|\Lambda | }+ e^{-ck^2|\Lambda|}.
    \end{split}
\end{equation*}
Also $\int f^{2}\le|\Lambda|4^{n}\left(\frac{\alpha}{L^{d}}\right)^{2}=O(\alpha^{2}L^{-d})$.
Thus taking $k= \max\{C,\frac 1 c \max\{1,\lambda^2K_0\}\log \frac 2 a\}$, the RHS of the previous display is $<a$
and taking $\alpha_{0}<\frac{a}{2k}$
\[
\P\left(W(f)+\frac{1}{2}\int f^{2}>a\right)\le \P\left(W(f)>\frac a 2\right)\le\P(W(f)>\alpha k)\le a.
\]
This concludes the proof of the lemma.
\end{proof}
\begin{lem}
\label{lem:m_is_injective}$m_f$ is injective for $f=f_1$ and $f_2$.
\end{lem}
\begin{proof}
For $f=f_{1}$ this holds trivially since $m_f$ is a translation. Assume
$f=f_{2}$. Then injectivity follows from the fact that $\varphi^{m_f\eta,\Pi}=\varphi^{\eta,\Pi}$
and $m_f\eta-\eta$ is determined by $\varphi^{\eta,\Pi}$. This equality
follows since for each $\varphi\in\Omega^{\Lambda}$, it holds that
\[
\eta(\varphi)-m_f\eta(\varphi)\le\eta(\varphi^{\eta,\Pi})-m_f\eta(\varphi^{\eta,\Pi}).
\]
Indeed since $\forall t,b(t)\ge0$ and $\int b(t)dt=1$, it holds
that $0\le\eta_{v,t}-(m_f\eta)_{v,t}\le\max(f)$, and for each $(v,t)$
with $t=\varphi_{v}^{\eta,\Pi}$ in the graph of $\varphi^{\eta,\Pi}$
it holds that $\eta_{v,t}-(m_f\eta)_{v,t}=\max(f)$. 
\end{proof}
The purpose of the following lemma is to obtain a lower bound on the ``ground energy gap'' when reducing the average noise using $f_1$ or $f_2$.
\begin{lem}

\label{lem:directed_bounty}Let $\alpha>0$. There is 

\begin{equation}\label{eq:B}
    B=\begin{cases}
\Omega\left(\alpha\lambda\sqrt{\frac{L^{d}}{h^{n}}}\right) & f=f_{1}\\
\Omega\left(\alpha\lambda\right) & f=f_{2}
\end{cases}
\end{equation}
such that for every deterministic function $\eta:\Z^{d}\times\R^{n}\to\R$,
\begin{equation}
\GE_{\Pi}^{\eta}-\GE_{\Pi}^{m_f\eta}+\inf_{\varphi\in \Pi'}[H^{m_f\eta}(\varphi)-H^{\eta}(\varphi)]\ge B.\label{eq:m_bounty}
\end{equation}
\end{lem}

\begin{proof}
For the case $f=f_{1}$: Put $B=\frac{\alpha\lambda|\Lambda|\int b(t)dt}{100\sqrt{L^{d}h^{n}}}$, which is as in (\ref{eq:B}), (using (\ref{eq:int_b_lb})).
Write $A^{\varphi,h}\coloneqq\left\{ v\in\Lambda:\|\varphi_{v}\|\le h\right\} $.
Write $\Delta=\lambda\int b(t)dt \max(-f)=\frac{\alpha\lambda \int b(t)dt}{\sqrt{L^{d}h^{n}}}=\frac{100B}{|\Lambda|}$. For any $\varphi$,
\[
-|A^{\varphi,h+1}|\Delta\le H^{m_f\eta}(\varphi)-H^{\eta}(\varphi)=\lambda\sum_{v\in\Lambda}\int b(t)f(\varphi_{v}-t)dt\le- |A^{\varphi,h}| \Delta
\]

For $\varphi=\varphi^{\eta,\Pi}$, it holds by (\ref{eq:loc_set})
that, $|A^{\varphi,h}|\ge0.51|\Lambda|$. For $\varphi\in\Pi'$,
it holds that $|A^{\varphi,h+1}|<|A^{\varphi,2h}|<0.50|\Lambda|$.
Substituting gives 
\begin{align*}
\GE_{\Pi}^{\eta}-\GE_{\Pi}^{m_f\eta} & \ge51B\text{ and}\\
\inf_{\varphi\in\Pi'}\left[H^{m_f\eta}(\varphi)-H^{\eta}(\varphi)\right] & \ge-50B.
\end{align*}
Combining the two previous inequalities yields the result.

For the case $f=f_{2}$: Put $B=\frac{\alpha\lambda|\Lambda|\int b(t)dt}{100L^{d}}$.
Put $$\Delta=\lambda\int b(t)dt \max(-f)=\lambda\int b(t)dt\frac{\alpha}{L^{d}}=\frac{100B}{|\Lambda|}.$$For $\varphi=\varphi^{\eta,\Pi}$
it holds that $f(\varphi_{v}-t)=-\Delta$ whenever $\|t\|\le1.$ Thus
\begin{equation} \label{eq:100B}
\GE_{\Pi}^{\eta}-\GE_{\Pi}^{m_f\eta}\ge H^{\eta}(\varphi)-H^{m_f\eta}(\varphi)=|\Lambda|\Delta\ge100B.
\end{equation}
For $\varphi\in\Pi'$,
it holds that 
\begin{align*}
|\left\{ v\in\Lambda:\|\varphi_{v}-\varphi_{v}^{\eta,\Pi}\|\le2\right\}|  & \le|\left\{ v\in\Lambda:\varphi_{v},\varphi_{v}^{\eta,\Pi}\in B_{2h}(0)\right\} |\\
 & +|\left\{ v\in\Lambda:\varphi_{v},\varphi_{v}^{\eta,\Pi}\notin B_{h}(0)\right\}| \\
 & \le0.50|\Lambda|+0.49|\Lambda|
\end{align*}
and therefore
\begin{align*}
\inf_{\varphi\in\Pi'}\left[H^{m_f\eta}(\varphi)-H^{\eta}(\varphi)\right] & \ge-0.99|\Lambda|\Delta \ge-99B
\end{align*}
which, together with (\ref{eq:100B}) yields the result.
\end{proof}
\begin{lem}
\label{lem:m_TV}There is $\alpha_{0}=\begin{cases}
\Omega_{\ep}(1) & f=f_{1}\\
\Omega_{\ep}(\frac 1{1+\lambda^2}) & f=f_{2}
\end{cases}$ such that, $d_{\mathrm{TV}}(\eta,m_f\eta)\le\ep$ when $\alpha\le\alpha_{0}$.
\end{lem}

\begin{proof}
Let $(g_{i})_{i\in\N}$ be the values that $f$ may take (as a random
variable). Let $a>0$ and denote $\mathcal{E}\coloneqq\{W(f)+\frac{1}{2}\int f^{2}\le a\}$.
Define $E_{i}\coloneqq\{f=g_{i}\}\cap\mathcal{E}$ and the restrictions
$m_f^{(i)}=m_f|_{E_{i}}$. By Proposition \ref{prop:translation_deriv},
for $w\in E_{i}$ 
\[
\frac{d\law_{m_f^{(i)}\eta}}{d\law_{\eta}}(w+f)=e^{w(f)+\frac{1}{2}\int f^{2}}\le e^{a}
\]
The restrictions are translations, and therefore are bijective and
measurable with measurable inverses. By Lemma \ref{lem:m_is_injective},
$m_f|_{\mathcal{E}}$ is injective. Let $\tilde{\eta}$ be distributed
as $\eta$, conditioned on $\mathcal{E}$. Thus by Lemma \ref{prop:injective_union},
$\frac{d\law_{m_f\tilde{\eta}}}{d\law_{\eta}}\le\frac{1}{\P(\mathcal{E})}e^{a}\le\frac{e^{a}}{1-a}$
almost surely. Let $\alpha_{0}$ be as in Lemma \ref{lem:deriv_bound},
so that $\P(\mathcal{E})\ge1-a$. Finally by Lemma \ref{prop:deriv_TV},
$d_{\mathrm{TV}}(\eta,m_f\eta)\le2\max\{\frac{e^{a}}{1-a}-1,a\}$. We
may choose $a$ such that this is at most $\ep$ and $\alpha_{0}$
is as in the current lemma. 
\end{proof}

\paragraph{Proof of Proposition \ref{lem:bounty_construction}}

Let $\eta\sim\eta^{\white}$. Let $\ep>0$ and let $\alpha=\alpha_{0}$
for $\alpha_{0}$ as in Lemma \ref{lem:m_TV}. Then $d_{\mathrm{TV}}(\mu,m_f\mu)\le\ep$
and thus there is a coupling $(\eta_{1},\eta_{2})$ such that $\eta_{1},\eta_{2}\sim\eta$
and $\P(\eta_{2}=m_f\eta_{1})\ge1-\ep$ (this is folklore according
to \cite{angel2021pairwise} where it appears as theorem 1). 
Choose $\ep$ such that $q\le1-\ep$.
Then choosing $f$ from $\{f_{1},f_{2}\}$ to maximize $B$ from Lemma
\ref{lem:directed_bounty} it follows that $B$ satisfies (\ref{eq:bounty_value}).

 \subsection{Delocalization to constant height} We prove here the following simple delocalization result. The result is tight for constant $\lambda$ in dimensions $d\ge 5$ by Theorem~\ref{thm:localization}.

 The statement in Theorem~\ref{thm:delocalization intro} for $d\ge 5$ is an easy corollary of the following proposition.
 \begin{prop}[$\eta=\eta^\white$] \label{prop:deloc d>=5}For any $\lambda>0$, there exists a positive constant $c_{\lambda}$, depending only on $d,n,\lambda$ and the Lipschitz constant $\mathcal{L}$ of the bump function in~\eqref{eq:eta white definition}, such that for all $\Lambda\subset \Z^d$, all boundary conditions $\tau$ and all $v\in\Lambda$,
     \[\P(\|\varphi^\eta_v\| \ge 1)\ge c_{\lambda}.\]
 \end{prop}
 \begin{proof}
 Let $\Lambda\subset \Z^d$ and consider boundary condition $\tau$.
 Let $v\in\Lambda$ and let $\varphi$ be any surface with the given boundary values $\tau$. Define two new surfaces $\varphi^+$ and $\varphi^-$ as follows:
$\varphi^+$ is equal to $\varphi$ everywhere except at $v$, where it equals $\varphi_v + 5e_1$,
$\varphi^-$ is equal to $\varphi$ everywhere except at $v$, where it equals $\varphi_v - 5e_1$.
Let us calculate the difference between the energy of $\varphi$ and the average of the energies of $\varphi^+$ and $\varphi^-$:
\begin{equation}\label{eq:phipm}
    H^\eta(\varphi) - \frac{1}{2} (H^\eta(\varphi^+) + H^\eta(\varphi^-)) = -50d + \lambda \left(\eta_{v, \varphi_v} - \frac{1}{2} (\eta_{v, \varphi^+_v} + \eta_{v, \varphi^-_v})\right).  
\end{equation}
Let us write
$m_v \coloneqq \min\{\eta_{v,t} : \|t\| < 1\}$ and
$M_v \coloneqq \max\{\eta_{v,t} : 4 < \|t\| < 6\}$.
Thanks to \eqref{eq:phipm}, we have the following inclusion
\[\{|\varphi^\eta_v| < 1\} \subset \left\{H^\eta(\varphi^\eta) - \frac12 (H^\eta((\varphi^\eta)^+) + H^\eta((\varphi^\eta)^-))\le 0\right\}
 \subset \{M_v \ge m_v - 50d/\lambda\}\]
where $\varphi^\eta$ is the minimal surface with disorder $\eta$. Note that $m_v$ and $M_v$ are independent for $\eta^\white$. Hence, it is easy to check that for $\eta^\white$, we have $\P(M_v < m_v - 50d/\lambda)>0$ uniformly for all $v$ and $\Lambda$. The proposition follows.
\end{proof}
\subsection{Lower bound for ground energy fluctuations }
In this section, we prove Theorem~\ref{thm:lower bounds on energy fluctuations by localization} using as a main ingredient the noise decomposition introduced in Lemma~\ref{lem : decomposition noise-1}. In all this section, $\eta = \eta^\white$.

\begin{proof}[Proof of Theorem~\ref{thm:lower bounds on energy fluctuations by localization}]
Let $\Lambda^0\subset \Lambda^1\subset \Z^d$.
Let $h\ge 1$ and $\delta>0$ such that
\begin{equation*}
     \frac{\mathbb E \big[ | \{v\in \Lambda^0 : \|\varphi^{\eta,\lambda,\Lambda^1}_v\| \le h \} | \big]}{|\Lambda^0|} \ge \delta.
\end{equation*}
Set 
\begin{equation*}
    \Pi\coloneqq\left\{\varphi\in\Omega^{\Lambda^1}:\frac{|\{v\in\Lambda^0:\|\varphi_{v}\|\le h\}|}{|\Lambda^0|}\ge \frac \delta 2\right\} .
\end{equation*}
In particular, we have
\begin{equation*}
   |\Lambda^0| \P(\varphi ^{\eta,\lambda,\Lambda^1}\in \Pi)+ \frac \delta 2 |\Lambda^0|\ge  \mathbb E \big[ | \{v\in \Lambda^0 : \|\varphi^{\eta,\lambda,\Lambda^1}_v\| \le h \} | \big]\ge \delta |\Lambda^0|.
\end{equation*}
It yields
\begin{equation}\label{eq: lower bound for phi in Pi}
    \P(\varphi ^{\eta,\lambda,\Lambda^1}\in \Pi)\ge \frac \delta 2.
\end{equation}

We will now use the noise decomposition introduced in Lemma \ref{lem : decomposition noise-1} with $\rm S= \Lambda^0\times B_{h+1}(0)$. Let $\eta^{\rm S}$ and $(\eta^\perp_{v,t})_{v\in \Lambda^1,t\in\R^n}$ be the corresponding noise decomposition.
By Markov's inequality and inequality \eqref{eq: lower bound for phi in Pi}, we have
\begin{equation*}
  \P\left( \P \big(  \varphi ^{\eta,\lambda,\Lambda^1}\notin \Pi \mid (\eta^\perp_{v,t})_{v\in \Lambda^1,t\in\R^n} \big) \ge 1-\frac \delta 4\right)\le \frac {1-\frac \delta 2}{ 1-\frac \delta 4}\le 1-\frac \delta 4.
\end{equation*}
Hence, with probability at least $\delta/4$, we have that the family $(\eta^\perp_{v,t})_{v\in \Lambda^1,t\in\R^n}$ is such that 
\begin{equation*}
    \P \big( \varphi ^{\eta,\lambda,\Lambda^1}\in \Pi \mid (\eta^\perp_{v,t} )_{v\in \Lambda^1,t\in\R^n} \big) \ge \frac \delta 4.
\end{equation*}
Let us now consider a family $(\eta^\perp_{v,t})_{v\in \Lambda^1,t\in\R^n}$ such that the latter inequality holds.
Let $\eta^{\rm S}_0$ and $\eta^{\rm S}_1$ be two independent copies of $\eta^{\rm S} $.
Denote by $\eta^0$ and $\eta^1$ the two corresponding disorders.
Thanks to the previous inequality, we get
\begin{equation*}
    \P \big(  \varphi ^{\eta^0,\lambda,\Lambda^1},\varphi ^{\eta^1,\lambda,\Lambda^1}\in \Pi \mid (\eta^\perp_{v,t})_{v\in \Lambda^1,t\in\R^n} \big) \ge \frac {\delta^2} {16}.
\end{equation*}
Moreover, since $\eta^{{\rm S}}\sim\mathcal{N}(0,\frac{\left(\int b(t)dt\right)^{2}}{{\rm {Vol}({\rm S)}}})$, there exists $c_\delta>0$ depending on $n,\delta$ and the Lipschitz constant $\mathcal{L}$ of the bump function in~\eqref{eq:eta white definition} such that
\begin{equation*}
    \P\left(|\eta^{\rm S}_0-\eta^{\rm S}_1|\ge \frac{c_\delta }{|\Lambda^0|^{1/2}h^{n/2}}\right)\ge 1- \frac {\delta^2} {32}.
\end{equation*}
By the previous inequalities, with probability at least $\delta ^3/ 2^7$, we have 
$ \varphi ^{\eta^0,\lambda,\Lambda^1},\varphi ^{\eta^1,\lambda,\Lambda^1}\in \Pi $ and $|\eta^{\rm S}_0-\eta^{\rm S}_1|\ge \frac{c_\delta }{|\Lambda^0|^{1/2}h^{n/2}}$.
For short, we write $\varphi^0,\varphi^1$ for $\varphi ^{\eta^0,\lambda,\Lambda^1},\varphi ^{\eta^1,\lambda,\Lambda^1}$.
Without loss of generality, assume that $\eta^{\rm S}_0\le \eta^{\rm S}_1$. Since $\kappa_{v,t}=1 $ for $v\in\Lambda^0$ and $\|t\|\le h$, it yields that
\begin{equation*}
\sum_{v\in\Lambda^1}\eta^0_{v,\varphi_v^1}\le \frac \delta 2 |\Lambda_0| (\eta^{\rm S}_0- \eta^{\rm S}_1)+ \sum_{v\in\Lambda^1}\eta^1_{v,\varphi_v^1}
\end{equation*}
and
\begin{equation*}
\GE^{\eta^0,\lambda,\Lambda^1}\le \GE^{\eta^1,\lambda,\Lambda^1}- \lambda \delta c_\delta \sqrt{\frac {|\Lambda^0|}{h^n}}.
\end{equation*}
We conclude that there exists a coupling of the disorders $\eta^0,\eta^1$ such that the marginals are distributed as $\eta $ and
\begin{equation*}
\P\left(|\GE^{\eta^0,\lambda,\Lambda^1}- \GE^{\eta^1,\lambda,\Lambda^1}|\ge  \lambda \delta c_\delta \sqrt{\frac {|\Lambda^0|}{h^n}}\right)\ge \frac {\delta ^3}{2^7}.
\end{equation*}
Finally, note that we have the following inclusion
\begin{equation*}
\begin{split}
   & \left\{|\GE^{\eta^0,\lambda,\Lambda^1}- \med(\GE^{\eta,\lambda,\Lambda^1})|\le  \frac 12\lambda \delta c_\delta \sqrt{\frac {|\Lambda^0|}{h^n}}\right\}\cap  \left\{|\GE^{\eta^1,\lambda,\Lambda^1}- \med(\GE^{\eta,\lambda,\Lambda^1})|\le  \frac 12\lambda \delta c_\delta \sqrt{\frac {|\Lambda^0|}{h^n}}\right\}\\&\hspace{3cm}\subset \left \{|\GE^{\eta^0,\lambda,\Lambda^1}- \GE^{\eta^1,\lambda,\Lambda^1}|\le  \lambda \delta c_\delta \sqrt{\frac {|\Lambda^0|}{h^n}}\right\}.
    \end{split}
\end{equation*}
This yields that 
\begin{equation*}
\P\left(|\GE^{\eta^0,\lambda,\Lambda^1}- \med(\GE^{\eta,\lambda,\Lambda^1})|\ge  \frac 12\lambda \delta c_\delta \sqrt{\frac {|\Lambda^0|}{h^n}}\right)\ge \frac {\delta^3}{2^8}
\end{equation*}
concluding the proof.
\end{proof}

\section{Verifying the assumptions for the disorder $\eta^{\white}$}\label{sec:verifying assumptions}

\label{sec:assumptions for eta white} In this section, we prove that
the disorder $\eta^{\white}$ satisfies the assumptions \ref{as:exiuni}
and \ref{as:conc}. Note that $\eta^{\white}$ trivially satisfies
the assumptions \ref{as:indep} and \ref{as:stat}. 

Throughout this
section, we let $\lambda>0,\Lambda\subset\Z^{d}$, $\tau:\Z^{d}\to\R^{n}$
and $\eta\sim\eta^{\white}$. For short, we will write $\varphi^{\eta}=\varphi^{\eta,\lambda,\Lambda,\tau}$
and $\GE^{\eta}=\GE^{\eta,\lambda,\Lambda,\tau}$. First note that
thanks to Corollary~\ref{cor:effect of boundary conditions}, it will be enough to
check that the assumptions \ref{as:exiuni} and \ref{as:conc} hold
for $\tau=0$ (to check~\ref{as:conc} for general $\tau$,
one uses~\eqref{eq:tal} with $\tau=0$ and the disorder $\eta^{-\bar{\tau}^{\Lambda}}$, which has the same distribution as $\eta$). From now on, we will assume that
$\tau=0$. We start with the following lemma. Denote: 
\begin{equation}
\Omega_{k}:=\big\{\varphi\in\Omega^{\Lambda,\tau}:\|\nabla\varphi\|_{\Lambda}^{2}\le k|\Lambda|\big\}.\label{eq:Omega_k}
\end{equation}
\begin{lemma}\label{lem:conc_lem} There exist $C,c>0$ depending
only on $d,n,\mathcal{L}$ such that for all $k\ge C\max\{1,\lambda^{2}\}$
\begin{equation}
\mathbb{P}\Big(\exists\varphi\in\Omega_{k},\ \lambda\sum_{v\in\Lambda}\eta_{v,\varphi_{v}}\le-k|\Lambda|/4\Big)\le\exp\Big(-c\frac{k^{2}}{\lambda^{2}}|\Lambda|\Big).\label{eq:926}
\end{equation}
\end{lemma} 
\begin{cor}\label{cor:conc}
\label{lem:nabla}There exist $C,c>0$ depending only on $d,n,\mathcal{L}$
such that for all $k\ge C\max\{1,\lambda^{2}\}$ 
\begin{equation}
\mathbb{P}\big(\GE^{\eta}<-k|\Lambda|\big)\le\exp\Big(-c\frac{k^{2}}{\lambda^{2}}|\Lambda|\Big).\label{eq:conc_cor1}
\end{equation}
Moreover, 
\begin{equation}
\mathbb{P}\Big(\inf_{\varphi\notin\Omega_{k}}H^{\eta}(\varphi)\le\frac{k}{2}|\Lambda|\Big)\le\exp\Big(-c\frac{k^{2}}{\lambda^{2}}|\Lambda|\Big).\label{eq:conc_cor2}
\end{equation}
If we furthermore assume \ref{as:exiuni}, then 
\begin{equation}
\mathbb{P}\big(\|\nabla\varphi^{\eta}\|_{\Lambda}^{2}\ge k|\Lambda|\big)\le\exp\Big(-c\frac{k^{2}}{\lambda^{2}}|\Lambda|\Big).\label{eq:conc_cor3}
\end{equation}
\end{cor}

For the proof of Lemma \ref{lem:conc_lem} and Lemma~\ref{lem:ass_conc} below, we will need the following concentration inequality for the maximum of a Gaussian process, a consequence of the fundamental results of Borell and Tsirelson--Ibragimov--Sudakov (see, e.g., \cite[Theorem~2]{zeitouni2014gaussian}).
\begin{theorem}\label{thm: BTIS}
  Let $\{X_t\}_{t\in T}$ be a Gaussian process (not necessarily centered) on a compact space $T$ with continuous covariance function and expectation.   
 Assume $S_T\coloneqq\sup_{t\in T}X_t$ is finite almost surely, then $\E S_T $ is finite and 
 \[\P(|S_T-\E S_T|\ge u )\le 2e^{-\frac {u^2}{2\sigma_T^2}}\]
 where $\sigma_T^2\coloneqq\sup_{t\in T} \var(X_t)$.
 \end{theorem}

We turn to prove Lemma~\ref{lem:conc_lem}.
\begin{proof}[Proof of Lemma~\ref{lem:conc_lem}]
Fix $d$, $n$ and $\mathcal{L}$. Let $C>0$ be constant to be chosen
later. Let $k\ge C\max\{1,\lambda^{2}\}$. In order to prove \eqref{eq:926}
we discretize $\Omega_{k}$ by (recall that $\tau=0$): 
\begin{equation}
\widetilde{\Omega}_{k}:=\Omega_{k}\cap\{\varphi:\Z^{d}\to\Z^{n}\}.
\label{eq:defomegatilde}
\end{equation}
We also discretize the disorder $\eta$ in the following way. For any
$v\in\Lambda$ and $z\in\mathbb{Z}^{n}$
define 
\[
\eta'_{v,z}:=\inf\big\{\eta_{v,s}:s\in z+[0,1]^{n}\big\}.
\]
We claim that there exists $c>0$ such that for all $v\in\Lambda$, $z\in\mathbb{Z}^{n}$
and $x>0$ we have that 
\begin{equation}
\mathbb{P}\big(|\eta'_{v,z}|\ge x\big)\le2e^{-cx^{2}}.\label{eq:957}
\end{equation}
Indeed, we would like to apply Theorem~\ref{thm: BTIS} with $T:=z+[0,1]^n$ and with the Gaussian process $X(s):=-\eta _{v,s}$. This process is almost surely continuous and therefore has continuous covariance function and expectation. Thus, by Theorem~\ref{thm: BTIS} we have that $\mathbb E [\eta '_{v,t}]<C$ for some constant $C$ depending on the bump function $b$. In fact, one may check using \cite[Theorem~23.7]{Kallenberg} and Theorem~\ref{thm: BTIS} that this constant depends only on the Lipschitz constant $\mathcal L$ of $b$. Thus, \eqref{eq:957} follows from Theorem~\ref{thm: BTIS} with a constant $c$ depending only on $n$ and $\mathcal L$.

It is convenient to extend $\eta '_{v,\cdot }$ from a function on $\mathbb Z ^n$ to a function on $\mathbb R^n$. For any $z\in \mathbb Z ^n$ and $t\in z+[0,1)^n$ we let $\eta '_{v,t}:=\eta '_{v,z}$.

Next, we use some basic facts about sub-Gaussian random variables given in \cite{vershynin2020high}. Recall that the sub-Gaussian norm of a random
variable $Y$ is given by 
\begin{equation}
\|Y\|_{\psi_{2}}\coloneqq\inf\big\{ x>0:\E[\exp(Y^{2}/x^{2})]\le2\big\}.
\end{equation}
(see, e.g., \cite[Definition 2.5.6]{vershynin2020high}). By \cite[Exercise 2.5.7]{vershynin2020high}, equation \eqref{eq:957} implies that $\|\eta'_{v,t}\|_{\psi_{2}}\le C$. Note that for $v\ne w\in\Lambda$ and $s,t\in\R^{n}$, the random
variables $\eta'_{v,s}$ and $\eta'_{w,t}$ are independent and therefore by \cite[Proposition 2.6.1]{vershynin2020high} it follows that
for every $\varphi:\Lambda\to\mathbb{R}^{n}$  we have $\|\sum_{v\in\Lambda}\eta'_{v,\varphi_{v}}\|_{\psi_{2}}^{2}\le C|\Lambda|$.
Hence, by \cite[Equation (2.14)]{vershynin2020high} for every such
$\varphi$ 
\begin{equation}
\P\Big(\lambda\sum_{v\in\Lambda}\eta'_{v,\varphi_{v}}\le-k|\Lambda|/4\Big)\le\P\Big(|\sum_{v\in\Lambda}\eta'_{v,\varphi_{v}}|\ge k|\Lambda|/(4\lambda)\Big) \le2\exp\Big(-c\frac{k^{2}}{\lambda^{2}}|\Lambda|\Big).
\label{eq:6385}
\end{equation}
We next union bound the last bound over $\varphi\in\Tilde{\Omega}_{2k}$.
Enumerating over the differences  $a_{u,v}=\varphi_{u}-\varphi_{v}\in \mathbb Z ^n$
for $u\sim v$ we obtain, for $k>C$ 
\begin{equation}
\begin{split}\big|\tilde{\Omega}_{2k}\big| & \le\Big|\Big\{\big(a_{u,v}^{i}\in\mathbb{Z}\ :\ u\sim v,\{u,v\}\cap\Lambda\neq\emptyset,i\in\{1,\dots,n\}\big):\sum|a_{u,v}^{i}|\le2k|\Lambda|\Big\}\Big|\\
 & \le2^{N_{\Lambda}}{2k|\Lambda|+nN_{\Lambda}-1 \choose nN_{\Lambda}-1}\le e^{4k|\Lambda|},
\end{split}
\label{eq:contomegak}
\end{equation}
where $N_{\Lambda}:=\big|\{u\sim v,\{u,v\}\cap\Lambda\neq\emptyset\}\big|$.
For sufficiently large $C$, if $k>C\lambda^{2}$ then $c\frac{k^{2}}{\lambda^{2}}>2\cdot4k$
and thus \eqref{eq:6385} and \eqref{eq:contomegak} imply 
\begin{equation}
\mathbb{P}\Big(\exists\varphi\in\tilde{\Omega}_{2k}, \ \lambda\sum_{v\in\Lambda}\eta'_{v,\varphi_{v}}\le-k|\Lambda|/4\Big)\le2\exp\Big(-c\frac{k^{2}}{2\lambda^{2}}|\Lambda|\Big).\label{eq:927}
\end{equation}
We claim that \eqref{eq:926} follows from \eqref{eq:927}. Indeed,
for any $\varphi\in\Omega_{k}$, we consider the function $\tilde{\varphi}$
defined by $\tilde{\varphi}_{v}:=\lfloor\varphi_{v}\rfloor$ (in here
the floor sign of a vector $t\in\mathbb{R}^{n}$ is taken in each
coordinate). We obtain that for all $\varphi\in\Omega_{k}$, since
$k>C$ 
\begin{equation}
\|\nabla\tilde{\varphi}\|_{\Lambda}^{2}\le(\|\nabla\varphi\|_{\Lambda}+\|\nabla(\tilde{\varphi}-\varphi)\|_{\Lambda})^{2}\le\big(\sqrt{k|\Lambda|}+\sqrt{2dn|\Lambda|}\big)^{2}\le2k|\Lambda|\label{eq:norml2discretephi}
\end{equation}
and therefore $\tilde{\varphi}\in\tilde{\Omega}_{2k}$. Finally, by
the definition of $\eta'$, if for some $\varphi\in\Omega_{k}$ we
have that $\lambda\sum\eta_{v,\varphi_{v}}\le-k|\Lambda|/4$ then
$\lambda\sum\eta'_{v,\tilde{\varphi}_{v}}\le-k|\Lambda|/4$. This
finishes the proof of \eqref{eq:926} using \eqref{eq:927}. 
\end{proof}

\begin{proof}[Proof of Corollary~\ref{lem:nabla}]
The first statement of the corollary follows immediately from \eqref{eq:926}.
Indeed, 
\begin{equation}
\mathbb{P}\big(\GE^{\eta}<-K|\Lambda|\big)\le\sum_{k=K}^{\infty}\mathbb{P}\Big(\exists\varphi\in\Omega_{k},\ \lambda\sum_{v\in\Lambda}\eta_{v,\varphi_{v}}\le-k|\Lambda|/4\Big)\le\exp\Big(-c\frac{K^{2}}{\lambda^{2}}|\Lambda|\Big).
\end{equation}
Note that we used $K\ge C\lambda^{2}$ in bounding the sum above.
We now turn to prove the second statement of the corollary. We have 
\begin{equation}
\begin{split}
\mathbb{P}\Big(\exists\varphi\in\Omega_{K}^{c},H^{\eta}(\varphi)\le K|\Lambda|/2\Big) & \le\sum_{k=K}^{\infty}\mathbb{P}\Big(\exists\varphi\in\Omega_{k+1}\!\!\setminus \! \Omega_{k},\ k|\Lambda|+\lambda\sum_{v\in\Lambda}\eta_{v,\varphi_{v}}\le\frac{K}{2}|\Lambda|\Big)\\
 & \le\sum_{k=K}^{\infty}\mathbb{P}\Big(\exists\varphi\in\Omega_{k+1},\lambda\sum_{v\in\Lambda}\eta_{v,\varphi_{v}}\le-\frac{k+1}{4}|\Lambda|\Big)\\
 & \le\sum_{k=K}^{\infty}\exp\Big(-c\frac{(k+1)^{2}}{\lambda^{2}}|\Lambda|\Big)\le\exp\Big(-c\frac{K^{2}}{\lambda^{2}}|\Lambda|\Big)
\end{split}
\label{eq:control energy phi}
\end{equation}
where we used inequality \eqref{eq:926} in the second to last inequality.

Let us furthermore assume that \ref{as:exiuni} holds. First note
that by a simple bound on the normal distribution, for all $k\ge0$
\[
\P\Big(\lambda\sum_{v\in\Lambda}\eta_{v,0}\ge k|\Lambda|/2\Big)\le\exp\Big(-c\frac{k^{2}}{\lambda^{2}}|\Lambda|\Big).
\]
We have 
\begin{equation}
\begin{split}\mathbb{P} & \big(\|\nabla\varphi^{\eta}\|_{\Lambda}^{2}\ge K|\Lambda|\big)\le\sum_{k=K}^{\infty}\mathbb{P}\big(\exists\varphi\in\Omega_{k+1}\setminus\Omega_{k},\ H^{\eta}(\varphi)\le H^{\eta}(0)\big)\\
 & \le\sum_{k=K}^{\infty}\mathbb{P}\big(H^{\eta}(0)\ge k|\Lambda|/2\big)+\mathbb{P}\Big(\exists\varphi\in\Omega_{k}^{c},H^{\eta}(\varphi)\le k|\Lambda|/2\Big)\le\exp\Big(-c\frac{K^{2}}{\lambda^{2}}|\Lambda|\Big)
\end{split}
\end{equation}
where we used \eqref{eq:control energy phi} in the last inequality.
This concludes the proof. 
\end{proof}
\begin{lemma} \label{lem:ass_basic} The disorder $\eta^{{\rm white}}$
satisfies the Assumption~\ref{as:exiuni}. Moreover, there is a unique
$\varphi$ achieving the ground energy, i.e., a unique $\varphi$
for which $\GE^{\eta,\lambda,\Lambda}:=H^{\eta,\lambda,\Lambda}(\varphi)$.
\end{lemma} The moreover part of the lemma is not part of our general
assumptions and is not used in our proofs, but is noted here for the
reader's convenience.
\begin{proof}[Proof of Lemma \ref{lem:ass_basic}]
\textbf{Existence:} Throughout the proof we fix $\lambda$ and $\Lambda$.
We start by proving the existence of $\varphi^{\eta}$. The set $\Omega_{k}$
is bounded and we moreover have almost surely thanks to (\ref{eq:conc_cor2})
in Corollary~\ref{lem:nabla} that 
\[
\lim_{k\rightarrow\infty}\inf_{\varphi\notin\Omega_{k}}H^{\eta}(\varphi)=+\infty.
\]
Hence, almost surely for every closed set of $\Pi\subset\Omega^{\Lambda}$
there exists a large enough $k$ such that the infimum in $\Pi$ is
equal to the infimum on the compact set $\Pi\cap\Omega_{k}$. The
existence of a minimum on $\Pi$ follows from the continuity of the
map $\varphi\mapsto H^{\eta}(\varphi)$ where we used the standard
topology on $(\R^{n})^{\Lambda}$ and that there exists a version
of $\eta$ such that for all $v\in\Lambda$, the function $t\mapsto\eta_{v,t}$
is continuous (see Section \ref{sec:whitenoise}).

\textbf{Uniqueness:} For the uniqueness of $\varphi^{\eta,\lambda,\Lambda}$
we define the following sets for all $v\in\Lambda$, $i\le n$ and
$q\in\mathbb{R}$ 
\begin{equation}
\Omega_{-}(v,i,q):=\Omega^{\Lambda}\cap\big\{\varphi:\mathbb{Z}^{d}\to\mathbb{R}^{n}:\varphi(v)_{i}\le q\big\},
\end{equation}
\begin{equation}
\Omega_{+}(v,i,q):=\Omega^{\Lambda}\cap\big\{\varphi:\mathbb{Z}^{d}\to\mathbb{R}^{n}:\varphi(v)_{i}\ge q\big\}.
\end{equation}
Next, define $\varphi_{-}(v,i,q)$ and $\varphi_{+}(v,i,q)$ to be
the the functions $\varphi$ that minimize the Hamiltonian and restricted
to be in $\Omega_{+}(v,i,q)$ and $\Omega_{-}(v,i,q)$ respectively.
These functions exist by the same arguments as above. We also let
$\GE_{-}(v,i,q):=H(\varphi_{-}(v,i,q))$ and $\GE_{+}(v,i,q):=H(\varphi_{+}(v,i,q))$
. We claim that for all $q_{1}<q_{2}$ we have 
\begin{equation}
\mathbb{P}\big(\GE_{-}(v,i,q_{1})=\GE_{+}(v,i,q_{2})\big)=0.\label{eq:1}
\end{equation}
Indeed, note that the random variable $\GE_{-}(v,i,q_{1})$ is measurable
in 
\begin{equation}
\mathcal{F}(v,i,q_{1}):=\sigma\big(\big\{\eta_{u,t}:u\neq v,t\in\mathbb{R}^{n}\big\}\cup\big\{\eta_{v,t}:t_{i}\le q_{1}\big\}\big).
\end{equation}
Moreover, conditioning on $\mathcal{F}(v,i,q_{1})$, it is not hard to check
that $\GE_{+}(v,i,q_{1})$ has a continuous distribution. This finishes the proof of \eqref{eq:1}.

Next, define the event 
\begin{equation}
\mathcal{A}:=\bigcup_{v\in\Lambda}\bigcup_{i\le n}\bigcup_{\substack{q_{1},q_{2}\in\mathbb{Q}\\
q_{1}<q_{2}
}
}\big\{\GE_{-}(v,i,q_{1})=\GE_{+}(v,i,q_{2})\big\}.
\end{equation}
By \eqref{eq:1} and a union bound we have that $\mathbb{P}(\mathcal{A})=0$.
This finishes the proof of the lemma since on $\mathcal{A}^{c}$ we
cannot have two distinct minimizers to the Hamiltonian. 
\end{proof}
\begin{lemma}\label{lem:ass_conc} The disorder $\eta^{white}$ satisfies
Assumption~\ref{as:conc}. \end{lemma}

We will need the following standard claim. We omit the proof of the
claim. 
\begin{claim}
\label{claim:something} For every $c_{1}$ there is $c_{2}$ such
that the following holds. Let $X$ be a random variable and suppose that for all $\rho>0$ there exists $a_{\rho}\in\mathbb{R}$ such that
$\mathbb{P}(|X-a_{\rho }|\ge\rho)\le2e^{-c_{1}\rho^{2}}$. Then, $\mathbb{E}|X|<\infty$
and $\mathbb{P}(|X-\mathbb{E}X|\ge\rho)\le2e^{-c_{2}\rho^{2}}$ for
all $\rho>0$. 
\end{claim}

\begin{proof}[Proof of Lemma \ref{lem:ass_conc}] Let $\Delta\subset \Lambda$. By Corollary \ref{cor:conc}, we have that $|\GE^\eta|<\infty$ almost surely.
Let $\rho>0$. By Corollary~\ref{cor:conc} and Markov's inequality, $\mathbb P \big( \mathbb P (\varphi^\eta \notin\Omega_k \mid \eta_{\Delta^c} ) \ge e^{-\rho ^2} \big) \le e^{\rho ^2}\mathbb P (\varphi^\eta \notin\Omega_k )\to 0$ as $k\to \infty$. Thus, conditioning on $\eta_{\Delta^c}$ there is $k$ sufficiently large (depending on $\eta_{\Delta^c}$ and $\rho $) such that almost surely
\begin{equation}\label{eq:for conc}
    \P(\varphi^\eta \notin\Omega_k \big |\eta_{\Delta^c}) \le e^{-\rho ^2}.
\end{equation}
From now on we condition on $\eta_{\Delta^c}$ and work in the conditional space. We would like to apply Theorem~\ref{thm: BTIS} with $T:=\Omega_k$ and the Gaussian process $X_\varphi:=-H^\eta(\varphi) $ for $\varphi\in T$. The set $T$ is compact, $S_T=\sup_{\varphi\in T} X_\varphi$ is finite almost surely by continuity of $X$ and
\[\var(X_\varphi|\eta_{\Delta^c})=\sum_{v\in\Delta}\var(\eta_{v,\varphi_v})\le C\lambda ^2 |\Delta|.\]
Since $(X_t)$ is continuous almost surely it has continuous covariance and expectation. We obtain
\begin{equation*}
    \begin{split}
        \P\big(& \big|\GE^{\eta} - \E [ S_T \mid \eta_{\Delta^c}] \big|\ge \rho\lambda\sqrt{|\Delta|}\ \big|  \eta_{\Delta^c}\big) \\
        &\le \P(\varphi^\eta \notin\Omega_k \ \big | \ \eta_{\Delta^c})+\P \big( |S_T-\E [ S_T \mid \eta_{\Delta^c}] |\ge \rho \lambda \sqrt{|\Delta|} \ \big| \  \eta_{\Delta^c}\big)\le 3e^{-c\rho ^2},
    \end{split}
\end{equation*}
where the last inequality follows from \eqref{eq:for conc} and Theorem~\ref{thm: BTIS}, using that $\sigma_T^2\le  C\lambda ^2 |\Delta|$. This finishes the proof of the lemma using Claim~\ref{claim:something}.
\end{proof}

\section{Discussion on disorder universality classes and examples}\label{sec:disorder types}

In this section we informally discuss several \emph{universality classes of disorders}. By such a universality class we mean general features of the distribution of the random environment $\eta:\Z^d\times\R^n\to(-\infty,\infty]$ which may be expected to govern the behavior of the surface minimizing~\eqref{eq:formal Hamiltonian}. We motivate these classes by discussing relations with disordered spin systems, provide examples of specific disorder distributions in each class and describe the predictions for the behavior of the minimal surface in such environments. More information and references can be found in the recent review~\cite{wiese2022theory}.

In all of the disorder classes discussed we assume that
\begin{equation}\label{eq:IID assump}
\text{the $(\eta_{v,\cdot})_{v\in\Z^d}$ are independent and identically distributed processes,} 
\end{equation}
which are almost surely non-constant.
We restrict attention to the case $\lambda=1$ in~\eqref{eq:formal Hamiltonian}, fix $d,n\ge 1$ and consider harmonic minimal surfaces (the minimizing surfaces of~\eqref{eq:finite volume Hamiltonian}) on $\Lambda=\Lambda_L$ for $L$ large with zero boundary conditions. Also, without trying to make such an assumption formal, we are considering situations in which the disorder distribution has sufficiently light tails so that these do not affect the distribution of the minimal surface (see~\cite{biroli2007extreme,hambly2007heavy,auffinger2011directed,gueudre2014revisiting,dey2016high,berger2019entropy,berger2019directed, berger2022non} for results on heavy-tailed cases with $d=1$).

\subsection{Independent disorder}\label{sec:independent disorder} This type of disorder is characterized by, for each $v\in\Z^d$,
\begin{enumerate}
    \item Stationarity: the process $t\mapsto\eta_{v,t}$ is stationary to translations in $\R^n$.
    \item Independence: The restrictions $\eta |_{\{v\}\times A}$ and $\eta |_{\{v\}\times B}$ are independent for each $A,B\subset\R^n$ satisfying $\inf_{t\in A, s\in B}\|t-s\|\ge 2$.
\end{enumerate}
Together with~\eqref{eq:IID assump}, we see that such disorders satisfy our assumptions~\ref{as:stat} and~\ref{as:indep}. This is the main disorder type considered in this paper and the transversal fluctuations predicted for it were discussed in Table~\ref{table:transversal fluctuations n=1} and Section~\ref{sec:physics results background}. We note that instead of the strict independence assumption above, it should also be possible to assume that the processes $t\mapsto\eta_{v,t}$ have sufficiently fast decay of correlations (stationary disorders with slow, power-law, decay of correlations are considered in~\cite{mezard1990interfaces, mezard1991replica, mezard1992manifolds, balents1993large}).

\subsubsection{Examples} Our main example in this class is $\eta^\white$, the regularized white noise process introduced in Section~\ref{sec:whitenoise}. A second natural example is based on Poisson processes: For each $v\in\Z^d$ independently, let $P_v$ be a Poisson point process of unit intensity on $\R^n$. Define $\eta^{\text{Poisson}}:\Z^d\times\R^n\to\{0,\infty\}$ by
\begin{equation}\label{eq:eta Poisson}
    \eta^{\text{Poisson}}_{v,t}:=\begin{cases}0&t\in P_v,\\
    \infty&t\notin P_v.
    \end{cases}
\end{equation}
With this disorder, the minimal surface must pass through the Poisson points (i.e., each minimal surface $\varphi$ satisfies $\varphi_v\in P_v$ for all $v$). The example may be enhanced by sampling at each Poisson point an independent weight from some weight distribution and setting the disorder equal to that weight there (instead of setting it equal to $0$).

A third natural example is obtained by a uniformly-shifted interpolation of a discrete-space process: Let $(W_{v,k})_{v\in\Z^d, k\in\Z^n}$ be independent samples from some probability measure $\mu$. Define $W_{v,t}$ for all $t\in\R^n$ by an interpolation rule (e.g., $W_{v,t}=W_{v,k}$ for the $k$ nearest to $t$, with some tie-breaking rule). Let $(U_v)_{v\in\Z^d}$ be independent and uniform on $[0,1]^n$. Then, the uniformly-shifted interpolated disorder $\eta$ is defined by setting $\eta_{v,t}:=W_{v,t+U_v}$. Note that $\eta$ satisfies the stationarity property above. While it doesn't satisfy the independence property, we still expect it to induce a similar behavior in the surfaces minimizing~\eqref{eq:formal Hamiltonian} as the other disorders of this class.

Lastly, for $n=1$, one may let $t\mapsto\eta_{v,t}$ be the Ornstein-Uhlenbeck process, independently between $v\in\Z^d$. This choice satisfies the stationarity property. It does not satisfy the independence property, but has fast decay of correlations and may still be expected to induce a similar behavior in the surfaces minimizing~\eqref{eq:formal Hamiltonian} as the other disorders of this class. 

\subsubsection{Disordered Ising ferromagnet and predicted roughening transition}\label{sec:disordered Ising ferromagnet} For $n=1$, minimal surfaces under independent disorder are used to model the domain walls of the \emph{disordered Ising ferromagnet} (or random-bond Ising model), as we now discuss (c.f.~\cite[Section 9.1]{bassan2023non}). The latter is the spin system on $\sigma:\Z^D\to\{-1,1\}$, with $D:=d+1$, with formal Hamiltonian
\begin{equation}\label{eq:Disordered ferromagnet formal Hamiltonian}
H^{J,\text{RB-Ising}}(\sigma):=-\sum_{\{x,y\}\in E(\Z^D)}J_{\{x,y\}}\sigma_x\sigma_y,
\end{equation}
in which the coupling constants $J=(J_e)_{e\in E(\Z^D)}$ are a \emph{non-negative} quenched random field, sampled independently from a non-constant distribution with, say, compact support in $(0,\infty)$. Consider this system in the infinite cylinder $\Delta_L:=\{-L,\ldots, L\}^d\times\Z$ with Dobrushin boundary conditions $\rho^{\Dob}\colon \Z^D\to\{-1,1\}$ given by $\rho^{\Dob}_{(v,k)} := \sign(k - 1/2)$ where $\sign$ denotes the sign function (the distribution of $\sigma$ is well-defined in $\Delta_L$ as the cylinder is only infinite in one direction). At zero temperature, the domain wall separating the spin states $-1$ and $1$ is a membrane anchored at height $1/2$ on the boundary of $\Delta_L$. In dimension $d=1$, it coincides with a first-passage percolation geodesic (on the dual graph to $\{-L,\ldots,L\}\times\Z$). Harmonic minimal surfaces (the model~\eqref{eq:formal Hamiltonian}) in an independent disorder are often used as an approximation to this membrane; c.f.~\cite{huse1985pinning, wiese2022theory}. The approximation is obtained by requiring the membrane to have no overhangs (i.e., to be expressible as a function of $\Z^d$), replacing the resulting (random) Solid-On-Solid interaction by a (deterministic) quadratic interaction and relaxing the surface heights to be real rather than integer (see also~\cite[Remark 1.12]{bassan2023non} which explains that the no overhangs property may be obtained by setting all coupling constants $J_e$ of edges in the first $d$ coordinate directions to be equal). 

The domain wall membrane is believed (but not proved) to behave similarly to the harmonic minimal surface in dimensions $d=1,2$ and $d\ge 5$ in the sense of delocalizing with the same transversal fluctuation exponent $\xi$ (Table~\ref{table:transversal fluctuations n=1}) in dimensions $d=1,2$ (a delocalization result for an approximating disordered Solid-On-Solid model is proved by Bovier--K\"ulske~\cite{bovier1996there}) and localizing in dimensions $d\ge 5$. However, the fact that the domain wall membrane is pinned to (half-)\emph{integer heights} leads to a new behavior in dimensions $d=3,4$: Indeed, it is proved in~\cite{bassan2023non} (and also~\cite{bovier1994rigorous} for the disordered Solid-On-Solid approximation) that when the coupling constant distribution is very concentrated (e.g., uniform on $[a,b]$ with $b>a>0$ satisfying that $(b-a)/a$ is small) then the domain wall membrane localizes, unlike the harmonic minimal surface which we prove to delocalize (Theorem~\ref{thm:delocalization intro}). Moreover, in dimension $d=4$ it is predicted that the domain wall membrane always localizes (e.g., for all $b>a>0$ in the above example) while in dimension $d=3$ it is predicted that a \emph{roughening transition} takes place as the distribution of the $(J_e)$ becomes sufficiently spread (e.g., as $(b-a)/a$ grows) ~\cite{emig1998roughening, emig1999disorder}. For $d=3$, at the critical roughening point the height fluctuations of the domain wall membrane are predicted to have order $\sqrt{\log L}$ while above the roughening transition they are predicted to be a power-law with the same transversal fluctuation exponent as the harmonic minimal surface~\cite{emig1998roughening, emig1999disorder}.

Minimal surfaces with codimension $n>1$ under independent disorder also have natural counterparts in disordered spin systems: see discussion in~\cite[Section 9.5]{bassan2023non} for a way to realize such a discrete minimal surface as the defect set of a disordered (generalized) Ising lattice gauge theory (which again coincides with a first-passage percolation geodesic when $d=1$).

\subsection{Brownian and fractional Brownian disorder}\label{sec:Brownian disorder} We first discuss Brownian disorder, by which we mean that $n=1$ and the disorder consists of two-sided Brownian motions. Precisely, $\eta_{v,t} = \text{BM}_v(t)$ with each $\text{BM}_v$ an independent Brownian motion on $\R$ normalized to satisfy $\text{BM}_v(0)=0$. While this disorder is not stationary, it has stationary (independent) increments.

\subsubsection{Predictions} It is predicted that the height fluctuations of harmonic minimal surfaces with Brownian disorder are of order $L^{\frac{4-d}{3}}$ in dimensions $d=1,2,3$~\cite{villain1982commensurate,grinstein1982roughening, grinstein1983surface, nattermann1988random, forgacs1991behavior,wiese2022theory}, of order $(\log L)^{1/3}$ in dimension $d=4$ (\cite{emig1998roughening} and ~\cite[following (78)]{emig1999disorder}), whereas the minimal surface is localized in dimensions $d\ge 5$. In a sequel paper we provide rigorous proofs of these predictions for $d\neq 4$, with less precise results for $d=4$. 

\subsubsection{Random-field Ising model}\label{sec:random-field Ising model}
Minimal surfaces under Brownian disorder model the domain walls of the random-field Ising model. The latter is the spin system on $\sigma:\Z^D\to\{-1,1\}$ whose formal Hamiltonian is given by
\begin{equation}\label{eq:random-field Ising formal Hamiltonian}
H^{h,\text{RF-Ising}}(\sigma):=-\sum_{\{x,y\}\in E(\Z^D)}\sigma_x\sigma_y-\lambda\sum_x h_x \sigma_x,
\end{equation}
in which the random field $h=(h_x)_{x\in \Z^D}$ is independent, with each $h_x$ having (say) the standard Gaussian distribution, and with $\lambda>0$ the disorder strength. Similarly to the case of the disordered Ising ferromagnet discussed in Section~\ref{sec:disordered Ising ferromagnet}, one places the random-field Ising model in Dobrushin boundary conditions in the infinite cylinder $\Delta_L:=\{-L,\ldots, L\}^d\times\Z$, where $d = D-1$. Then, at zero (or low) temperature and weak disorder (small $\lambda>0$) one studies the domain wall membrane anchored at height $1/2$ on the boundary of $\Delta_L$ (the membrane may have some ``width" but should be approximately well defined at weak disorder). Harmonic minimal surfaces in a Brownian disorder are often used as an approximation to this domain wall membrane~\cite{wiese2022theory}. The situation is then similar to that described in Section~\ref{sec:disordered Ising ferromagnet}: the domain wall membrane is believed to delocalize with the same transversal fluctuation exponent $\xi$ as its minimal surface approximation in dimensions $d=1,2$, and to localize in dimensions $d\ge 5$. For $d=4$ the domain wall membrane is believed to localize (unlike its minimal surface approximation). For $d=3$, a transition from a localized to a delocalized regime may occur as $\lambda$ increases~\cite{emig1998roughening,emig1999disorder}, with $\sqrt{\log L}$ fluctuations at the transition point and power-law fluctuations for larger $\lambda$ (though not too large, so the domain wall membrane can still be defined), with the same transversal fluctuation exponent as the minimal surface approximation. 

\subsubsection{Fractional Brownian disorder}\label{sec:fractional Brownian disorder} The physics literature  studies how the transversal fluctuations of the minimal surface depend on the long-range correlations in the disorder processes $t\mapsto\eta_{v,t}$~\cite{mezard1990interfaces, mezard1991replica, mezard1992manifolds, balents1993large}. Extending beyond the Brownian disorder discussed above, we propose to consider fractional Brownian disorder of Hurst parameter $0<H<1$. By this we mean that for each $v$, the process  $t\mapsto\eta_{v,t}$ is distributed as a centered Gaussian process on $\R^n$, normalized to satisfy $\eta_{v,0}=0$, and with covariance described by
\begin{equation}
    \var(\eta_{v,t} - \eta_{v,s}) = \|t-s\|^{2H},\qquad t,s\in\R^n.
\end{equation}
See, e.g.,~\cite{ossiander1989certain} for a short proof that such a process exists.
The fractional Brownian disorder has continuous sample paths and stationary increments (but is not stationary). The Brownian disorder case discussed above is recovered by taking $n=1$ and $H=1/2$.

While fractional Brownian disorder is not explicitly considered in~\cite{mezard1990interfaces, mezard1991replica, mezard1992manifolds, balents1993large}, extrapolating from their results we may predict that, for all values of $n$, harmonic minimal surfaces with this disorder have transversal fluctuations of order $L^{\frac{4-d}{4-2H}}$ in dimensions $d=1,2,3$ and are localized in dimensions $d\ge5$. We may further predict that the transversal fluctuations are of order $(\log L)^{\frac{1}{4-2H}}$ in dimension $d=4$, where the power of the logarithm equals the coefficient of $4-d$ in the transversal fluctuation exponent predicted for dimensions $d<4$ (as done, e.g., in~\cite{emig1998roughening,emig1999disorder} for the Brownian case). In our sequel paper mentioned above we provide rigorous proofs for these transversal fluctuation predictions for $d\neq 4$, with less precise results for $d=4$.

\subsection{Linear disorder}\label{sec:linear disorder} By a linear disorder we mean that $\eta_{v,t} = \zeta_v\cdot t$ with the $\zeta=(\zeta_v)_{v\in\Z^d}$ independent, each having (say) the standard Gaussian distribution in $\R^n$. This disorder is the limiting $H=1$ case of the fractional Brownian disorder introduced in Section~\ref{sec:fractional Brownian disorder}. As $\eta_{v,t} = \sum_{j=1}^n \zeta_{v,j} t_j$, the components of the harmonic minimal surface decouple under this disorder. It thus suffices, when studying the delocalization properties, to consider the case $n=1$. The harmonic minimal surface is then given by an \emph{exact formula}: the minimal surface $\varphi^\zeta$ on a finite $\Lambda\subset\Z^d$ with zero boundary values satisfies (with $\Delta_\Lambda$ defined in~\eqref{eq:Laplacian definition})
\begin{equation}\label{eq:real-valued linear disorder surface}
    \varphi^\zeta = -\lambda(-\Delta_\Lambda)^{-1}\zeta.
\end{equation}
On $\Lambda=\Lambda_L$, one then checks that $\varphi$ fluctuates to height $L^{\frac{4-d}{2}}$ in dimensions $d=1,2,3$, to height $\sqrt{\log L}$ in dimension $d=4$ and is localized in dimensions $d\ge 5$. These facts and their further extensions may be found in~\cite{dario2023random}, where it is also shown that harmonic minimal surfaces in linear disorder are equal in distribution to the, so called, membrane model~\cite[Section 7.2]{dario2023random}.

The roughening transition predicted for the disordered Ising ferromagnet (Section~\ref{sec:disordered Ising ferromagnet}) and the random-field Ising model (Section~\ref{sec:random-field Ising model}) is believed to also occur here: Consider minimizing the harmonic minimal surface Hamiltonian~\eqref{eq:formal Hamiltonian} with linear disorder over \emph{integer-valued} surfaces $\varphi:\Z^d\to\Z$. In this case, it is proven in~\cite{dario2023random} that the resulting integer-valued minimal surface still delocalizes to height $L^{\frac{4-d}{2}}$ in dimension $d=1,2$. However, in dimensions $d\ge 3$, it is shown to \emph{localize} for small $\lambda$ (thus displaying a different behavior from the real-valued minimal surface~\eqref{eq:real-valued linear disorder surface} in dimensions $d=3,4$). The behavior for large $\lambda$ is still open: It is conjectured in~\cite{dario2023random} that a roughening transition takes place in dimension $d=3$, whereas the surface always remains localized for $d\ge 5$. It is unclear whether a transition occurs for $d=4$ or whether the surface remains localized for all $\lambda$.

The harmonic minimal surface model with linear disorder is sometimes termed the \emph{random-rod model}~\cite{forgacs1991behavior} or the \emph{Larkin model} following Larkin~\cite{larkin1970effect} (c.f.~\cite{giamarchi2009disordered}). It appears naturally in the investigation of harmonic minimal surfaces with disorder which is smooth at short scales and independent and identically distributed between vertices (e.g., with $\eta^\white$ disorder) in the regime of small disorder strength $\lambda$: indeed, taking small $\lambda$ has the effect of rescaling the disorder (Section~\ref{sec:effect of lambda}), which then allows (in smooth cases) to approximate it by its first-order Taylor expansion, effectively yielding a linear disorder.

\subsection{Periodic disorder}\label{sec:periodic disorder} This type of disorder is characterized by, for each $v\in\Z^d$,
\begin{enumerate}
    \item Stationarity: the process $t\mapsto\eta_{v,t}$ is stationary to translations in $\R^n$.
    \item Periodicity: $\eta_{v,\cdot} = \eta_{v,\cdot+e_i}$ for $1\le i\le n$, where $e_i$ is the $i$th standard unit vector in~$\R^n$.
\end{enumerate}
As such processes satisfy~\ref{as:stat}, the versions of the relation $\chi\ge2\xi+d-2$ in Theorem~\ref{theorem:half scaling relation} and Theorem~\ref{thm: scaling relation d=1} apply to them (our proofs of $\chi\le2\xi+d-2$ do not apply as they use~\ref{as:indep}, and, indeed, it seems this other direction need not hold). Additionally, when our concentration assumption~\ref{as:conc} holds then our localization results in Theorem~\ref{thm:localization} apply (one may check that~\ref{as:conc} holds for the examples presented next).
\subsubsection{Examples} In one natural example, we start with an independent collection $(\text{PWN}_v)_{v\in\Z^d}$ of white noises on the $n$-dimensional torus of side length 1, extended periodically to the full $\R^n$. We then set the disorder to be their convolution with a bump function $b$, i.e., $\eta^{\text{per-white}}_{v,t}:=\text{PWN}_v(b(\cdot-t))$ similarly to the way that the disorder $\eta^\white$ is defined in~\eqref{eq:eta white definition} (making sure to choose $b$ so that $\eta^{\text{per-white}}_{v,\cdot}$ is non-constant). It is straightforward that $\eta^{\text{per-white}}$ is stationary and periodic.

To give another example, let $(U_v)_{v\in\Z^d}$ be independent and uniform on $[0,1]^n$. Let $\Z_v^n := \Z^n + U_v$ be a uniformly shifted version of the $\Z^n$ lattice, independently for each $v$. Define
\begin{equation}\label{eq:random-phase sine-Gordon disorder}
    \eta^{\text{RPSG}}_{v,t}:=\begin{cases}0&t\in\Z_v^n\\
    \infty&t\notin\Z_v^n\end{cases}.
\end{equation}
The harmonic minimal surface with this disorder is constrained to take values in $\Z_v^n$ for each~$v$. It is termed the (infinite activity) \emph{random-phase sine-Gordon} (RPSG) surface. To understand its delocalization properties it suffices to consider the $n=1$ case, as its components are independent, each with the distribution of the $n=1$ RPSG surface. Indeed, this follows from the fact that $t\mapsto \eta^{\text{RPSG}}_{v,t}$ has the distribution of $t\mapsto\sum_{j=1}^n \eta^{\text{RPSG}, j}_{v,t_j}$ where the $(\eta^{\text{RPSG}, j}_{v,\cdot})_{j=1}^n$ are independent and distributed as~\eqref{eq:random-phase sine-Gordon disorder} with $n=1$. Thus, in particular, the exponents $\chi,\xi$ for the RPSG surface are the same for all $n$.

\subsubsection{Predictions}\label{sec:predictions on periodic disorder} The RPSG surface with $n=1$ has received much attention in the physics literature (see, e.g.,~\cite{CO82, TD90, HWA94, LDS07, RLDS12}) and the following behavior is predicted (on $\Lambda_L$ with zero boundary conditions): In two dimensions, the heights are predicted to delocalize to height of order $\sqrt{\log L}$ at \emph{high temperature} (rough phase), and to order $\log L$ at low and zero temperature (super-rough phase); the height fluctuations are thus expected to \emph{decrease} as the temperature rises! It is further predicted, at low and zero temperature, that the heights delocalize to order $\sqrt{\log L}$ in three dimensions and are localized when $d\ge 5$ (see e.g.~\cite{GLD95, N90SC, OS95, VF84}). In dimension $d=4$, the prediction may be that the heights delocalize to order $\sqrt{\log \log L}$~\cite{chitra1999disordered}.

These predictions appear to be open in the mathematical literature, apart from a result of Garban--Sep\'ulveda~\cite{GS20} which proves a lower bound of $\sqrt{\log L}$ on the order of the heights at high temperature in two dimensions, even when the shifted lattices $\Z_v^{n=1}$ are chosen as arbitrary (deterministic) shifts of $\Z$ (R. Bauerschmidt pointed out to us that this also follows from the result of Fr\"ohlich--Spencer~\cite{frohlich1981kosterlitz} together with the correlation inequality of~\cite[Lemma C.3.]{aizenman2021depinning}. The correlation inequality shows that the covariance matrix of the integer-valued Gaussian free field with a ``magnetic field'' is smallest, in the order of symmetric matrices, when the magnetic field vanishes).

Harmonic minimal surfaces with both disorder examples above, $\eta^{\text{per-white}}$ and $\eta^{\text{RPSG}}$, satisfy our concentration assumption~\ref{as:conc} (with $\lambda=1$) for arbitrary $n\ge 1$. For $\eta^{\text{per-white}}$ this may be checked similarly to $\eta^\white$ (Section~\ref{sec:assumptions for eta white}) while for $\eta^{\text{RPSG}}$ it follows from Hoeffding's inequality by noting that changing a single $\Z^n_v$ has a bounded effect on the ground energy. Thus, localization in dimensions $d\ge 5$ as well as bounds for the height fluctuations in lower dimensions follow for these surfaces from our Theorem~\ref{thm:localization}.

\subsubsection{Random-field XY model} For $n=1$, minimal surfaces under periodic disorder model the angle variables of the random-field XY model in a ``no-vortices" approximation, as we now discuss. The random-field XY model is the spin system on $\sigma:\Z^d\to\mathbb{S}^1$, with $\mathbb{S}^1$ the unit circle in $\R^2$, with formal Hamiltonian (with $\cdot$ defined in~\eqref{eq:inner product R n})
\begin{equation}\label{eq:Disordered XY formal Hamiltonian}
H^{h,\text{RF-XY}}(\sigma):=-\sum_{\{u,v\}\in E(\Z^d)}\sigma_u\cdot\sigma_v - \lambda\sum_{v\in\Z^d}h_v\cdot\sigma_v,
\end{equation}
in which the random field $h=(h_v)_{v\in \Z^d}$ is independent with, e.g., each $h_v$ having the standard Gaussian distribution in $\R^2$ and with disorder strength $\lambda>0$.

We wish to identify $\sigma$ with the complex exponential of an angle field, $\sigma_v:=e^{2\pi i\varphi_v}$. As this only defines each $\varphi_v$ up to an additive integer, we further require the minimality of $|\varphi_u - \varphi_v|$ for each edge $\{u,v\}\in E(\Z^d)$. This requirement may only hold if $\sigma$ has \emph{no vortices}, i.e., no $2\times 2$ cycles in $\Z^d$ on which the angle of $\sigma$ completes a full turn when rotating in the direction of the smaller angle difference on each edge. For our heuristics, we assume this ``no-vortices" approximation. Translating~\eqref{eq:Disordered XY formal Hamiltonian} to the angle variables leads to a Hamiltonian for $\varphi$ similar to~\eqref{eq:formal Hamiltonian} (even more similar in the case of the random-field Villain model), with a \emph{periodic} disorder $\eta$ determined from the random field $h$. Thus, studying harmonic minimal surfaces in periodic disorder may be insightful for the study of the random-field XY model.

One seeks to understand whether the random-field XY model has a magnetized phase at given temperatures and disorder strengths. Imry--Ma~\cite{imry1975random} predicted that a magnetized phase is absent in dimensions $1\le d\le 4$, while such a phase exists at low temperatures and weak disorder in dimensions $d\ge 5$.
Aizenman--Wehr~\cite{aizenman1989rounding,aizenman1990rounding} proved the prediction in dimensions $1\le d\le 4$ (at all non-negative temperatures and positive disorder strengths). A proof that a magnetized phase exists in dimensions $d\ge 5$ is still missing (even at zero temperature); our proof that harmonic minimal surfaces with periodic disorder localize in these dimensions lends support to the prediction.

An important challenge is to determine the rate of magnetization decay in dimensions $2\le d\le 4$ at low temperature and weak disorder: exponential decay is expected in dimension $d=2$ (as proved for the random-field Ising model~\cite{ding2021exponential, aizenman2020exponential}), but it is not clear whether exponential or inverse power-law decay (or some intermediate rate) should occur in dimension $d=3$~\cite{aharony1980infinite, tissier2006unified, gingras1996topological,feldman2001quasi}. Dario--Harel--Peled~\cite{dario2024quantitative} prove inverse power-law upper bounds on the magnetization decay in dimensions $d\le 3$ and an inverse poly-loglog upper bound in dimension $d=4$. The predicted $\sqrt{\log L}$ height fluctuations for harmonic minimal surfaces in periodic disorder in dimension $d=3$ lends support to the possibility of inverse power-law decay in that dimension (c.f.~\cite{giamarchi2009disordered}). Such decay would imply a transition of the \emph{Berezinskii–Kosterlitz–Thouless
type} as the temperature or disorder strength varies and would be of
great interest. Similarly, order $\sqrt{\log\log L}$ height fluctuations for harmonic minimal surfaces in periodic disorder in dimension $d=4$ lends support to the possibility of an inverse power of logarithm decay in that dimension~\cite{chitra1999disordered} (see also~\cite[Equation~(53)]{feldman2001quasi}). See~\cite{dario2024quantitative} for further discussion.

Charge-density waves (CDWs) and flux-line lattices, as well as the term Bragg glass, are associated with harmonic minimal surfaces with periodic disorder (see, e.g.,~\cite{emig1999disorder,giamarchi2009disordered}). An interesting connection is proposed between CDWs and loop-erased random walks~\cite{wiese2019field, wiese2022theory}.

\section{Open questions and extensions}\label{sec:open questions and extensions}

This section presents and discusses some of the open problems for minimal surfaces in an ``independent'' random environment such as $\eta^\white$ (Section~\ref{sec:whitenoise}) or $\eta^{\text{Poisson}}$ (see~\eqref{eq:eta Poisson}). Unless otherwise stated, we consider a fixed environment strength $\lambda>0$ (not varying with the domain size).

\subsection{Improving the fluctuation estimates}\label{sec:improving fluctuation estimates} Our paper provides rigorous upper and lower bounds on the transversal and ground energy fluctuations of minimal surfaces in an ``independent'' random environment. It is an important challenge to improve upon the obtained expressions. We highlight several specific instances of this challenge:
\begin{itemize}
    \item A case of special interest is to obtain rigorous proofs of the prediction that the exponents $\chi=1/3$ and $\xi=2/3$ when $d=n=1$.
    \item An outstanding issue, open also in the physics literature (see~\cite[equation (783) and discussion in Section 7.11]{wiese2022theory}), is to determine whether for $d=1$ and some $n_0\ge 2$, the ground energy fluctuations are of order $1$ (equivalently by the scaling relation, whether the height fluctuations are of order~$\sqrt{L}$).
    \item Another important problem is to verify that the transversal fluctuations \emph{decrease} with the number of components $n$, until some critical $n_0$ (possibly infinity) after which their order does not change, as the physics literature seems to suggest; our current upper bounds do not depend on $n$.

    A related question is to determine the limit of $\xi$ as $n\to\infty$ in dimensions $1\le d\le 3$.
    \item It would be interesting to verify the physics prediction (\cite{emig1998roughening} and ~\cite[following (79)]{emig1999disorder}) that the transversal fluctuations have poly-logarithmic order in the critical dimension $d=4$ (this entails improving our $(\log\log L)^{\frac{1}{4+n}}$ lower bound to such order).
    \item One would like to combine our delocalization results (Theorem~\ref{thm:delocalization intro}) with the scaling relation in order to deduce lower bounds on the ground energy fluctuations. We achieve this for $d=1$, and achieve it up to a logarithmic loss for $d=2$, but have not obtained the corresponding result for $d=3$ (see Corollary~\ref{cor:lower bounds on energy fluctuations in d=1,2}, Table~\ref{table:energy fluctuations n=1} and~\eqref{eq:best lower bounds on energy fluctuations}). The reason is that the version of the scaling inequality $\chi\ge2\xi+d-2$ proved in~\eqref{eq:ge direction of scaling relation} is in terms of the \emph{average} height, which is not compatible with our delocalization results. Overcoming this issue would yield, for instance, $\chi\ge \frac{7}{5}$ for $d=3$, $n=1$.
\end{itemize}

\subsection{Interpretations of the scaling relation}\label{sec:scaling relation interpretations}
A possible interpretation of the scaling relation is that it relates the standard deviation of the ground energy to the expected absolute height in the center of the domain $\Lambda_L$ (or at other bulk points). In our notation, this means
\begin{equation}\label{eq:naive scaling relation}
    \std({\GE^{\eta,\lambda,\Lambda_L}})\approx \left(\E\|\varphi^{\eta,\lambda,\Lambda_L}_0\|\right)^2 L^{d-2}.
\end{equation}
However, our results for $d=1$ in Section~\ref{sec:1d scaling relation} show that in order to relate the standard deviation of the ground energy to the height fluctuations one has to further take into account the way that the height fluctuations decrease as one approaches the boundary of the domain. For $d=1$ in the most standard cases (e.g., $\eta=\eta^\white$) we still expect~\eqref{eq:naive scaling relation} to be correct as we expect that the height fluctuations in the bulk of the domain yield the most significant contribution to the ground energy fluctuation. Nevertheless, it is not clear to us whether this should always be expected in higher dimensions $d$. Specifically, in dimension $d=2$ it seems possible that the transversal fluctuations decrease to order $1$ (uniformly in $L$) as the number of components $n\to\infty$, which would imply the same for the ground energy fluctuations if~\eqref{eq:naive scaling relation} held. However, we know that $\std({\GE^{\eta,\lambda,\Lambda_L}})\ge \sqrt{L}$ for all $n$ due to the boundary contribution in dimension $d=2$ (see~\eqref{eq:boundary contribution to energy fluctuations}), which would then contradict~\eqref{eq:naive scaling relation}.

An alternative possible interpretation of the scaling relation, which would avoid the boundary contribution to the ground energy fluctuation, is that
\begin{equation}\label{eq:alternative scaling relation}
    \E|{\GE^{\eta,\lambda,\Lambda_L}}-\GE^{\eta[\Lambda_{\lfloor L/2\rfloor}],\lambda,\Lambda_{L}}|\approx \left(\E\|\varphi^{\eta,\lambda,\Lambda_L}_0\|\right)^2 L^{d-2}.
\end{equation}
as we have in our version of the scaling inequality $\chi\le 2\xi+d-2$ in Theorem~\ref{theorem:other half of scaling relation} (recalling that by $\eta[\Delta]$ we mean resampling $\eta$ on $\Delta\times\R^n$). With this interpretation in mind, it is reasonable to try and bound the left-hand side of~\eqref{eq:alternative scaling relation}. Our lower bound based on the boundary contribution in~\eqref{eq:boundary contribution to energy fluctuations} does not adapt to this task, while we expect our lower bounds resulting from the scaling relation for $d=1,2$ (Corollary~\ref{cor:lower bounds on energy fluctuations in d=1,2}) and from the bulk localization bounds~\eqref{eq:bulk contribution to energy fluctuations} to continue to hold for this quantity. An interesting prediction which the form~\eqref{eq:alternative scaling relation} raises is that its left-hand side is at least of order $L$ in dimension $d=3$ for all $n$. We do not have a proof of this statement (the lower bound of this form presented in~\eqref{eq:best lower bounds on energy fluctuations} is based on the boundary contribution).

We have already pointed out after Theorem~\ref{theorem:half scaling relation} that ``pointwise versions'' of the scaling relation, such as~\eqref{eq:naive scaling relation} or~\eqref{eq:alternative scaling relation}, cannot hold in dimensions $d\ge 4$. It is also natural to wonder whether a scaling relation such as~\eqref{eq:naive scaling relation} or~\eqref{eq:alternative scaling relation} will hold in dimensions $d\ge 4$ if one replaces the pointwise term $\|\varphi^{\eta,\lambda,\Lambda_L}_0\|$ by the average term $\|\frac{1}{|\Lambda_L|}\sum_{v\in\Lambda_L} \varphi^{\eta,\lambda,\Lambda_{L}}_v\|$. Theorem~\ref{theorem:half scaling relation} provides the direction $\chi\ge2\xi+d-2$ of such a relation and holds in all dimensions. It may be that both directions of such a relation will hold in dimension $d=4$, however, in dimensions $d\ge 5$, there is doubt that the other direction of the relation will hold, as it would predict that the average height is of order $L^{\frac{4-d}{4}}$ (at least for $\eta^\white$, as $\chi= d/2$ for $d\ge 5$ by~\ref{as:conc} and~\eqref{eq:best lower bounds on energy fluctuations}) whereas a tentative connection with the membrane model, discussed in Section~\ref{sec:scaling limit and rate of correlation decay} below, suggests a smaller order. Altogether, it is possible that the scaling relation does not make sense in dimensions $d\ge 5$ (beyond the one-sided inequalities of Theorem~\ref{theorem:half scaling relation} and Theorem~\ref{theorem:other half of scaling relation}).

\subsection{Scaling limit and rate of correlation decay}\label{sec:scaling limit and rate of correlation decay}

A fundamental open problem is to identify the scaling limit of harmonic minimal surfaces. For certain integrable models of planar last-passage percolation, the scaling limit is known to be the directed landscape of Dauvergne, Ortmann and Vir{\'a}g \cite{dauvergne2022directed, dauvergne2021scaling}; one may thus expect it to also be the scaling limit of harmonic minimal surfaces with independent disorder (Section~\ref{sec:independent disorder}) in dimensions $d=n=1$. Another case where the scaling limit is known is for harmonic minimal surfaces with \emph{linear disorder} (Section~\ref{sec:linear disorder}): These equal in distribution the lattice membrane model (with specific boundary conditions), which has the continuum membrane model as its scaling limit~\cite[Section 7.2]{dario2023random},\cite{caravenna2009scaling, cipriani2019scaling}. By analogy with this case, one may speculate that the scaling limit of harmonic minimal surfaces will be a random continuous function in dimensions $d\le 3$, but only a random distribution in dimensions $d\ge 4$. 

Recall from Section~\ref{sec:linear disorder} that when the disorder strength $\lambda$ is sufficiently small, harmonic minimal surfaces in disorder which is smooth at short scales and independent and identically distributed between vertices (such as $\eta^\white$) will behave similarly to harmonic minimal surfaces in linear disorder. That is, similarly to the membrane model. Moreover, it may be that $\lambda$ may tend to $0$ at a relatively slow rate in the localized dimensions $d\ge 5$.

A related question is to study the rate of correlation decay. For instance, in the localized dimensions $d\ge 5$: (i) Is it true that the covariance between the heights of the minimal surface at vertices $x$ and $y$ exhibits power-law decay in the distance between $x$ and $y$? (ii) What is the typical order of the average of the minimal surface on $\Lambda_L$? Theorem~\ref{theorem:half scaling relation} and assumption~\ref{as:conc} (which holds for $\eta^\white$, say) show that it is at most of order $L^{\frac{4-d}{4}}$. Is the correct rate the one used in the scaling of the membrane model~\cite{cipriani2019scaling}, $L^{\frac{4-d}{2}}$? (iii) Does this average converge, after a suitable rescaling, to a Gaussian distribution?

\subsection{The dependence on the disorder strength}\label{sec:dependence on the disorder strength}

It is natural to ask how the environment strength $\lambda$ affects the behavior of the minimal surface. We conjecture that the transversal fluctuations on $\Lambda_L = \{-L,\ldots,L\}^d$ have the following order in dimensions $1\le d\le 3$ for the disorder $\eta = \eta^\white$:
\begin{equation}\label{eq:lambda dependence transversal prediction}
    \begin{cases}
        \lambda L^{\frac{4-d}{2}}&\lambda\le \lambda_0\\
        \lambda^\nu L^\xi&\lambda_0\le \lambda\le 1
    \end{cases}
\end{equation}
where $\xi$ stands for the transversal fluctuation exponent in dimension $d$ and
\begin{equation}\label{eq:lambda_0 and nu predictions}
    \lambda_0=\frac{1}{L^{\frac{4-d}{2}}}\quad\text{and}\quad\nu = \frac{2}{4-d}\xi.
\end{equation}
Let us reason for this conjecture. First, for sufficiently small $\lambda$, as discussed in Section~\ref{sec:linear disorder}, the disorder is well approximated by a linear disorder, whence the transversal fluctuations of the minimal surface are of order $\lambda L^{\frac{4-d}{2}}$. We believe that this behavior persists until the transversal fluctuations reach order $1$, whence the approximation by linear disorder necessarily breaks down as the ``width'' of the bump function $b$ defining $\eta^\white$ is of order $1$ (Section~\ref{sec:whitenoise}). This explains the regime $\lambda\le \lambda_0$ in~\eqref{eq:lambda dependence transversal prediction} and the value of $\lambda_0$ in~\eqref{eq:lambda_0 and nu predictions}.

Second, in the regime $\lambda_0\le \lambda\le 1$ the discrete gradients of the minimal surface do not exceed order 1, by Lemma~\ref{lem:holder} and~\ref{as:conc}, while the fluctuations of the entire minimal surface do exceed this order. Our conjecture says that the dependence on $\lambda$ in this regime is a simple multiplicative power law $\lambda^\nu$. The value of $\nu$ in~\eqref{eq:lambda_0 and nu predictions} then follows by equating the two expressions in~\eqref{eq:lambda dependence transversal prediction} at $\lambda=\lambda_0$ and using the predicted value for $\lambda_0$.

We also note an alternative renormalization heuristic for the scaling $\lambda^\nu L^\xi$ when $\lambda_0\le \lambda\le 1$: Write $\ell_0(\lambda) = \lambda^{-\frac{2}{4-d}}$ for the length scale at which the transversal fluctuations are of order $1$ (by the $\lambda\le\lambda_0$ regime of~\eqref{eq:lambda dependence transversal prediction}. This length scale is called the Larkin length, following its introduction by Larkin~\cite{larkin1970effect} and Larkin--Ovchinnikov~\cite{larkin1979pinning}). Then $\lambda^\nu L^\xi = (L / \ell_0(\lambda))^\xi$. In other words, the transversal fluctuations of the minimal surface at scale $L$ with the given $\lambda$, behave as the transversal fluctuations of the minimal surface at scale $\frac{L}{\ell_0(\lambda)}$ with $\lambda = 1$.

In the setting of~\eqref{eq:lambda dependence transversal prediction}, we further conjecture the ground energy fluctuations to be of order
\begin{equation}\label{eq:lambda dependence energy prediction}
    \begin{cases}
        \lambda L^{\frac{d}{2}}&\lambda\le \lambda_0\\
        \lambda^{2\nu}L^{\chi}&\lambda_0\le \lambda\le 1
    \end{cases}
\end{equation}
for the ground energy fluctuation exponent $\chi$ in dimension $d$. For $\lambda\le \lambda_0$ this follows from~\ref{as:conc} and Theorem~\ref{thm:lower bounds on energy fluctuations by localization} if the transversal fluctuations in this regime do not exceed $1$ (as in~\eqref{eq:lambda dependence transversal prediction}). For $\lambda_0\le \lambda\le 1$ this follows either from~\eqref{eq:lambda dependence transversal prediction} assuming that the scaling relation holds (consistently with our Theorem~\ref{theorem:half scaling relation} and Theorem~\ref{theorem:other half of scaling relation}), or from a renormalization heuristic like the above, stating that the energy fluctuations are of order $e(\ell_0) (L / \ell_0(\lambda))^{\chi}$ with $e(\ell_0) = \lambda \ell_0(\lambda)^{d/2}$ being the ground energy fluctuations of the minimal surface on $\Lambda_{\ell_0(\lambda)}$ at the given $\lambda$ (by the $\lambda\le \lambda_0$ regime of~\eqref{eq:lambda dependence energy prediction}).

We note that~\eqref{eq:lambda dependence energy prediction} furnishes the prediction $f(\lambda) = \lambda^{2\nu}$ for the function $f(\lambda)$ appearing in~\eqref{eq:chi conc}, in the regime $\lambda_0\le \lambda\le 1$ (at least for $\ell_0(\lambda)\le \ell\le L$).

\smallskip
We point out that the above conjectures,~\eqref{eq:lambda dependence transversal prediction} and~\eqref{eq:lambda dependence energy prediction}, depend on the fact that $\eta^\white$ admits a Taylor approximation and may thus be approximated by linear disorder on short scales. Other behavior may arise in other cases, e.g., for an Ornstein-Uhlenbeck disorder (which might be approximable by Brownian disorder on short scales; see Section~\ref{sec:Brownian disorder}).

The disorder $\eta^{\text{Poisson}}$ (see~\eqref{eq:eta Poisson}) is especially interesting as a scale-covariant behavior is obtained, in the following sense: As $\eta^{\text{Poisson}}_{v,t}\in\{0,\infty\}$, the disorder strength $\lambda$ has no effect in~\eqref{eq:formal Hamiltonian}. Instead, the natural parameter is a scaling factor for the underlying Poisson processes (equivalently, their intensity). For $\alpha>0$, define the rescaled disorder $\eta^{\text{Poisson},\alpha}_{v,t}:=\eta^{\text{Poisson}}_{v,\alpha t}$ for all $v,t$ (so that $\eta^{\text{Poisson},\alpha}$ has the distribution of~\eqref{eq:eta Poisson}, but based on Poisson processes of intensity $\alpha^n$). Then $H^{\eta^{\text{Poisson},\alpha},\lambda,\Lambda}(\varphi) = \alpha^{-2} H^{\eta^{\text{Poisson}},\lambda,\Lambda}(\alpha \varphi)$ (using that $\eta^{\text{Poisson}}_{v,t}\in\{0,\infty\}$). Thus we have the following scale-covariance property: Passing from $\eta^{\text{Poisson}}$ to $\eta^{\text{Poisson},\alpha}$ changes the minimal surface $\varphi$ to $\alpha^{-1}\varphi$ and the ground energy $\GE$ to $\alpha^{-2}\GE$.

\subsection{Coalescence of minimal surfaces}\label{sec:coalescence questions}

In models of first- and last-passage percolation (with discrete heights) one studies the phenomenon of \emph{coalescence of geodesics}. This means that geodesics with nearby starting and ending points will typically overlap almost entirely, differing only in short segments near their endpoints. The existing literature on coalescence covers mainly the case of \emph{planar} first-passage percolation \cite{newman1995surface,licea1996geodesics,damron2014busemann, ahlberghoffman, dembin2022coalescence} and the exactly-solvable models of planar last-passage percolation \cite{pimentel2016duality,basu2019coalescence,zhang2020optimal,seppalainen2020coalescence,balazs2021local}. It is not clear whether geodesics will also coalesce in higher-dimensional ($n\ge 2$) first-passage percolation.

The minimal surface model~\eqref{eq:formal Hamiltonian} with $d=n=1$ is believed to be in the same universality class as planar first-passage percolation. However, since the minimal surface takes values in the real numbers, the question of coalescence is more subtle. In fact, it may be that for the disorder $\eta^\white$, geodesics will come arbitrarily close to one another but will not coalesce (see also \cite{bakhtin2016inviscid}). Some support for this possibility is given by the fact that coalescence does not occur for linear disorder, which approximates the behavior of $\eta^\white$ on short scales (see Section~\ref{sec:linear disorder}). Nevertheless, coalescence may be expected for the point-process based disorder $\eta^{\text{Poisson}}$ (see~\eqref{eq:eta Poisson}); moreover, for $d=n=1$ this may be provable via the methods of \cite{dembin2022coalescence}.

A natural question is whether minimal \emph{surfaces} ($d\ge 2$) and $n=1$ will feature coalescence, e.g., for the disorder $\eta^{\text{Poisson}}$, in the sense that two minimal surfaces on a large domain with boundary values differing by a constant (say) will typically equal each other on most vertices. As minimal surfaces with $n=1$ maintain their order (i.e., surfaces with higher boundary values are higher throughout) we expect coalescence to occur in the delocalized dimensions ($d\le 4$). The behavior for $n\ge 2$ is again less clear. 

\subsection{Pointwise delocalization}\label{sec:pointwise delocalization}
Theorem~\ref{thm:delocalization intro} proves that minimal surfaces in dimensions $d\le 4$ have a large norm on a \emph{uniformly positive fraction of the vertices} of the domain, with uniformly positive probability. It would be interesting to establish the stronger fact that, when $d\le 4$, the surfaces have a large norm at \emph{specific vertices} of the domain: For instance, to prove that
\begin{equation}\label{eq:bks}
 \lim _{L\to \infty }\E\big(\|\varphi^{\eta,\lambda,\Lambda_L}_v\|\big) = \infty     
\end{equation}
for a fixed vertex $v$ (say, at the midpoint $v=0$). 

In first-passage percolation, this issue was highlighted by Benjamini--Kalai--Schramm~\cite{BKS} and is often referred to as the BKS midpoint problem; it also received consideration earlier: for instance, it is closely related to the discussion in (9.22) of Kesten's lectures~\cite{Kesten:StFlour}. Once again, in first-passage percolation, the planar case is better understood \cite{damron2017bigeodesics,ahlberghoffman,dembin2022coalescence}, while in higher dimensions partial results are available~\cite{dembin2023ell}, as well as strong results under assumptions which are still unverified~\cite{alexander2020geodesics}. 

Coalescence of minimal surfaces (Section~\ref{sec:coalescence questions}) implies pointwise delocalization (e.g., as in~\cite{dembin2022coalescence}), but is not necessary for its proof (e.g., as in~\cite{alexander2020geodesics}).

\subsection{Infinite minimal surfaces} It is natural to ask whether an infinite-volume limit of the minimal surfaces may be taken, to obtain a surface on an infinite domain which minimizes the Hamiltonian~\ref{eq:formal Hamiltonian} in the sense that its energy may not be lowered by modifying it on any finite subdomain. For $d=1$, such infinite geodesics have been heavily investigated~\cite[Section 4.4]{50years},~\cite{ahlberghoffman, alexander2020geodesics}.

One possibility is to consider surfaces defined on the semi-infinite domain $\Lambda_\infty^{\ge 0} := \{0,1,\ldots\}^d$. The \emph{existence} of such semi-infinite minimal surfaces, for $\eta=\eta^\white$ say, in all dimensions $d\ge 1$ and codimensions $n\ge 1$ is a consequence of our Theorem~\ref{thm:localization}. 
One way to demonstrate this is as follows, akin to the \emph{metastate} approach of Aizenman-Wehr~\cite{aizenman1989rounding,aizenman1990rounding}: Set $\Lambda_L^{\ge 0} := \{0,1,\ldots,L\}^d$. For a given $L$, consider the joint distribution of $(\varphi^{\eta,\lambda,\Lambda_L^{\ge 0}}, \eta)$. Estimate~\eqref{eq:localization estimate} implies that it converges in distribution along a \emph{subsequence} of $L$. The limit is the distribution of a pair $(\varphi^{\tilde{\eta},\lambda,\Lambda_\infty^{\ge 0}}, \tilde{\eta})$ with $\tilde{\eta}\eqd \eta$ and $\varphi^{\tilde{\eta},\lambda,\Lambda_\infty^{\ge 0}}$ a minimal surface on $\Lambda_\infty^{\ge 0}$ in the disorder $\tilde{\eta}$ (but note that conditioned on $\tilde{\eta}$, there may still be additional randomness in $\varphi^{\tilde{\eta},\lambda,\Lambda_\infty^{\ge 0}}$). The minimal surface $\varphi^{\tilde{\eta},\lambda,\Lambda_\infty^{\ge 0}}$ continues to satisfy the same bound as in~\eqref{eq:localization estimate} (with $r_v:=\min_{w\in\Z^d\setminus\Lambda_\infty^{\ge 0}}\|v-w\|_\infty$). The construction gives a (random) minimal surface on $\Lambda_\infty^{\ge 0}$ with zero boundary conditions, but it may be applied with other boundary conditions using the harmonic extension property of Corollary~\ref{cor:effect of boundary conditions}.

Among the questions of interest:
\begin{enumerate}
    \item Does  $\lim_{L\to\infty}\varphi^{\eta,\lambda,\Lambda_L^{\ge 0}}_v$ exist for every $v\in\Lambda_\infty^{\ge 0}$? \item Is there a \emph{unique} minimal surface $\varphi$ on $\Lambda_\infty^{\ge 0}$ with zero boundary conditions satisfying that $\limsup_{v\in\Lambda_\infty^{\ge 0}}\|\varphi_v\|/\|v\| = 0$? If yes, will this continue to hold when $\|v\|$ is replaced by $\|v\|^\alpha$ for some $\alpha>1$?
\end{enumerate}
These questions are related to the coalescence of surfaces (Section~\ref{sec:coalescence questions}). Variants may be asked for other semi-infinite domains, such as the half-space $\Lambda^{\text{HS}}:=\{0,1,\ldots\}\times\Z^{d-1}$.
\begin{enumerate}
\setcounter{enumi}{2}
    \item   
    Are there minimal surfaces on the infinite domain $\Lambda_\infty:=\Z^d$?
\end{enumerate}
For $d=1$, this is analogous to the famous question of the existence of \emph{bigeodesics} in first-passage percolation, originating from Furstenberg; see~\cite[(9.22)]{Kesten:StFlour}. It is conjectured that no bigeodesics exist when $n=1$, while it is not clear whether bigeodesics exist for $n\ge 2$ (but see Alexander~\cite{alexander2020geodesics} for results under certain assumptions).

To obtain such minimal surfaces, one may try and take a (subsequential) limit as $L\to\infty$ of $\varphi^{\eta,\lambda,\Lambda_L}$, the minimal surface on $\Lambda_L$ with zero boundary conditions. This can be done when the distribution of $\varphi^{\eta,\lambda,\Lambda_L}_v$ is tight as $L\to\infty$, for every $v\in\Lambda_\infty$, which is the case in dimensions $d\ge 5$ by Theorem~\ref{thm:localization}. However, as pointwise delocalization (Section~\ref{sec:pointwise delocalization}) is expected, this approach should fail in dimensions $d\le 4$.

Notwithstanding, we conjecture that minimal surfaces on $\Lambda_\infty$ do exist in all dimensions $d\ge 2$ and codimensions $n\ge 1$. This is due to the relatively small transversal fluctuation exponent predicted in these dimensions (Table~\ref{table:transversal fluctuations n=1} for $n=1$ and the fact that the exponent is expected to decrease with $n$); see~\cite[Section 4.5.1]{50years} for a related heuristic of Newman for the $d=1$ case, and~\cite[Section 9]{bassan2023non} for related discussions for minimal surfaces with overhangs.

\subsection{Regularity of the ground energy and height distributions}\label{sec:regularity of the distributions}

It is natural to expect the distributions of the ground energy $\GE^{\eta,\lambda,\Lambda}$ and height $\varphi^{\eta,\lambda,\Lambda}_v$ (or the average/maximum height) to exhibit some regularity. There are several meanings for regularity: One important notion is that of fast tail decay, in the sense that the distribution exhibits (stretched) exponential tail decay around its expectation, on the scale of its standard deviation. For the distribution of $\varphi^{\eta,\lambda,\Lambda}_v$, we may further ask whether it has a density, whether the density is approximately decreasing in $\|\varphi^{\eta,\lambda,\Lambda}_v\|$ and whether it is approximately rotationally symmetric (rotational symmetry follows when the distribution of the disorder is rotationally invariant, but may also hold approximately in other cases). It seems natural to expect such regularity, for instance, for the disorders $\eta^{\white}$ and $\eta^{\text{Poisson}}$. Certain regularity properties may be transferred between the ground energy and height distributions via the scaling relations proved in Section~\ref{sec:scaling relation}.

Exponential tail decay on the scale of the standard deviation, as well as existence of a unimodular density, would follow from the stronger property of log-concavity. Log-concavity of the passage time (an analogue of the ground energy) was very recently established for certain integrable models of last-passage percolation, and deduced also for the Tracy-Widom distribution which 
arises in their scaling limit~\cite{baslingker2024log}. It is natural to wonder whether the ground energy and height distributions of harmonic minimal surfaces are also log concave (say, for the disorder $\eta^\white$ with a suitable bump function).

Recently, Basu, Sidoravicius and Sly proved stretched exponential tail decay, on the scale of the standard deviation, for the passage times of a class of rotationally-invariant planar first-passage percolation models~\cite{basu2023rotationally} (under assumptions to be verified in a sequel paper). It would be interesting to try and adapt their methods to harmonic minimal surfaces.

\subsection{More general geometries and elastic energy operators}
The methods developed in this paper to analyze finite-volume minimizers of the Hamiltonian~\eqref{eq:formal Hamiltonian} may be applied in a more general setting in which the elastic energy term is replaced by a general positive quadratic form.

\subsubsection{Setting}
Let $V$ be a finite or infinite set. Surfaces (with $n$ components) are modeled by functions in the vector space
\begin{equation}
    \Omega_0:=\{\varphi:V\to\R^n\colon \varphi_v = 0\text{ for all but finitely many $v\in V$}\},
\end{equation}
equipped with the scalar product $(\varphi,\psi):=\sum_{v\in V}\varphi_v\cdot\psi_v$. Let $A$ be a symmetric positive-definite linear operator acting on $\Omega_0$; i.e., $A$ is a linear operator from $\Omega_0$ to itself satisfying $(\varphi,A\psi)=(A\varphi,\psi)$ for $\varphi,\psi\in\Omega_0$ and $(\varphi, A\varphi)>0$ for all $\varphi\in\Omega_0$ except $\varphi\equiv 0$.

The formal Hamiltonian of the model on $\Omega_0$ is given by
\begin{equation}\label{eq:generalized formal Hamiltonian}
    H^{A,\eta,\lambda}(\varphi):=\frac{1}{2}(\varphi, A\varphi) + \lambda \sum_{v\in V} \eta_{v,\varphi_v},
\end{equation}
where $\eta:V\times\R^n\to(-\infty,\infty]$ is the environment and $\lambda>0$ is the environment strength.

Given a finite $\Lambda\subset V$ and a surface $\tau\in\Omega_0$, consider the \emph{minimal surface} of the model in $\Lambda$ with boundary value $\tau$, i.e., the surface $\varphi$ in
\begin{equation}
   \Omega_0^{\Lambda,\tau} := \{\varphi\in\Omega_0\colon \varphi_v=\tau_v\text{ for $v\in V\setminus\Lambda$}\}
\end{equation}
which minimizes $H^{A,\eta,\lambda}$ (a minimizer is defined using the fact that $H^{A,\eta,\lambda}(\varphi) - H^{A,\eta,\lambda}(\psi)$ is well defined when $\varphi,\psi\in\Omega_0^{\Lambda,\tau}$); we assume that a minimizer exists and choose one suitably if it is not unique. The goal is to study the geometry of these minimal surfaces and their associated (suitably defined) minimal energy. For the environment $\eta$ we may again take a smoothed white noise as in Section~\ref{sec:whitenoise}, or work under general assumptions as in Section~\ref{sec:disorder assumptions}.

\subsubsection{Examples}
Let us list some specific cases of interest for this model:
\begin{enumerate}
    \item Graph Laplacian and its powers: If $V$ is the vertex set of a connected graph $G=(V,E)$, we may take $A = -\Delta_G$, where $\Delta_G$ is the graph Laplacian defined by $(-\Delta_G\varphi)_v:=\sum_{u\colon \{u,v\}\in E(G)}(\varphi_v-\varphi_u)$. The formal Hamiltonian~\eqref{eq:formal Hamiltonian} studied in this paper is obtained as the special case where $G$ is the lattice $\Z^d$.

    More generally, one may consider powers of the graph Laplacian $A = (-\Delta_G)^\alpha$. For instance, the case $\alpha=2$ gives the bilaplacian operator $(-\Delta_G)^2$, which satisfies the equivalence
    \begin{equation}
    (\varphi, (-\Delta_G)^2 \varphi) = (-\Delta_G \varphi, -\Delta_G \varphi) = \sum_{v\in V} \|(-\Delta_G\varphi)_v\|^2.
    \end{equation}
    We note that the model with Hamiltonian $(\varphi, (-\Delta_G)^2 \varphi)$ (in the absence of disorder) is called the membrane model (see, e.g.,~\cite{cipriani2019scaling}).

    \item Adding a mass: starting from a symmetric positive-definite linear operator $A_0$ acting on $\Omega_0$, one may set $A := A_0 + m I$, where $m>0$ (termed the mass) and $I$ is the identity operator, leading to
    \begin{equation}
    (\varphi, (A_0+ m I) \varphi) = (\varphi, A_0 \varphi) + m \sum_{v\in V} \|\varphi_v\|^2.
    \end{equation}
    The study~\cite{ben2024landscape} considers the model~\eqref{eq:formal Hamiltonian} with the addition of a mass term.

    \item General coupling constants and hierarchical models: Generalizing the graph Laplacian operator, we may set $A = -\Delta_J$, with $(-\Delta_J\varphi)_v:=\sum_{u\in V}J_{uv}(\varphi_v-\varphi_u)$, for some $J:V\times V\to[0,\infty)$ (the coupling constants) satisfying $J_{uv} = J_{vu}$ for all $u,v\in V$. We note the discrete Green identity
    \begin{equation}
    (\varphi, -\Delta_J\varphi) = \sum_{\{u,v\}\colon u,v\in V} J_{uv}\|\varphi_u - \varphi_v\|^2.
    \end{equation}
    This operator is positive-definite when $\{\{u,v\}\colon J_{uv}>0\}$ is the set of edges of a \emph{connected} graph.

    This choice allows, for instance, to consider hierarchical versions of the lattice $\Z^d$ (see, e.g., \cite{baker1972ising,bleher1987critical} for background and use of hierarchical lattices in the context of the Ising model and the recent~\cite{hutchcroft2022critical} for the context of percolation).
\end{enumerate}

\subsubsection{Generalized main identity}
Importantly, the generalized Hamiltonian~\eqref{eq:generalized formal Hamiltonian} still satisfies a version of our main identity (Proposition~\ref{prop:main identity}). Indeed, the same proof shows that for each environment $\eta$ and environment strength $\lambda$ and each $\varphi,s\in\Omega_0$ we have
\begin{equation}\label{eq:new main identity}
    H^{A,\eta^s,\lambda}(\varphi+s) - H^{A,\eta,\lambda}(\varphi) = (\varphi,A s) + \frac{1}{2}(s,As),
\end{equation}
where $\eta^s$ is defined in~\eqref{eq:eta s def}, leading to the cancellation of the disorder terms on the left-hand side. We also note that, for the disorder choice $\eta=\eta^\white$ (Section~\ref{sec:whitenoise}), the natural analogue of our concentration assumption~\ref{as:conc} remains in effect (with the same proof).

\subsubsection{Scaling relation and threshold dimension}
Using the main identity, one may try to extend the techniques of our paper to various choices of the operator $A$. As an example, on the lattice $\Z^d$, consider $A = (-\Delta_{\Z^d})^\alpha$ for $\alpha>0$. What will be the analog of the scaling relation in this case?

We would like to estimate the minimal ``elastic energy'' required for a surface with zero boundary conditions outside $\Lambda_L=\{-L,-L+1,\ldots, L\}^d$ to reach height $h>0$. One way to quantify it is by
\begin{equation}
    \min\left\{\frac{1}{2}(s, As)\colon s\equiv 0\text{ outside $\Lambda_L$}, (s,s)=h^2 L^d\right\}.
\end{equation}
The solution is of order $h^2 L^{d-2\alpha}$, using the fact that the minimal eigenvalue of the operator $(-\Delta_{\Z})$ with Dirichlet boundary conditions outside $\{-L,\ldots, L\}$ is of order $L^{-2}$. Introducing the notation of the transversal fluctuation exponent $\xi$, we write $h=L^{\xi}$.

We compare the above required elastic energy with the typical fluctuations of the energy of the minimal surface on $\Lambda_L$, which we denote by $L^\chi$. The comparison yields the prediction for the scaling relation
\begin{equation}\label{eq:scaling relation for power of Laplacian}
    \chi = 2\xi + d - 2\alpha.
\end{equation}

For $\eta=\eta^\white$ and $\lambda=1$, say, the (natural analogue of our) concentration assumption~\ref{as:conc} implies that $\chi\le \frac{d}{2}$. As this estimate becomes tight when the minimal surface has transversal fluctuations of order $1$, we may define the threshold dimension
\begin{equation}
    d_c = 4\alpha
\end{equation}
as the dimension where the scaling relation~\eqref{eq:scaling relation for power of Laplacian} holds with $\chi = \frac{d}{2}$ and $\xi=0$. Similarly to the results of this work, one may predict that the minimal surface delocalizes with power-law transversal fluctuations (i.e., $\xi>0$) for $d<d_c$, delocalizes with sub-power-law transversal fluctuations for $d=d_c$ and localizes for $d>d_c$. The methods of this work may adapt to prove such statements, using the generalized main identity~\eqref{eq:new main identity}. 

Quantitatively, our arguments may adapt to give the following conclusions for $d<d_c$ (and for $\eta=\eta^\white$ and $\lambda=1$, say).
On the localization side, the bound $
\chi\le d/2$ may combine with the scaling relation~\eqref{eq:scaling relation for power of Laplacian} to yield
\begin{equation}
    \xi\le \alpha - \frac{d}{4}.
\end{equation}
On the delocalization side, the methods used to prove Theorem~\ref{thm:lower bounds on energy fluctuations by localization} should imply that the energy fluctuations of the minimal surface are at least of order $\sqrt{L^d/h^n}$ when the transversal fluctuations are bounded by some $h\ge1$, as in~\eqref{eq:lower bound on ground energy fluctuations via average disorder} (the perturbation $f_1$ in \eqref{eq:B} is used for this purpose). This may combine with the scaling relation~\eqref{eq:scaling relation for power of Laplacian} to yield the ``Flory-type estimate''
\begin{equation}
    \xi \ge \frac{4\alpha - d}{4+n}.
\end{equation}
Additionally, if the contribution of the disorder terms to the energy of the minimal surface (i.e., $\lambda \sum_{v\in V} \eta_{v,\varphi_v}$) is linear in the volume of the domain (as shown in Corollary~\ref{lem:nabla} for our setup), we expect our methods to show that the energy fluctuations of the minimal surface are at least of order $1$ (the perturbation $f_2$ in \eqref{eq:B} is used for this purpose). This may combine with the scaling relation~\eqref{eq:scaling relation for power of Laplacian} to yield
\begin{equation}
    \xi\ge \alpha-\frac{d}{2}\quad\text{when the disorder energy grows linearly in the volume of the domain}.
\end{equation}

A related discussion for harmonic minimal surfaces with $A = (-\Delta_{\Z^d})^2$ and \emph{linear disorder} appears in~\cite[Section 7.7]{dario2023random}, where it is demonstrated that the critical dimension is indeed $d_c=8$ and exact quantitative fluctuation bounds are derived. Another related discussion for the random-field XY model with interactions having ``higher-order continuous symmetries'' is in~\cite[Section 9]{dario2024quantitative}.

\subsection{Integer-valued harmonic minimal surfaces and roughening transition}

To more closely mimic the interfaces of disordered spin systems, one may consider \emph{integer-valued} harmonic minimal surfaces, i.e., the surfaces which minimize~\eqref{eq:formal Hamiltonian} among $\varphi:\Z^d\to\Z^n$ (with specified boundary values). We expect these to exhibit the same features described for the domain walls of the disordered Ising ferromagnet in Section~\ref{sec:disordered Ising ferromagnet}. Briefly: Delocalization at all (fixed) disorder strengths $\lambda$ in dimensions $d=1,2$, a \emph{roughening transition} from a localized to a delocalized regime  as the disorder strength increases for $d=3$, and localization at all (fixed) disorder strengths in dimensions $d\ge 4$. These assertions remain open, though it may be possible to adapt the methods 
of \cite{licea1996superdiffusivity,bovier1996there,dembin2023ell} to prove delocalization (with non-optimal quantitative dependence) in dimensions $d=1,2$ and the methods of \cite{bovier1994rigorous,dario2023random, bassan2023non} to prove localization at weak disorder in dimensions $d\ge 3$. Strikingly, while one may expect integer-valued harmonic minimal surfaces to fluctuate less than real-valued harmonic minimal surfaces, it is not even known, in any dimension $d\ge1$, that the height fluctuations for integer-valued harmonic minimal surfaces on $\Lambda_L$ are $o(L)$ as $L\to\infty$.

Proposition~\ref{prop:main identity} (the main identity) seems of limited use in the integer-valued case as it seems one is restricted to use it with integer-valued shift functions $s$. On a positive note, this still suffices for Corollary~\ref{cor:effect of boundary conditions} (effect of boundary conditions) for domains $\Lambda$ and boundary conditions $\tau$ for which the harmonic extension $\overline{\tau}^\Lambda$ is integer valued. In particular, we obtain partial information on the ``limit shape'' for $d=n=1$: As in~\eqref{eq:d=1 limit shape}, we conclude that the time constant function~\eqref{eq:time constant function} (adapted to the integer-valued setup, and when it is well defined) satisfies $\mu(x) - \mu(x-1) = \frac{1}{2}(x^2 - (x-1)^2)$ for $x\in\R$ under the integer-valued analogues of~\ref{as:exiuni}+\ref{as:stat}.

\subsection*{Acknowledgements}
We are grateful to Michael Aizenman, Antonio Auffinger, Roland Bauerschmidt, Gerard Ben Arous, Itai Benjamini, Paul Bourgade, Paul Dario, Duncan Dauvergne, Pierre Le Doussal, Ronald Fisch, Yan Fyodorov, Thierry Giamarchi, David Huse, Shirshendu Ganguly, Benjamin McKenna, Asaf Nachmias, Felix Otto, Jeremy Quastel, Timo Sepp\"al\"ainen, Thomas Spencer, Allan Sly, B\'alint Vir\'ag, Christian Wagner, Lingfu Zhang and Nikos Zygouras for stimulating discussions. We especially thank David Huse for very helpful discussions around the predictions of Section~\ref{sec:dependence on the disorder strength}. We also give special thanks to Itai Benjamini for raising questions on minimal surfaces in random environment, for providing encouragement and for involvement in the early stages of this project.

The research of B.D. is partially funded by the SNF Grant 175505 and the ERC Starting Grant CriSP (grant agreement No 851565) and is part of NCCR SwissMAP. The research of R.P. is partially supported by the Israel Science Foundation grants
1971/19 and 2340/23, and by the European Research Council Consolidator grant 101002733 (Transitions).

Part of this work was completed while R.P. was a Cynthia and Robert Hillas Founders' Circle Member of the Institute for Advanced Study and a visiting fellow at the Mathematics Department of Princeton University. R.P. is grateful for their support.
\appendix

\section{Green's function estimates}\label{appendix:greenfunction}

In this section we study the Green's function $G_\Lambda ^v$ and establish bounds that are used in Section~\ref{sec:loc}. Recall that for a subset $\Lambda \subseteq \mathbb Z ^d$ and $v\in \Lambda $ the Green's function $G_\Lambda ^v:\mathbb Z ^d \to \mathbb R$ is given by 
\[ G_\Lambda^v(x):=\frac{1}{2d} \cdot \E_x\left[ \big| \big\{ t\in [0, \tau _{\Lambda }] : X_t=v\big\} \big| \right],\]
where $X_t$ is a simple, discrete time random walk starting from $x$ and where $\tau _\Lambda $ is the first exit time of $X$ from $\Lambda$.

Throughout this section we use the following result on the full space Green's function. This result is taken from \cite[Theorems~4.4.8, 4.4.4, 4.3.1]{lawler2010random}. Recall that the potential kernel $a=a_d$ is given by 
    \begin{equation}
        a(x)=\lim _{N\to \infty }\sum _{n=0}^{N }\big( \mathbb P ^0(X_n=0)-\mathbb P ^0(X_n=x) \big)
    \end{equation}
where $X_n$ is a simple random walk on $\Z^ d$ starting at $0$. In dimensions $d\ge 3$, we have that $a(x)=G_{\mathbb Z ^d}^0(0)-G_{\mathbb Z ^d}^0(x)$.
\begin{thm}\label{thm:Lawler}
    The potential kernel has the following asymptotic behavior as $\|x\|\to \infty $:
    \begin{enumerate}
        \item 
        If $d=1$ then 
        \begin{equation*}
            a(x)=|x|.
        \end{equation*}
        \item 
        If $d=2$ then there are $C_2,C_2'>0$ such that
        \begin{equation*}
            a(x)=C_2\log (\|x\|+1 )+C_2'+O(\|x\|^{-2}).
        \end{equation*}
        \item 
        If $d\ge 3$ then there are $C_d,C_d'>0$ such that
        \begin{equation*}
            a(x)=C_d\|x\|^{2-d}+C_d'+O(\|x\|^{-d}).
        \end{equation*}
    \end{enumerate}
\end{thm}

We will also need the following standard Gambler’s ruin estimates (see for example \cite[Propositions~5.1.1,  5.1.5]{lawler2010random}).

\begin{claim}\label{claim:GR}
    Let $X_t$ be a one dimensional discrete or continuous simple random walk on $\mathbb Z$ starting from $0$. For any $n\in \mathbb Z$ let $\tau _n$ be the first hitting time of $n$. 
    \begin{enumerate}
        \item For any $m,n>0$ we have that $\mathbb P (\tau _n\le \tau _{-m}) =m/(n+m)$.
        \item 
        For any $n>0$ and time $t>0$ we have that $\mathbb P (\tau _n \ge t ) \le Cn/\sqrt{t}$.
    \end{enumerate}
\end{claim}

We turn to prove Lemma~\ref{lem:88} and Lemma~\ref{lem:89}.

\begin{proof}[Proof of Lemma~\ref{lem:88} and Lemma~\ref{lem:89}]
    Let $P$ be the hyperplane determined by the face of $\Lambda _L$ that is closest to $v$ (with arbitrary choice if there is more than one such face) and let $H$ be the half space corresponding to $P$ that contains $\Lambda _L$. Note that $d(v,H^c)=d(v,\Lambda ^c)=r_v$. Moreover, we have that  $G_L ^v(v)\le G_H^v(v)$ and therefore it suffices to bound $G_H^v(v)$. Let $v'$ be the reflection of $v$ with respect to $P$. That is, $v'$ is the unique vertex such that $\|v-v'\|=2r_v$ and $d(v',H)=r_v+1$. We claim that $G_H^v(v)=a(v'-v)$. Indeed, by symmetry, for any $y\in P=H^+\setminus H$ we have that $\mathbb P ^y(X_n=v)=\mathbb P ^y(X_n=v')$ and therefore using translation invariance in the first equality and the strong Markov property in the last equality
\begin{equation*}
\begin{split}
    a(v'-v)&=\lim _{N\to \infty }\sum _{n=1}^{N} \mathbb P ^v (X_n=v)-\mathbb P ^v (X_n=v') \\
    &=G_H^v(v) +\lim _{N\to \infty }\sum _{n=1}^{N} \mathbb P ^v (X_n=v, \tau <n)-\mathbb P ^v (X_n=v', \tau <n) =G_H^v(v),
\end{split}
\end{equation*}
where $\tau $ is the first exit time from $H$ (hence $X_\tau\in P$).

Using the first part of Theorem~\ref{thm:Lawler} we obtain that $G_H^v(v)= a(v'-v)=|v'-v|=2r_v$ in dimension $d=1$. Similarly, by the second part of Theorem~\ref{thm:Lawler} we have $|G_H^v(v)|\le C\log (1+r_v)$ in dimension $d=2$. This finishes the proof of Lemma~\ref{lem:89} and the first part of Lemma~\ref{lem:88}.

We turn to prove the second part of Lemma~\ref{lem:88} ($d=1$). Let $X_n$ be a random walk starting from $u$ and let $\tau '$ be the first time $X_n$ either hits $v$ or exits $\Lambda _L$. By the strong Markov property we have that $G_L ^v(u)=\mathbb P^u (X_{\tau '}=v)G_L^v (v)$ and therefore $G_L^v(v)-G_L^v(u)\ge 0$. Moreover, using the first part of Lemma~\ref{lem:88} we obtain 
    \begin{equation*}
        G_L ^v(v)-G_L ^v(u) =\mathbb P ^u(X_{\tau '}\neq v) G_L^v(v) \le Cr_v \mathbb P^u (X_{\tau '}\neq v) \le Cr_v \frac{|u-v|}{r_u+|u-v|} \le C|u-v|,
    \end{equation*}
    where in the second to last inequality we used Claim~\ref{claim:GR} and in the last inequality that $r_v\le r_u+|u-v|$.
\end{proof}

We turn to prove Lemma~\ref{lem:Green}. Throughout the rest of the appendix, $\Lambda $ is an arbitrary box of the form $\Lambda :=[a_1,b_1]\times \cdot \times [a_d,b_d] \subseteq \mathbb Z^d$. We start with the following claim.

\begin{claim}\label{claim:1}
    For any $v,x\in \Lambda $ with $r_v\le 16\|x-v\|$ we have that $G_\Lambda ^v(x)\le Cr_v \|x-v\|^{1-d}$.
\end{claim}

\begin{proof}
    As in the proof of Lemma~\ref{lem:88} and Lemma~\ref{lem:89} we let $H$ be a half space that contains $\Lambda $ and determined by the face of $\Lambda $ that is closest to $v$. We have that $G_\Lambda ^v(x)\le G_H^v(x)$ and
\begin{equation*}
\begin{split}
    a(v'-x)&-a(v-x)=\lim _{N\to \infty }\sum _{n=1}^{N} \mathbb P ^x (X_n=v)-\mathbb P ^x (X_n=v') \\
    &=G_H^v(x) +\lim _{N\to \infty }\sum _{n=1}^{N} \mathbb P ^x (X_n=v, \tau <n)-\mathbb P ^x (X_n=v', \tau <n) =G_H^v(x),
\end{split}
\end{equation*}
where $\tau $ is the first exit time from $H$ and $v'$ is the reflection of $v$ through the hyperplane corresponding to $H$. We used in the last inequality, the strong Markov property and the symmetry.

Finally, it is straightforward to check that by Theorem~\ref{thm:Lawler} and a Taylor expansion, in any dimension $d\ge 1$ we have 
\begin{equation*}
 G_H^v(x) =a(v'-x)-a(v-x)\le C\|v-v'\|\cdot \|v-x\|^{1-d} \le Cr_v\|v-x\|^{1-d}. 
\end{equation*}
This finishes the proof of the claim.
\end{proof}

Before we proceed with the proof of Lemma~\ref{lem:Green}, let us note that in the definition of the Green's function $G_\Lambda ^v$, one can use a continuous time random walk instead of a discrete time walk. More precisely, we have that 
\[G_\Lambda^v(x)=\frac{1}{2d} \cdot \E_x\left[ \big| \big\{ t\in [0, \tau _{\Lambda }] : X_t=v\big\} \big| \right],\]
where in here $|\cdot |$ stands for the one dimensional Lebesgue measure and where $X_t$ is a continuous time random walk with unit jump rate. As before, $\tau _\Lambda $ is the first exit time from $\Lambda $. Indeed, this equality follows as the expected waiting time of $X_t$ in a vertex is $1$. 

\begin{cor}\label{cor:1}
   For any $v,x\in \Lambda $ with $r_v\le 8\|x-v\|$ we have that $G_\Lambda ^v(x)\le Cr_xr_v\|x-v\|^{-d}$.
\end{cor}

\begin{proof}
    If $r_x\ge \|x-v\|/2$ then the bound follows immediately from Claim~\ref{claim:1}. Thus, we may assume that $r_x\le \|x-v\|/2$. Let $X_t$ be a continuous time random walk starting from $x$ and define the stopping time 
\begin{equation*}
    \tau ':=\inf \{t>0: X_t\notin \Lambda \text{ or }\|X_t-x\|_\infty=\lfloor \|x-v\|/2 \rfloor \}.
\end{equation*}
By the strong Markov property we have 
\begin{equation*}
    G^v_\Lambda (x)= \mathbb E [G_\Lambda ^v(X_{\tau '})] \le Cr_v\|x-v\|^{1-d} \cdot \mathbb P (X_{\tau '}\in \Lambda ),
\end{equation*}
where in the last inequality we used Claim~\ref{claim:1} and the fact that $\|X_{\tau '}-v\|\ge \|x-v\|/2\ge r_v/16$. Thus, it suffices to prove that 
\begin{equation}\label{eq:RW estimate}
 \mathbb P (X_{\tau '} \in \Lambda )\le Cr_x/\|v-x\|.
\end{equation}
This estimate is somewhat standard and some of the details are omitted. First, assume without loss of generality that the face closest to $x$ is in the $e_1$ direction. That is, assume that $x+r_xe_1\notin \Lambda $. Next, for any coordinate $i\le d$, define the stopping time 
\begin{equation}
    \tau _i :=\inf \big\{ t>0:|(X_t-x)_i| =\lfloor \|x-v\|/2 \rfloor   \big\}.
\end{equation}
It suffices to prove that for all $i\le d$ we have that $\mathbb P (\tau _i=\tau ') \le Cr_x/\|x-v\|$ since $ \mathbb P (X_{\tau '}\in \Lambda )\le \sum_{i=1}^ d\mathbb P (\tau _i=\tau ') $. When $i=1$, this probability is bounded by the probability that a one dimensional random walk starting from $0$ will hit $-\lfloor \|v-x\|/2\rfloor $ before hitting $r_x$ which is at most $Cr_x/\|v-x\|$ by Claim~\ref{claim:GR}. When $i>1$ we have 
\begin{equation*}
\begin{split}
    \mathbb P (\tau _i=\tau ')&\le \mathbb E \big[ \mathbb P\big( \forall t<\tau _i, \ (X_t-x)_1<r_x  \mid \{(X_t)_i\}_{t>0} \big) \big] \\
    &\le \mathbb E \big[ \min (1, Cr_x/\sqrt{\tau _i})  \big] \le \int _0^1 \mathbb P \big( \tau _i \le Cr_x^2/t^2 \big) dt \le Cr_x/\|v-x\| ,
\end{split}
\end{equation*}
where in the second inequality we used Claim~\ref{claim:GR} and the fact that the different coordinates of $X_t$ are independent. The last inequality follows as the probability inside the integral decays rapidly for $t \gg r_x/\|v-x\|$. This finishes the proof of \eqref{eq:RW estimate} and of the corollary.
\end{proof}

\begin{lem}\label{lem:G diff}
    For any $x,u,v\in \Lambda $ with $r_v\le 4\|x-v\|$ we have that $$|G_\Lambda ^v(x)-G_\Lambda ^u(x)|\le C\|u-v\|\cdot \|x-v\|^{1-d}.$$
\end{lem}

\begin{proof}
    First, we may assume that $\|u-v\|\le \min (r_x,r_v)/10$ as otherwise the bound follows from Corollary~\ref{cor:1}. Second, note that it suffices to prove the lemma when $u$ and $v$ differ only in one coordinate. Indeed, in the general case, we can write the difference $G_\Lambda ^v(x)-G_\Lambda ^u(x)$ as a sum of $d$ differences corresponding to pairs of vertices that differ only in one coordinate. Without loss of generality, suppose that $u$ and $v$ differ only in the first coordinate. By similar arguments, we may further assume that $u_1-v_1$ is odd and that $v_1>u_1$ (an even number is the sum of two smaller odd numbers).
    
Throughout this proof we will use the notation $X(t)$ instead of $X_t$ to denote the location of the walk at time $t$. The notation $X_i$ will be used to denote the $i$th coordinate of the walk.

Next, we construct a coupling of continuous time random walks $X,X'$ starting at $x$ in such a way that $X(t)-X'(t)$ becomes equal to $v-u$ in a short amount of time. Let $X(t):=(X_1(t),\dots ,X_d(t))$ be a continuous time random walk starting at $x$ and let $\bar{X}_1(t)$ be an independent one dimensional walk with rate $1/(2d)$ starting at $x_1$. Define the stopping time  $\tau _1 :=\inf \{t>0: X_1(t)\neq \bar X_1(t)\}$ and the walk 
\begin{equation}
    \tilde{X}_1(t)=\begin{cases}
        \bar{X}_1(t) \quad &t\le \tau _1\\
        \bar{X}_1(\tau _1)-(X_1(t)-X_1(\tau _1 )) \quad & t \ge \tau _1
    \end{cases}.
\end{equation}
Next, define the stopping time $\tau _2:=\inf \{ t>\tau _1: X_1(t)-\tilde{X}_1(t)=v_1-u_1\}$
and the walk 
\begin{equation}
    X_1'(t)=\begin{cases}
        \tilde{X}_1(t) \quad &t\le \tau _2\\
        \tilde{X}_1(\tau _2 ) +X_1(t)-X_1(\tau _2) \quad &t\ge \tau _2
    \end{cases}.
\end{equation}
Finally, we let $X'(t):=(X_1'(t),X_2(t),\dots ,X_d(t))$. By construction, for all $t>\tau _2 $ we have that $X(t)-X'(t)=v-u$. The idea here is to first couple the parity of the difference $X_1'(t)-X_1(t)$ to make it odd and then let $X_1'(t)$ take minus the steps taken by $X_1(t)$ until they have the right difference. Once they have the right difference, they have the same jumps.

Let $\tau ,\tau '$ be the first exit times of $X$ and $X'$ respectively from $\Lambda $. For any two (possibly random) times $t_1,t_2$ define
\begin{equation*}
    T[t_1,t_2]:=\big| \big\{ t\in [t_1,t_2] : X(t)=v\big\} \big| , \quad T'[t_1,t_2]:=\big| \big\{t\in [t_1,t_2] : X'(t)=u \big\} \big|,
\end{equation*}
with the convention that $T[t_1,t_2]=T'[t_1,t_2]=0$ if $t_1\ge t_2$.
Note that $G_\Lambda ^v (x)=\mathbb E [T(0,\tau )]$ and $G_\Lambda ^u (x)=\mathbb E [T'(0,\tau ')]$.

 Since $X(t)-X'(t)=v-u$ for all $t>\tau _2$ we have almost surely that $T(\tau _2,\tau \wedge \tau ')=T'(\tau _2,\tau \wedge \tau ')$.

 Define the stopping time $\tau _3:=\inf \{t: \|X(t)-x\| \ge r_x/5\}$ and note that $\tau _3\le \tau \wedge \tau '$ almost surely and that $X$ and $X'$ cannot reach $u$ and $v$ by time $\tau _3$ since $\|x-v\|+r_v\ge r_x $ and hence $\|x-v\|\ge \frac{r_x}2$; moreover $\|u-x\|\ge \|x-v\|-\|u-v\|\ge \frac 2 {5} r_x$. Define the events 
 \begin{equation}
     \mathcal A :=\{\tau _3\le \tau _2\},\quad \mathcal B:= \{ \tau _2 \le \tau _3\le   \tau '\le \tau  \} \quad \text{and} \quad  \mathcal C := \{ \tau _2 \le \tau _3 \le 
     \tau \le \tau ' \}
 \end{equation}
and note that at least one of them occur. Our goal is to bound the expectation of the random variable $Z:=|T(0,\tau )-T'(0,\tau ')|$ separately on each one of the events $\mathcal A, \mathcal B$ and $\mathcal C$. By the strong Markov property we have that
\begin{equation}\label{eq:jfv}
    \mathbb E [\mathds 1 _{\mathcal A } Z] \le \mathbb E \big[ \mathds 1 _{\mathcal A } \big( G_\Lambda ^v(X(\tau _3)) +G_\Lambda ^u(X(\tau _3)) \big) \big] \le C r_x r_v \|x-v\|^{-d}\mathbb P (A),
\end{equation}
where in the second inequality we used Corollary~\ref{cor:1} and that $\|X(\tau _3)-v\|\ge \|x-v\|-r_x/5\ge \frac 35 \|x-v\|\ge \frac 118 r_v$ (where we used that $\|x-v\|\ge r_v/4$) and $d(X(\tau _3),\Lambda ^c)\le 2r_x$. To bound the probability of $\mathcal A$, note that $\tau_2 \le \hat{\tau } _2$ almost surely where $\hat \tau _2 :=\inf \{t>0: X_1(t) >(u_1-v_1)/2\}$ and by the same arguments as in the proof of \eqref{eq:RW estimate} we have that $\mathbb P (\mathcal A )\le \mathbb P (\tau _3 \le \hat{\tau } _2) \le C \|u-v\|/r_x$. Substituting this bound into \eqref{eq:jfv} we obtain that $\mathbb E [\mathds 1 _{\mathcal A } Z] \le \|u-v\| r_v \|x-v\|^{-d} \le C\|u-v\| \|x-v\|^{1-d}$.

We turn to bound $\mathbb E[\mathds 1 _\mathcal B Z]$. Note that on the event $\mathcal B$ we have that $Z=T(\tau ',\tau )$. Thus, by the strong Markov property and the fact that $\mathcal B$ is measurable by time $\tau '$ we have
\begin{equation}\label{eq:fvbd}
\begin{split}
    \mathbb E &\big[ \mathds 1 _\mathcal B Z \big] \le \mathbb E [\mathds 1 _\mathcal B G_\Lambda ^v (X(\tau '))] \\
    &\le C\|u-v\|\|x-v\|^{1-d} +C\sum _{j=\lfloor \log _2 r_v \rfloor }^{\log_2 \|x-v\| -3} \!\! \mathbb P \big( \mathcal B,  \|X(\tau ')-v\| \in [2^j,2^{j+1}] \big) \|u-v\| r_v 2^{-jd},
\end{split}
\end{equation}
 where in the last inequality we bound $G_\Lambda ^v (X(\tau '))$ separately, depending on the scale of $\|X(\tau ')-v\|$ or whether it is larger than $\|x-v\|/20$. In all these cases we use Corollary~\ref{cor:1} to bound $G_\Lambda ^v (X(\tau '))$ and the fact that on $\mathcal B$ we have that $d(X(\tau '),\Lambda ^c) \le \|u-v\|$. To bound the probabilities on the right hand side of \eqref{eq:fvbd}, note that for any $w\in \Lambda ^+\setminus \Lambda $ we have that $\mathbb P (X'(\tau ')=w) \le G_{\bar{\Lambda }} ^w(x)$ where $\bar{\Lambda }$ is the box containing all vertices at $\ell _\infty $ distance at most $1$ from $\Lambda $. Thus, by Claim~\ref{claim:1}, if $\|x-w\| \ge \|x-v\|/2$ then $\mathbb P (X'(\tau ')=w) \le C\|x-v\|^{1-d}$. Thus, we have that 
 \begin{equation*}
     \mathbb P \big( \mathcal B,  \|X(\tau ')-v\| \in [2^j,2^{j+1}] \big) \le \mathbb P \big( \|X'(\tau ')-v\| \le 2^{j+2} \big) \le C2^{j(d-1)}\|x-v\|^{1-d}, 
 \end{equation*}
 where in the last inequality we used that there are at most $C2^{j(d-1)}$ vertices on the boundary of $\Lambda $ at distance at most $2^{j+2}$ from $v$. Substituting this bound into \eqref{eq:fvbd} we get $\mathbb E [ \mathds 1 _\mathcal B Z ] \le C\|u-v\|\|x-v\|^{1-d}$, as needed. The bound on $\mathbb E [ \mathds 1 _\mathcal C Z ]$ is identical. This finishes the proof of the lemma.
\end{proof}

The following corollary follows from Lemma~\ref{lem:G diff} in the same way that Corollary~\ref{cor:1} follows from Claim~\ref{claim:1}.

\begin{cor}\label{cor:2}
    For all $x,v,u\in \Lambda $ with $ r_v\le 2\|x-v\| $ we have that 
    $$|G_\Lambda ^v (x)-G_\Lambda ^u(x)| \le Cr_x\|u-v\| \cdot \|x-v\|^{-d}.$$
\end{cor}

Lemma~\ref{lem:Green} follows from Corollary~\ref{cor:1} and Corollary~\ref{cor:2}.

\bibliographystyle{plain}
\bibliography{MSRE}

\end{document}